\documentclass[a4paper,USenglish]{lipics-v2021}
\pdfoutput=1
\usepackage[hideChapterNumber,refchapterCite,noCCBy]{craigstyle}
\usepackage{plabels}
\usepackage{xspace}
\usepackage{fdsymbol} %
\usepackage[usenames,dvipsnames]{xcolor}

\usepackage{placeins}

\usepackage{tikz}
\usetikzlibrary{shapes.geometric}
\usetikzlibrary{shapes.symbols}
\usetikzlibrary{positioning}
\usetikzlibrary{arrows}

\usepackage{environ}
\NewEnviron{draft}{\BODY}
\RenewEnviron{draft}{\color{Sepia} \BODY}
\RenewEnviron{draft}{}

\newcommand{\name}[1]{\emph{#1}}
\newcommand{\defname}[1]{\emph{#1}}
\newcommand{\f}[1]{\mathsf{#1}}

\newcommand{\true}{\top}
\newcommand{\false}{\bot}
\newcommand{\imp}{\rightarrow}

\newcommand{\equi}{\leftrightarrow}

\newcommand{\eqdef}{\;
\raisebox{-0.1ex}[0mm]{$ \stackrel{\raisebox{-0.2ex}{\tiny
\textnormal{def}}}{=} $}\; }
\newcommand{\eqdefnospace}{%
  \raisebox{-0.1ex}[0mm]{$ \stackrel{\raisebox{-0.2ex}{\tiny\textnormal{def}}}{=}$}}

\newcommand{\entails}{\models}

\newcommand{\la}{\langle}
\newcommand{\ra}{\rangle}

\newcommand{\du}[1]{\overline{#1}}
\newcommand{\duh}[1]{\overline{#1}{}}

\newcommand{\Vampire}{\textit{Vampire}\xspace}
\newcommand{\OTTER}{\textit{OTTER}\xspace}

\newcommand{\EProver}{\textit{E}\xspace}
\newcommand{\ProverN}{\textit{Prover9}\xspace}
\newcommand{\leanCoP}{\textit{leanCoP}\xspace}

\newcommand{\CMProver}{\textit{CMProver}\xspace}
\newcommand{\SETHEO}{\textit{SETHEO}\xspace}
\newcommand{\METEOR}{\textit{METEOR}\xspace}

\newcommand{\HKPYM}{HKPYM\xspace}

\newcommand{\LK}{\textbf{LK}\xspace}

\newcommand{\algoinput}{\smallskip \noindent\textsc{Input: }}
\newcommand{\algooutput}{\smallskip \noindent\textsc{Output: }}
\newcommand{\algomethod}{\smallskip \noindent\textsc{Method: }}

\newcommand{\OPTIONS}{\textsc{Options}}

\newcommand{\strictsubterm}[2]{#1 \lhd #2}
\newcommand{\invsubst}[2]{#1 \la #2^{-1} \ra}
\newcommand{\invsubstpost}[1]{\la #1^{-1} \ra}

\newcommand{\pred}[1]{\m{\mathcal{P}\hspace{-0.11em}red}^\pm(#1)}

\newcommand{\predplain}[1]{\m{\mathcal{P}\hspace{-0.11em}red}(#1)}
\newcommand{\fun}[1]{\m{\mathcal{F}\hspace{-0.18em}un}(#1)}
\newcommand{\lit}[1]{\m{\mathcal{L}\hspace{-0.01em}iterals}(#1)}
\newcommand{\complit}[1]{\overline{\m{\mathcal{L}\hspace{-0.01em}iterals}(#1)}}

\newcommand{\m}[1]{\mathit{#1}}

\newcommand{\dom}[1]{\m{\mathcal{D}\hspace{-0.08em}om}(#1)}
\newcommand{\rng}[1]{\m{\mathcal{R}\hspace{-0.02em}ng}(#1)}

\newcommand{\varfun}{\m{\mathcal{V}\hspace{-0.11em}ar}}

\newcommand{\var}[1]{\varfun(#1)}

\newcommand{\voc}[1]{\m{\mathcal{V}\hspace{-0.11em}oc}^\pm{(#1)}}
\newcommand{\vocplain}[1]{\m{\mathcal{V}\hspace{-0.11em}oc}{(#1)}}

\newcommand{\vocfun}{\m{\mathcal{V}\hspace{-0.11em}oc}^\pm}
\newcommand{\vocplainfun}{\m{\mathcal{V}\hspace{-0.11em}oc}}

\newcommand{\FL}{F}
\newcommand{\GR}{G}
\newcommand{\lnotGR}{\lnot \GR}

\newcommand{\aaa}{\f{F}}
\newcommand{\bbb}{\f{G}}

\newcommand{\nlit}[1]{\f{lit}(#1)}
\newcommand{\nclause}[1]{\f{clause}(#1)}

\newcommand{\nside}[1]{\f{side}(#1)}
\newcommand{\ntgt}[1]{\f{tgt}(#1)}
\newcommand{\nipol}[1]{\f{ipol}(#1)}
\newcommand{\npath}[2]{\f{path}_{#1}(#2)}
\newcommand{\ncopy}[1]{\f{copy}(N)}

\newcommand{\nipolfun}{\f{ipol}}

\newcommand{\npathL}[1]{\npath{\aaa}{#1}}
\newcommand{\npathR}[1]{\npath{\bbb}{#1}}

\newcommand{\tup}[1]{\boldsymbol{#1}}
\newcommand{\tts}{\tup{t}}

\newcommand{\sss}{\tup{s}}
\newcommand{\sffs}{\tup{f}}
\newcommand{\sggs}{\tup{g}}

\newcommand{\tuple}[1]{\boldsymbol{#1}}
\newcommand{\xs}{\tuple{x}}
\newcommand{\ys}{\tuple{y}}
\newcommand{\ts}{\tuple{t}}

\newcommand{\us}{\tuple{u}}
\newcommand{\vs}{\tuple{v}}

\newcommand{\sided}{two-sided\xspace}
\newcommand{\Sided}{Two-Sided\xspace}

\newcommand{\fgs}{\symset{F\!G}}
\newcommand{\ffs}{\symset{F}}
\newcommand{\ggs}{\symset{G}}

\newcommand{\FE}{F_{\textsc{exp}}}

\newcommand{\GE}{G_{\textsc{exp}}}
\newcommand{\HG}{H_{\textsc{grd}}}

\newcommand{\sth}{\eta}       %
\newcommand{\stt}{\sigma}     %

\newcounter{ibcounter}

\newcommand{\ibn}[1]{\prl{#1}}
\newcommand{\ibm}[1]{(\prefNumber{#1}^\prime)}

\newcommand{\tightop}[1]{\hspace{1pt}{#1}\hspace{1pt}}

\newcommand{\sterm}{\text{-term}}
\newcommand{\sterms}{\text{-terms}}

\newcommand{\vbar}{\raisebox{-0.60ex}{\rule{0pt}{2.35ex}}}

\definecolor{tcolbbbbg}{rgb}{0.8,0.8,0.8}
\definecolor{tcolaaabg}{rgb}{1.0,1.0,1.0}
\newcommand{\taaa}[1]{\colorbox{tcolaaabg}{$#1\vbar$}}
\newcommand{\tbbb}[1]{\colorbox{tcolbbbbg}{$#1\vbar$}}

\renewcommand{\taaa}[1]{#1{}^\aaa}
\renewcommand{\tbbb}[1]{#1{}^\bbb}

\newcommand{\nhphantom}[1]{\sbox0{#1}\hspace*{-\the\wd0}}
\newcommand{\nannot}[1]%
           {\hspace{0.2em}{[}#1{]}\nhphantom{\hspace{0.2em}{[}#1{]}}}

\newcommand{\fa}{\f{a}}
\newcommand{\fb}{\f{b}}
\newcommand{\fc}{\f{c}}

\newcommand{\ff}{\f{f}}
\newcommand{\fg}{\f{g}}
\newcommand{\fh}{\f{h}}
\newcommand{\fp}{\f{p}}
\newcommand{\fq}{\f{q}}
\newcommand{\fr}{\f{r}}
\newcommand{\fs}{\f{s}}

\newcommand{\FONLY}{$F$-only\xspace}
\newcommand{\GONLY}{$G$-only\xspace}
\newcommand{\FGSHARED}{$FG$-shared\xspace}

\newcommand{\rname}[1]{\textup{\textsc{#1}}}
\newcommand{\axname}[1]{\textsc{#1}}

\theoremstyle{definition}
\newtheorem{algo}[theorem]{Algorithm}

\newcommand{\FAIL}{\f{FAIL}}
\newcommand{\lfp}[3]{\text{\upshape{\textsc{Lfp}}}_{#1,#2}#3}
\newcommand{\gfp}[3]{\text{\upshape{\textsc{Gfp}}}_{#1,#2}#3}
\newcommand{\LFP}{\textsc{Lfp}}
\newcommand{\GFP}{\textsc{Gfp}}

\newcommand{\crc}{\textsf{C}\xspace}
\newcommand{\csaturated}{\textsf{C}-saturated\xspace}
\newcommand{\cderivation}{\textsf{C}-derivation\xspace}
\newcommand{\cclosed}{\textsf{C}-closed\xspace}
\newcommand{\cresolution}{\textsf{C}-resolution\xspace}
\newcommand{\cfactor}{\textsf{C}-factor\xspace}
\newcommand{\cfactoring}{\textsf{C}-factoring\xspace}
\newcommand{\cresolvent}{\textsf{C}-resolvent\xspace}

\newcommand{\FF}{\aaa}
\newcommand{\GG}{\bbb}
\newcommand{\ripol}{\f{ripol}}
\newcommand{\rsubclause}[2]{\f{subclause}_{#1}(#2)}

\newcommand{\ICALC}{\mathcal{I}}
\newcommand{\RCALC}{\mathcal{R}}
\newcommand{\GRCALC}{\mathcal{R}_{\textsc{grd}}}
\newcommand{\DED}{\mathcal{D_{\textsc{grd}}}}
\newcommand{\DEDR}{\mathcal{D}}

\newcommand{\FGRD}{F_{\textsc{grd}}}
\newcommand{\GGRD}{G_{\textsc{grd}}}

\newcommand{\prov}{\mathcal{A}}
\newcommand{\prova}{\mathcal{A}}
\newcommand{\provb}{\mathcal{B}}
\newcommand{\xside}{\mathcal{A}}

\newcommand{\const}[1]{\m{\mathcal{C}\hspace{-0.11em}onst}(#1)}

\newcommand{\rewrite}{\Rightarrow}

\newcommand{\cutleaf}[1]{#1}

\newcommand{\ME}{M_{\textsc{exp}}}
\newcommand{\NE}{N_{\textsc{exp}}}

\newcommand{\symseq}[1]{\boldsymbol{#1}}
\renewcommand{\us}{\symseq{u}}
\renewcommand{\vs}{\symseq{v}}
\renewcommand{\fgs}{\symseq{f\!g}}
\renewcommand{\ffs}{\symseq{f}}
\renewcommand{\ggs}{\symseq{g}}

\newcolumntype{A}{>{$}l<{$}}

\newcolumntype{L}[1]{>{\raggedright\let\newline\\\arraybackslash\hspace{0pt}}p{#1}}

\newcommand\subsumedBy{\mathrel{\ooalign{$\geq$\cr
      \hidewidth\raise.225ex\hbox{$\cdot\mkern7.0mu$}\cr}}}
\newcommand\subsumes{\mathrel{\ooalign{$\leq$\cr
      \hidewidth\raise.225ex\hbox{$\cdot\mkern2.0mu$}\cr}}}
\newcommand\strictlySubsumes{\mathrel{\ooalign{$<$\cr
      \hidewidth\raise.0ex\hbox{$\cdot\mkern2.0mu$}\cr}}}
\newcommand\strictlySubsumedBy{\mathrel{\ooalign{$>$\cr
       \hidewidth\raise.0ex\hbox{$\cdot\mkern7.0mu$}\cr}}}
\newcommand\variant{\mathrel{\ooalign{$=$\cr
      \hidewidth\raise.7ex\hbox{$\cdot\mkern4.5mu$}\cr}}}

\newcommand{\eref}[1]{\prefGlobalNumber{equiv:#1}}

\title{\vspace{26pt}Interpolation with Automated First-Order Reasoning}
\titlerunning{Interpolation with Automated First-Order Reasoning}

\author{Christoph Wernhard}
       {University of Potsdam, Germany}
       {info@christophwernhard.com}
       {https://orcid.org/0000-0002-0438-8829}
       {}

\begin{document}

\maketitle
\noindent Draft July 6, 2026%

\begin{abstract}
We consider interpolation from the viewpoint of fully automated theorem
proving in first-order logic as a general core technique for mechanized
knowledge processing. For Craig interpolation, our focus is on the two-stage
approach, where first an essentially propositional ground interpolant is
calculated that is then lifted to a quantified first-order formula. We discuss
two possibilities to obtain a ground interpolant from a proof: with clausal
tableaux, and with resolution. Established preprocessing techniques for
first-order proving can also be applied for Craig interpolation if they are
restricted in specific ways. Equality encodings from automated reasoning
justify strengthened variations of Craig interpolation. Contributions to Craig
interpolation that emerged from automated reasoning include variations for
logics used in databases and logic programming. As an approach to uniform
interpolation we introduce second-order quantifier elimination with examples
and describe the basic algorithms DLS and SCAN.
\end{abstract}

\tableofcontents

\section{Introduction}

In this chapter we approach interpolation from the viewpoint of fully
automated theorem proving in first-order logic as a general core technique for
mechanized knowledge processing.
Craig interpolation \cite{craig:linear,craig:uses} fits quite naturally into
this methodology centered around first-order proving:
In first-order logic, if a formula entails another one, then there is a finite
proof of this entailment. A Craig interpolant -- a first-order formula
that is semantically between the entailing formulas and syntactically within
their shared vocabulary -- can be calculated from the proof.

Automated reasoning gives us Craig interpolation in \emph{mechanized} form. A
first-order theorem prover searches for the underlying proof, which is then
converted to an interpolant. This makes Craig interpolation available for
practical applications. Novel variations where interpolants satisfy stronger
syntactic constraints than just shared vocabulary can be explored with
numerous and large problem instances.
For proof search, we can benefit from decades of research in automated
reasoning, manifested in the \name{Handbook of Automated Reasoning}
\cite{handbook:ar:2001} and conference series such as \name{Conference on
  Automated Deduction (CADE)} and \name{International Joint Conference on
  Automated Reasoning (IJCAR)}. We can utilize advanced highly-optimized
systems and the \name{TPTP (Thousands of Problems for Theorem Provers) World}
\cite{tptp}, a research infrastructure that includes a problem library,
specified standard formats, software tools and data from prover evaluations.

Although taming proof search is a core objective of automated reasoning, it is
less relevant for Craig interpolation, where we can start from a given proof,
assuming it had already been found by some powerful system. However, a closer
inspection of the regular winners of the annual \name{CADE ATP System
  Competition (CASC)}, \Vampire \cite{vampire} and \EProver \cite{eprover},
reveals that in powerful configurations they do not output proof objects in some
defined calculus. Applications that require actual proofs, such as hammers
\cite{hammer}, which invoke automated systems on subproblems in an interactive
setting, use a workflow where just lemmas are taken from the powerful systems,
to guide the search with a weaker system that builds proofs. A practical
alternative is provided with \ProverN \cite{prover9}, a fairly strong
first-order prover that can output actual resolution proofs and represents the
state of the art in about 2009.

\Vampire, \EProver and \ProverN operate by maintaining clause sets, which grow
through adding clauses inferred by calculi that can be described as resolution
and equality processing with paramodulation or superposition
\cite{robinson:1965,paramodulation:1969,bachm:ganz:92}.
Also vice-versa, equality inferences by superposition can be taken as basis,
with resolution as a special case. There is a second tradition of fully
automated provers, going back to Prawitz
\cite{prawitz:1960:improved,prawitz:1969:advances}, which operate by
enumerating tableau-like proof structures in combination with unification of
formulas. Model elimination \cite{loveland:1968}, the connection
method \cite{bibel:atp:1987,bibel:otten:2020}, and clausal tableaux
\cite{letz:diss,letz:stenz:handbook,letz:habil} are such methods.
With respect to power for general proof search, this approach currently
resides at the state of the art of around 2000 but for some applications and
investigations it is still well suited, as demonstrated with the \leanCoP
family, e.g.,
\cite{otten:2003:leancopinabstract,kaliszyk:2015:femalecop,satcop:2021,otten:2021:nanocop}
and recent developments \cite{rwzb:lemmas:2023}.
This approach, enumerating proof structures, inherently leads to well-defined
proof objects. Moreover, due to the typical enumeration of structures by
increasing size, proofs tend to be small. \CMProver
\cite{cw:pie:2016,cw:pie:2020}, a system implementing this approach, thus
provides a further practical way to obtain proof objects for Craig
interpolation.

\emph{Uniform} interpolation is in automated reasoning considered since the
1990s as \name{second-order quantifier elimination}. The approach is based on
\emph{equivalence}, computing for a given second-order formula an
\emph{equivalent} first-order formula. It continues a thread that has its
roots -- just like Craig interpolation
\cite{craig:2008:history,craig:2008:road} -- in the \emph{elimination problem}
considered in the \name{algebra of logic}, e.g., by Boole and Schröder. In
early studies of first-order logic, elimination was applied to first-order
formulas with just unary predicates and no function symbols as a decision
procedure by Löwenheim, Skolem and Behmann. Behmann's presentation from 1922
\cite{beh:22} can be viewed as a modern computational method, using
equivalence-preserving formula rewriting until innermost second-order
subformulas have a shape that allows schematic elimination \cite{cw-relmon}.
Ackermann studied the elimination problem on the basis of full first-order
logic in depth and presented in 1935 numerous results
\cite{ackermann:35,ackermann:35:arity}, including: a polarity-related
elimination schema, known today as \name{Ackermann's lemma}, a variation of
resolution, which, for certain formula classes, yields an infinite conjunction
as elimination result, a form of Skolemization for predicates that, in certain
cases, allows reduction to unary predicates, and the negative result that
second-order quantifier elimination on the basis of first-order logic cannot
succeed in general.
Ackermann's results were reinvented and rediscovered in the early 1990s,
leading to two algorithms that expanded into main families of modern
elimination algorithms: First, the SCAN algorithm by Gabbay and Ohlbach
\cite{scan}, which eliminates predicates by a form of resolution. It had been
implemented \cite{scan-system-paper,scan:engel} on the basis of the \OTTER
first-order prover \cite{otter}, the predecessor of \ProverN mentioned above.
The presentation of SCAN \cite{scan} shows applications in different areas of
computational logic and initiated the modern study of second-order quantifier
elimination. Second, the DLS algorithm by Doherty, {\L}ukaszewicz and
Sza{\l}as \cite{dls:early,dls:1997} which rewrites formulas to let Ackermann's
lemma become applicable.
\pagebreak

For the concepts and methods from \emph{automated} proving some distinguishing
aspects can be observed that seem to have their roots in the requirement to
provide a basis for implementation but may also be useful in wider and
abstract contexts. We sketch four of these.

\enlargethispage{6pt}
\vspace{-3pt}
\begin{description}

\item[Certain Forms of Simplicity.] Formulas considered for proof search in
  automated proving are typically in clausal form. Skolemization makes
  explicit quantifier symbols dispensable. Variable instantiation is typically
  driven by most general unifiers (most general substitutions that make two
  terms or atoms identical) of pairs of atoms with the same predicate but in
  literals with complementary polarity.
  Herbrand's theorem is a common tool to justify completeness of calculi and
  to design calculi. In a sense it simplifies the problem of first-order
  proving by reducing unsatisfiability of a set of first-order clauses to
  propositional unsatisfiability of a finite set of clause instances.
  Robinson's resolution \cite{robinson:1965} has just a \emph{a single rule}.
  Gentzen's \LK, aiming to model human reasoning, has 19 rules. As
  traced by Bibel and Otten \cite{bibel:otten:2020}, the connection method
  emerged from \LK in a series of ``compressions'', reducing the number of
  rules, via Schütte's one-sided GS
  \cite{schuette:1956,bibel:atp:1987,troelstra:schwichtenberg:2000} (three
  rules plus cut) until there are no more rules and only a generalization of
  the \name{axiom} property remains that characterizes validity. It
  superimposes a graph structure on the input formula that links
  \name{connections}, certain occurrences of literals with the same predicate
  but complementary polarity. A graph labeling represents the implicit
  involvement of multiple formula copies, which, under the most general
  substitution obtained by unifying atoms in connection instances, form a
  propositionally valid disjunction. Proof search shifts in the connection
  method from exploring possibilities to apply sequent rules to search for a
  certificate of validity based on structures superimposed on the formula.

\item[Some Details Turn into Relevant Spaces of Possibilities.] An
  implementation can reveal details that have effect on specific features of
  an output or on the resources required to obtain it. For Craig
  interpolation, we have, e.g., choices in proof search that may influence the
  size of the found proof and thus the size of the interpolant, and choices
  that have heuristic effect on the duration of the proof search. A subformula
  might have occurrences in each of the two interpolated formulas. For
  interpolant calculation it can be considered as belonging to either one,
  each possibility resulting in a different interpolant.

\item[Approaching Problems and Special Logics with Encodings into Classical
  Logic.] Systems\linebreak of automated reasoning are complex and highly
  optimized, making internals hard to modify. Typically they process formulas
  of some general fundamental class, e.g., classical first-order logic
  (first-order provers) or classical propositional logic (SAT solvers). The
  most immediate approach to solve an application problem or to process some
  special logic is thus to \emph{encode} the problem or the special logic in
  the machine-supported classical logic.

\item[Robustness of Results.] The passage of an abstractly specified method to
  an implementation typically goes along with a strengthening of the
  correctness because overlooked subtle aspects and omissions come to light.
  Claims about the method can be substantiated by experiments that relate it
  to the state of the art. For automated theorem proving this is well
  supported by the \name{TPTP World} with evaluation records from many provers
  and standardized prover interfaces.
\end{description}

\pagebreak
\subparagraph*{Outline of the Chapter.}
After setting up preliminaries (Sect.~\ref{sec-notation}), we address
preprocessing for first-order proving, with specific constraints for Craig
interpolation (Sect.~\ref{sec-preprocessing}). For Craig interpolation our
focus is on the \name{two-stage approach}, where first an essentially
propositional ground interpolant is calculated that is then \name{lifted} to a
quantified first-order formula (Sect.~\ref{sec-lifting}). We discuss two
possibilities to obtain a ground interpolant from a proof, with \name{clausal
  tableaux} (Sect.~\ref{sec-tableaux}), and with \name{resolution}
(Sect.~\ref{sec-resolution}). Equality encodings from automated reasoning
justify\linebreak strengthened variations of Craig interpolation
(Sect.~\ref{sec-equality}). Also further contributions to Craig interpolation
emerged from automated reasoning (Sect.~\ref{sec-contrib}). We introduce
second-order quantifier elimination with examples and describe the basic
algorithms, \name{DLS} and \name{SCAN}~(Sect.~\ref{sec-soqe}).

\section{Notation and Preliminaries}
\label{sec-notation}

Unless specially noted, we consider formulas of first-order logic without
equality (which does not preclude incorporation of equality as an axiomatized
predicate).
The symbol $\entails$ expresses entailment, $\equiv$ equivalence.
A \defname{negation normal form (NNF) formula} is a quantifier-free formula
built up from \defname{literals} (atoms or negated atoms), truth value
constants $\true, \false$, conjunction and disjunction. A \defname{conjunctive
  normal form (CNF) formula} is an NNF formula that is a conjunction of
disjunctions (\defname{clauses}) of literals. A CNF formula is represented by
a \defname{set of clauses}. For clauses and sets of clauses the semantic
notions of entailment and equivalence apply with respect to their universal
closure, i.e., considering their free variables as governed by a universal
quantifier.
The complement of a literal~$L$ is denoted by $\du{L}$. For a set~$F$ of
clauses, the set of all literals in the clauses of $F$ is denoted by
$\lit{F}$, and the set of their complements by $\complit{F}$.

An occurrence of a subformula in a formula has positive (negative)
\defname{polarity}, depending on whether it is in the scope of an even (odd)
number of possibly implicit occurrences of negation.
We call predicate and function symbols briefly \name{predicates} and
\name{functions}. Constants are considered as 0-ary functions. Let $F$ be a
formula. $\var{F}$ is the set of its free individual variables, $\fun{F}$ the
set of functions occurring in it, $\const{F}$ the set of constants among
these, and $\pred{F}$ is the set of all polarity-predicate pairs ${+}p$ and
${-}p$ such that $p$ is the predicate of an atom occurrence in $F$ with the
indicated polarity, positive by $+$ and negative by $-$. For example,
$\pred{(\lnot \fp \lor \fq) \land (\lnot \fq \lor \fr)} = \{{-}\fp, {-}\fq,
{+}\fq, {+}\fr\}$. $\voc{F}$ is $\pred{F} \cup \fun{F} \cup \var{F}$. We also
define $\predplain{F}$ as the set of all predicates that occur in $F$ and
$\vocplain{F}$ as $\predplain{F} \cup \fun{F} \cup \var{F}$.

For second-order formulas~$F$, $\pred{F}$ and $\predplain{F}$ only contain
predicates with \emph{free} occurrences in $F$. We indicate second-order
quantification over predicates with letters $p, q$, e.g, $\exists p\, F$, and
over functions with letters $f, g$. The special case of quantification over a
function that is a constant is just first-order quantification over the
constant considered as individual variable.
A \defname{sentence} is a formula without free variables. A formula is
\defname{ground} if it is quantifier-free and has no free variables.
We can now define the central notion of a \name{Craig-Lyndon} interpolant
as follows.
\begin{definition}[Craig-Lyndon Interpolant]
  Let $F, G$ be formulas such that $F \entails G$. A \defname{Craig-Lyndon
    interpolant} for $F, G$ is a formula $H$ such that $F \entails H$, $H
  \entails G$, and $\voc{H} \subseteq \voc{F} \cap \voc{G}$.
\end{definition}
In the context of Craig-Lyndon interpolation for formulas $F, G$,
we call $\voc{F} \setminus \voc{G}$ the \defname{\FONLY symbols},
$\voc{G} \setminus \voc{F}$ the \defname{\GONLY symbols}, and
$\voc{F} \cap \voc{G}$ the \defname{\FGSHARED symbols}.
The perspective of validating an entailment $F \entails G$ by showing
unsatisfiability of $F \land \lnot G$ is reflected in the notion of
\defname{Craig-Lyndon separator} for $F, \GR$, defined as Craig-Lyndon
interpolant for $F, \lnot \GR$. A \defname{Craig interpolant} is defined
like a Craig-Lyndon interpolant, but using $\vocplainfun$ instead of
$\vocfun$ for the syntactic condition.

We represent a substitution $\sigma$ of variables by terms as a set of
assignments $\{x_1 \mapsto t_1, \ldots, x_n \mapsto t_n\}$. The application of
a substitution $\sigma$ to a formula or term $E$ is written $E\sigma$. We call
$E\sigma$ an \defname{instance} of $E$. If $E$ is an instance of $E'$ and $E'$
is an instance of $E$, we say that both are \defname{variants} of each other.
A substitution is called \defname{ground} if its range is a set of ground
terms.
We use the notation for substitutions also to express the substitution of a
predicate by another one. To express the substitution of a predicate by a
complex formula we generalize it as follows. Let $F$ be a formula, let $p$ be
an $n$-ary predicate, let $\xs = x_1,\ldots,x_n$ be distinct individual
variables, and let $G$ be a formula. Then $F\{p \mapsto \lambda \xs . G\}$
denotes $F$ with all occurrences $p(t_1,\ldots,t_n)$ of~$p$ replaced by the
respective instance $G\{x_1 \mapsto t_1, \ldots, x_n \mapsto t_n\}$ of~$G$.
For example, $(p(\fa) \land p(\fb))\{p \mapsto \lambda x. (\fq(x) \lor
\fr(x))\}$ stands for $(\fq(\fa) \lor \fr(\fa)) \land (\fq(\fb) \lor
\fr(\fb))$.

We use \emph{tuples} of terms, variables, functions, and predicates for
several purposes. For example, to abbreviate a nested quantification $\exists
x_1 \ldots \exists x_n\, F$ with the same quantifier as a single
quantification $\exists \xs\, F$, with $\xs$ representing the tuple $x_1
\ldots x_n$. As another example, we write an atom $p(t_1,\ldots,t_n)$ as
$p(\tts)$ with $\tts$ representing the tuple $t_1 \ldots t_n$ of argument
terms. If the ordering and number of occurrences of members are irrelevant, a
tuple can be identified with the set of its members. In such cases we use a
permissive notation, where tuples can directly appear as arguments of set
operations. For tuples $\xs$ of variables we assume that the members are
pairwise different. For tuples $\ts = t_1\ldots t_n$ and $\sss = s_1\ldots
s_n$ of terms, we write the formula $\ts = \sss$ as shorthand for $t_1 = s_1
\land \ldots \land t_n = s_n$, and $\ts \neq \sss$ for $t_1 \neq s_1 \lor
\ldots \lor t_n \neq s_n$. For formulas $F$, tuples $\xs =
x_1,\ldots,x_n$ of variables and $\ts = t_1,\ldots,t_n$ of terms we write
$F\{\xs \mapsto \ts\}$ as shorthand for $F\{x_1 \mapsto t_1, \ldots, x_n
\mapsto t_n\}$.

Tables~\ref{tab:equiv}--\ref{tab:skolem} show equivalences
\eref{imp}--\eref{aqs} that will be referenced and discussed in various
contexts. The presentation is such that rewriting from the left to the right
side is the more common case. Nevertheless, for some equivalences also
rewriting in the converse direction has applications.

\begin{table}[p]\small\setlength{\arraycolsep}{3pt}
  \prlpReset{equiv}{EQ~}
    \caption{Equivalences that are independent from quantification. They are
      to be considered modulo commutativity of $\land$ and $\lor$.
      \eref{equi:cnf} and~\eref{equi:dnf} provide alternate ways to eliminate
      $\equi$.}
    \label{tab:equiv}
    \vspace{-3pt}
  \noindent
  $\begin{array}[t]{Arcll}
    \multicolumn{5}{l}{\text{\bf Eliminating Implication and Biconditional}}\\[0.45ex]
    \prlp{imp} & F \imp G & \equiv & \lnot F \lor G\\
    \prlp{equi:cnf} & F \equi G & \equiv & (\lnot F \lor G) \land (F \lor \lnot G)\\
    \prlp{equi:dnf} & F \equi G & \equiv & (F \land G) \lor (\lnot F \land \lnot G)\\[1ex]
    \multicolumn{5}{l}{\text{\bf Converting to Negation Normal Form}}\\[0.45ex]
    \prlp{nnf:not} & \lnot \lnot F & \equiv & F\\
    \prlp{nnf:and} & \lnot (F \land G) & \equiv & \lnot F \lor \lnot G\\
    \prlp{nnf:or} & \lnot (F \lor G) & \equiv & \lnot F \land \lnot G\\[1ex]
  \end{array}$
  \hspace{4pt}
  $\begin{array}[t]{Arcll}
    \multicolumn{5}{l}{\text{\bf Distributing}}\\[0.45ex]
    \prlp{dist:cnf} & F \lor (G \land H) & \equiv & (F \lor G) \land (F \lor H)\\
    \prlp{dist:dnf} & F \land (G \lor H) & \equiv & (F \land G) \lor (F \land H)\\[1ex]
    \multicolumn{5}{l}{\text{\bf Truth Value Simplification}}\\[0.45ex]
    \prlp{tv:at} & F \land \true & \equiv & F\\
    \prlp{tv:af} & F \land \false & \equiv & \false\\
    \prlp{tv:ot} & F \lor \true & \equiv & \true\\
    \prlp{tv:of} & F \lor \false & \equiv & F\\
    \prlp{tv:nf} & \lnot \true & \equiv & \false\\
    \prlp{tv:nt} & \lnot \false & \equiv & \true
  \end{array}$
  \end{table}
\begin{table}[p]\small\setlength{\arraycolsep}{3pt}
    \prlpContinue{equiv}{EQ~}
  \caption{Equivalences involving quantifiers. They apply to first- and
    second-order formulas $F,G$, and also to individual and second-order
    quantifiers as well as to combinations of these. The placeholder~$Q$
    indicates that the equivalence holds for both quantifiers $Q \in
    \{\forall, \exists\}$.}
  \vspace{-3pt}
  \label{tab:q}
  \noindent
  $\begin{array}[t]{Arcl@{\hspace{0.5em}}l}
    \multicolumn{5}{l}{\text{\bf Moving Negation over Quantifiers}}\\[0.45ex]
    \prlp{q:na} & \lnot \forall x\, F & \equiv & \exists x\, \lnot F\\
    \prlp{q:ne} & \lnot \exists x\, F & \equiv & \forall x\, \lnot F\\[1ex]
    \multicolumn{5}{l}{\text{\bf Prenexing / Moving Quantifiers Inwards}}\\[0.45ex]
    \prlp{q:aa} & \forall x\, F \land \forall x\, G & \equiv & \forall x\, (F \land G)\\
    \prlp{q:ee} & \exists x\, F \lor \exists x\, G & \equiv & \exists x\, (F \lor G)\\
    \prlp{q:qo} & Q x\, F \lor G & \equiv & Q x\, (F \lor G) &
    \text{ if } x \notin \vocplain{G}\\
    \prlp{q:qa} & Q x\, F \land G & \equiv & Q x\, (F \land G) &
    \text{ if } x \notin \vocplain{G}
  \end{array}$
$\begin{array}[t]{Arcl@{\hspace{0.5em}}l}
    \multicolumn{5}{l}{\text{\bf Reordering Quantifiers}}\\[0.45ex]
    \prlp{q:rr} & Q x Q y\, F & \equiv & Q y Q x\, F\\[1ex]
    \multicolumn{5}{l}{\text{\bf Eliminating Void Quantifiers}}\\[0.45ex]
    \prlp{q:vv} & Q x\, F & \equiv & F
    & \hspace{-2em} \text{ if } x \notin \vocplain{F}\\
  \end{array}$
\end{table}
\begin{table}[p]\small
    \prlpContinue{equiv}{EQ~}
  \caption{Equivalences that justify pushing in and pulling out terms or
    formulas. Formulas $F, G$ are first-order. For \eref{t:dis} we assume that
    $\vocplain{\ts_i} \cap \xs = \emptyset$, $i \in \{1,\ldots,n\}$, for
    \eref{t:pos} and \eref{t:neg} that $\vocplain{\ts} \cap \xs =
    \emptyset$, and for \eref{p:def}--\eref{p:ackneg} that
    $p$ does not occur in $G$ and that variables free in
    $G$ are not quantified in~$F$.}
  \vspace{-3pt}
  \label{tab:pullpush}
  \noindent
  $\begin{array}{Arcl@{\hspace{2em}}l}
    \multicolumn{5}{l}{\text{\bf Pushing In / Pulling Out Terms}}\\[0.45ex]
    \prlp{t:dis} &
    \exists \xs\, [(\bigvee_{i=1}^n \xs = \ts_i) \land F]
    & \equiv & \bigvee_{i=1}^n F\{\xs \mapsto \ts_i\}\\
    \prlp{t:pos} & \forall \xs\, (\xs = \ts \imp F) & \equiv & F\{\xs \mapsto \ts\}\\
    \prlp{t:neg} & \forall \xs\, (F \imp \xs \neq \ts) & \equiv &
    \lnot F\{\xs \mapsto \ts\}\\[1ex]
    \multicolumn{5}{l}{\text{\bf Pushing In / Pulling Out Formulas,
        Ackermann's Lemma}}\\[0.45ex]
    \prlp{p:def} &
    \exists p\, [(\forall \xs\, (p(\xs) \equi G) \land F] & \equiv
    & F\{p \mapsto \lambda \xs.G\}\\
    \prlp{p:ackpos} &
    \exists p\, [(\forall \xs\, (p(\xs) \imp G) \land F] & \equiv &
    F\{p \mapsto \lambda \xs.G\} &
    \text{ if } {-}p \notin \pred{F}\\
    \prlp{p:ackneg} &
    \exists p\, [(\forall \xs\, (G \imp p(\xs)) \land F] & \equiv &
    F\{p \mapsto \lambda \xs.G\} &
    \text{ if } {+}p \notin \pred{F}
  \end{array}$
  \end{table}
\begin{table}[p]\small\setlength{\arraycolsep}{3pt}
    \prlpContinue{equiv}{EQ~}
  \caption{Equivalences involving quantifier switching. Here $f$ is a function
    such that $f \notin \fun{F}$, $p, q$ are predicates such that $q \notin
    \predplain{F}$, $\xs, \ys$ are tuples of variables, and $y$ is a
    variable.}
  \vspace{-3pt}
  \label{tab:skolem}
  $\begin{array}[t]{Arcl}
    \multicolumn{4}{l}{\text{\bf Second-Order Skolemization}}\\[0.45ex]
    \prlp{sk} & \forall \xs \exists y\, F & \equiv & \exists f \forall \xs F\{y \mapsto f(\xs)\}
  \end{array}$
  \hspace{12pt}
    $\begin{array}[t]{Arcl}
    \multicolumn{4}{l}{\text{\bf Ackermann's Quantifier Switching}}\\[0.45ex]
    \prlp{aqs} & \forall \xs \exists p\, F & \equiv & \exists q \forall \xs
    F\{p \mapsto \lambda \ys. q(\ys\xs)\}
  \end{array}$
\end{table}

\section{Preprocessing: Skolemization, Clausification, Simplification}
\label{sec-preprocessing}

Most core techniques of automated first-order provers operate on a set of
clauses. The given first-order formula whose unsatisfiability is to be proven
is thus \emph{preprocessed} to a set of clauses which is
\name{equi-satisfiable}, i.e., is unsatisfiable iff the original formula is
so. Conceptually, this proceeds in the following steps: (1)~\name{Prenexing},
bringing the formula into a form where a quantifier prefix is applied to a
quantifier-free formula by rewriting with \eref{q:na}--\eref{q:qa}.
(2)~\name{Skolemization}, eliminating existential variables with \eref{sk}.
This yields a formula whose quantifier prefix starts with existential
quantifiers over fresh functions, the \name{Skolem functions}. Unless the
Skolem function is a constant, the quantifier is a second-order quantifier.
(For determining unsatisfiability, the existential quantifier prefix over the
Skolem functions can be dropped.) (3)~\name{CNF transformation}, converting
the formula to a set of clauses, by \eref{imp}--\eref{dist:cnf}.
(4)~\name{Simplification}, simplifying the set of clauses to an
equi-satisfiable set.

Typically, these phases are performed interleaved. Skolemization, e.g., might
be applied separately to subformulas obtained by shifting quantifiers with
\eref{q:aa}--\eref{q:qa} outwards as well as inwards
\cite{soqe:book:2008,nonnengart:weidenbach:small:nf:2001}.
Distributing with \eref{dist:cnf} at CNF transformation calls for simplifying
formulas before they get duplicated. Blow-up through distributing can be
completely avoided with \name{structure-preserving} (or \name{definitional})
\name{normal forms}
\cite{handbook:baaz:egly:leitsch:nf}
that are based on rewriting with \eref{p:def}--\eref{p:ackneg} from right to
left before CNF transformation.

Simplification is important in automated reasoning and has many facets. The
general idea is that the expensive core proving is accompanied by cheaper
operations that do whatever can be done with few resources to make the problem
easier. Simplifications are not just applied at preprocessing but also
incorporated into the core proving. For example, the core operation of CDCL
SAT solvers involves unit propagation, a propositional simplification. As
another example, CDCL SAT solvers interrupt the core reasoning to simplify the
computed lemmas (learned clauses) in \name{inprocessing} phases. Resolution
provers can delete a clause if it is newly inferred and subsumed by a
previously inferred clause (\name{forward subsumption}) and if it was
previously inferred and is subsumed by a newly inferred clause (\name{backward
  subsumption}).

Simplifications for theorem proving fall into two broad categories, those that
preserve equivalence and those that do not preserve equivalence but
equi-satisfiability.
Well-known simplifications of the first class are
deletion of tautological clauses, deletion of subsumed clauses, replacing
clauses by condensation, and replacing clauses by subsumption resolution
\cite{chang:lee,resolution:handbook:2001,nonnengart:weidenbach:small:nf:2001}.
In view of interpolation, we note that none of these introduces additional
functions or predicates, which makes them straightforwardly applicable to each
of the two interpolated formulas.

A well-known simplification of the second class is the deletion of a clause
that contains a \name{pure} literal, a literal whose predicate does not occur
with opposite polarity in the set of clauses. A closer looks shows that this
does not just preserve equi-satisfiability, but, moreover, if $F'$ is obtained
from~$F$ by deleting a clause with a pure literal with predicate $p$, then
$\exists p\, F \equiv \exists p\, F'$. Evidently, no additional functions or
predicates are introduced. We can utilize such properties of simplifications
that do not preserve equivalence but actually preserve more than just
equi-satisfiability to justify their use in the preprocessing of interpolated
formulas. The following proposition shows corresponding constraints that are
suitable for Craig interpolation.

\pagebreak
\begin{proposition}[Simplifying Interpolated Formulas]
  \label{prop-interpol-constraints}
  Let $F, G, F', G'$ be first-order formulas such that
  \[\exists \sffs'\, F' \equiv \exists \sffs F,\;
  \forall \sggs'\, G' \equiv \forall \sggs\, G,\; \vocplain{\exists \sffs'\,
    F'} = \vocplain{\exists \sffs\, F},\text{ and }\; \vocplain{\forall
    \sffs'\, G'} = \vocplain{\forall \sggs\, G},\] where $\sffs' =
  \vocplain{F'} \setminus \vocplain{G'}$, $\sggs' = \vocplain{G'} \setminus
  \vocplain{F'}$, $\sffs = \vocplain{F} \setminus \vocplain{G}$, and $\sggs =
  \vocplain{G} \setminus \vocplain{F}$. Then, a first-order formulas $H$ is a
  Craig interpolant for $F, G$ iff it is a Craig interpolant for $F', G'$.
\end{proposition}
The quantifications in this proposition are mixed, first- and second-order,
over predicates, functions, and constants or variables. The proposition shows
constraints for preprocessing given formulas $F, G$ to simpler formulas $F',
G'$ that have the same Craig interpolants.
Simplifications that match these constraints are, e.g., the mentioned deletion
of a clause with a pure literal, unless the literal's predicate is
\FGSHARED. At least in the propositional case also \name{blocked clause
  elimination} \cite{kullmann:1999,jarv:blocked:2010,blocked:fol:2017}, which
generalizes the purity deletion, is justified by
Prop.~\ref{prop-interpol-constraints}, if the predicates of the blocking
literals are not \FGSHARED. Moreover, since
Prop.~\ref{prop-interpol-constraints} stays unchanged if $F,G$ and $F',G'$ are
switched, also cases of \emph{addition} of blocked clauses
\cite{kullmann:1999} are justified.
Justified transformations that enrich the set of predicates include conversion
to structure-preserving normal forms, if the fresh predicates
introduced in the conversion of~$F$ and of $G$ are distinct.
Also second-order quantifier elimination, if it does not eliminate \FGSHARED
predicates, is justified. Modern SAT solvers \cite{een:biere:elim} as well as
some first-order provers, such as \ProverN, use elimination in specific cases
as simplifications. However, typically these systems provide no options for
specifying predicates as protected from being eliminated.

For \emph{Craig-Lyndon} interpolation, Prop.~\ref{prop-interpol-constraints}
can be relaxed by ``quantification over predicates only in a specific
polarity'', which can be defined as follows.
\begin{definition}[Polarity-Sensitive Predicate Quantification]
  \label{def-polarity-quantification}
  For second-order formulas $F$ and predicates $p$ define $\exists {+}p\, F$
  and $\exists {-}p\, F$ as follows, where $p'$ is a fresh predicate.
  \[\begin{array}{lcl@{\hspace{2em}}lcl}
  \exists {+}p\, F & \eqdefnospace &
  \exists p'\, (F\{p \mapsto p'\} \land \forall \xs\, (p(\xs) \imp p'(\xs)). &
  \forall {-}p\, F & \eqdefnospace & \lnot \exists {+}p\, \lnot F.\\
  \exists {-}p\, F & \eqdefnospace &
  \exists p'\, (F\{p \mapsto p'\} \land \forall \xs\, (p'(\xs) \imp p(\xs)). &
  \forall {+}p\, F & \eqdefnospace & \lnot \exists {-}p\, \lnot F.\\
  \end{array}
  \]
\end{definition}
The following examples illustrate polarity-sensitive predicate quantification
by showing the expansion into conventional predicate quantification and an
equivalent first-order formula, obtained by second-order quantifier
elimination, which is discussed in Sect.~\ref{sec-soqe}.
\begin{example}
  \ \subexample{ex-polquant-1} $\exists {-}\fp\, (\fp(\fa) \land \lnot
  \fp(\fb))\; \equiv\; \exists p'\, (p'(\fa) \land \lnot p'(\fb) \land \forall
  x (p'(x) \imp \fp(x)))\; \equiv\; \fa \neq \fb \land \fp(\fa)$.

  \subexample{ex-polquant-2}
  $\begin{array}[t]{ll}
  & \exists {-}\fq\, (\forall x (\fp(x) \imp \fq(x)) \land \forall x\, (\fq(x) \imp \fr(x)))\\
  \equiv & \exists q'\, (\forall x\, (\fp(x) \imp q'(x)) \land
  \forall x\, (q'(x) \imp \fr(x)) \land  \forall x\, (q'(x) \imp \fq(x)))\\
  \equiv & \forall x (\fp(x) \imp (\fq(x) \land \fr(x))).
  \end{array}$
  \par\vspace{-12pt}\lipicsEnd
\end{example}
Vice versa, conventional predicate quantification can be expressed by
polarity-sensitive quantification: $\exists p\, F \equiv \exists {+}p \exists
{-}p\, F$. Proposition~\ref{prop-interpol-constraints} can be adapted to
Craig-Lyndon interpolation by considering predicate quantification as
polarity-sensitive and using $\vocfun$ in place of $\vocplainfun$.

\section{The Two-Stage Approach to Craig Interpolation in First-Order Logic}
\label{sec-lifting}

Herbrand's theorem, due to Herbrand in 1930 \cite{herbrand}, is a central tool
in automated first-order theorem proving, where it is commonly stated in the
following form.
\begin{theorem}[Herbrand's Theorem]
  \label{thm-h}
A (possibly infinite) set $F$ of clauses is unsatisfiable iff there is a
finite set of ground instances of clauses of $F$.
\end{theorem}
\begin{example}
  Let $F = \{\fp(\fa),\; \lnot \fp(x) \lor \fp \fp(\ff(x)),\; \lnot
  \fp(\ff(\ff(\fa)))\}$. Then $\{\fp(\fa),\; (\lnot \fp(x) \lor \fp
  \fp(\ff(x)))\{x \mapsto \fa\},\; (\lnot \fp(x) \lor \fp \fp(\ff(x)))\{x
  \mapsto \ff(\fa)\},\; \lnot \fp(\ff(\ff(\fa)))\}$
  is an unsatisfiable set of
  ground instances of $F$.
  \lipicsEnd
\end{example}
To obtain an unsatisfiable set of ground clauses according to Herbrand's
theorem it is sufficient to instantiate with terms built from functions in the
given set of clauses and, if there is no constant among these, an additional
constant~$c_0$.

Given the conversion of a first-order formula to a set of clauses as discussed
in Sect.~\ref{sec-preprocessing}, Herbrand's theorem tells us that behind a
first-order proof there is a propositional proof, with ground atoms in the
role of propositional variables. The proof object obtained from a theorem
prover can represent such a ground proof. This suggests to calculate a
Craig-Lyndon interpolant for first-order formulas in \emph{two stages}: (1)
Calculating a Craig-Lyndon interpolant from a ground proof with propositional
techniques; (2) Interpolant lifting, that is replacing ground terms with
quantified variables.

A corresponding two-stage proof of the Craig interpolation property of
first-order logic was given in the 1960s by Kreisel and Krivine
\cite{kreisel:krivine:modeltheory:1971}. Harrison \cite{harrison:2009}
presents a version of their proof that is adapted to automated reasoning,
explicitly referring to Skolemization and Herbrand's theorem. In automated
reasoning, the two-stage approach was introduced in 1995 with Huang's paper
\cite{huang:95} on interpolation with resolution, which was rediscovered in
the 2010s by Bonacina and Johansson \cite{bonacina:2015:on}, who coined the
name \name{two-stage approach}. They prove a limited form of interpolant
lifting that applies just to constants in contrast to compound terms. Baaz and
Leitsch \cite{baaz:leitsch:2011} prove a general form, which they call
\name{abstraction}, on the basis of a natural deduction calculus.

We present here a general form of interpolant lifting from \cite{cw:ipol},
where it is proven on the basis of Skolemization and Herbrand's theorem. The
relationships that allow a ground formula $\HG$ to be lifted to a Craig
interpolant for first-order formulas $F,G$ are captured by the notion of
\name{interpolant lifting base}. Given a lifting base, i.e., if the specified
relationships hold, a first-order Craig interpolant for $F,G$ can be
constructed from $\HG$ by replacing certain occurrences of terms with
variables and prepending a certain first-order quantifier prefix upon these
variables. The Craig interpolation property for first-order logic follows
since if $F \entails G$, then an abstract construction ensures existence of a
lifting base.
Clausal tableaux and resolution deductions provide proofs that can
straightforwardly be viewed \emph{as if} they were constructed by the abstract
method. This makes the two-stage approach applicable with different underlying
first-order calculi for ground interpolation, where we will discuss clausal
tableaux (Sect.~\ref{sec-tableaux}) and resolution
(Sect.~\ref{sec-resolution}).

We now specify \name{interpolant lifting base} formally. W.l.o.g. we assume
that the interpolated first-order formulas $F,G$ are \emph{sentences}: To
interpolate $F,G$ with free variables, we first replace these by dedicated
constants and finally replace the constants in the obtained interpolant by the
corresponding variables.
\begin{definition}[Interpolant Lifting Base]
  \label{def-ib}
  \prlReset{def-ib} An \defname{interpolant lifting base} is a tuple \[\la F,
  G, \ffs, \ggs, \HG\ra,\] where $F, G$ are first-order sentences, $\ffs,
  \ggs$ are disjoint tuples of distinct function symbols, $\HG$ (the subscript
  \textsc{grd} suggesting \name{ground}) is a ground formula such that there
  exist quantifier-free formulas $\FE(\us), \GE(\vs)$ (the subscript
  \textsc{exp} suggesting \name{expansion}) with disjoint tuples of free
  variables $\us, \vs$, respectively, and a ground substitution $\sth$ with
  the following properties.

  \smallskip
  \noindent
  $\begin{array}{r@{\hspace{0.4em}}l@{\hspace{0.5em}}r@{\hspace{0.4em}}l}
    \ibn{ib:sem} &
    F \entails \exists \ffs \forall \us\, \FE(\us). &
    \ibm{ib:sem} &
    \forall \ggs \exists \vs\, \GE(\vs) \entails G.\\
    \ibn{ib:pred} &
    \pred{\FE(\us)} \subseteq \pred{F}.
    & \ibm{ib:pred} &
    \pred{\GE(\vs)} \subseteq \pred{G}.\\
    \ibn{ib:fe} &
    \fun{\FE(\us)} \subseteq (\fun{F} \tightop{\cap} \fun{G}) \tightop{\cup} \ffs.
    & \ibm{ib:fe} &
    \fun{\GE(\vs)} \subseteq (\fun{F} \tightop{\cap} \fun{G}) \tightop{\cup} \ggs.\\
    \ibn{ib:fcap} &
    \fun{F} \cap \ggs = \emptyset.
    & \ibm{ib:fcap} &
    \fun{G} \cap \ffs = \emptyset.\\
  \end{array}\\[1ex]
  \begin{array}{r@{\hspace{0.4em}}l}
    \ibn{ib:domh} & \dom{\sth} = \us \cup \vs.\\
    \ibn{ib:rngh} &
    \fun{\rng{\sth}} \subseteq \fun{\FE(\us)} \cup \fun{\GE(\vs)} \cup \{c_0\},
    \text{where } c_0 \text{ is a constant in } \ffs\ggs.\\
    \ibn{ib:ipol} & \HG \text{ is a Craig-Lyndon interpolant for } \FE(\us)\sth
    \text{ and } \GE(\vs)\sth.\\
  \end{array}$

  \smallskip
\end{definition}

For first-order sentences $F,G$ such that $F \entails G$ we can construct an
interpolant lifting base as follows. Apply prenexing, Skolemization and CNF
transformation independently to each of $F, \lnot G$ to obtain formulas
$\exists \ffs' \forall \us'\, M'(\us')$, $\exists \ggs' \forall \vs'\,
N'(\vs')$ that are equivalent to $F, \lnot G$, respectively, where
$\ffs',\ggs'$ are the introduced Skolem functions and $M'(\us'), N'(\vs')$ are
sets of clauses, with free variables $\us', \vs'$.
Since $M'(\us') \cup N'(\vs')$ is unsatisfiable, by Herbrand's theorem there
is an unsatisfiable finite set of ground instances of clauses from $M'(\us')
\cup N'(\vs')$.
This set of ground instances can contain different instances of the same
clause in $M'(\us') \cup N'(\vs')$. Thus, this set can be considered as
obtained in two steps, by first creating \emph{copies} of clauses in $M'(\us')
\cup N'(\vs')$, one copy for each ground instance, where a \name{copy} of a
clause is a variant with fresh variables, and, second, applying a single
ground substitution to the set of copies. Let $\ME(\us)$ ($\NE(\vs)$) with
free variables $\us$ ($\vs$) be the set of these copies and let $\sth$ be the
ground substitution. The range of $\sth$ is a set of terms built from
functions in $M'(\us') \cup N'(\vs')$ and, if there is no constant among
these, a fresh constant $c_0$. Let $\ffs$ ($\ggs$) be the union of the Skolem
functions $\ffs'$ ($\ggs'$) introduced for $F$ ($G$) and the \FONLY (\GONLY)
functions. In case a fresh $c_0$ was introduced, add it either to $\ffs$ or to
$\ggs$. Let $\FE(\us) = \ME(\us)$, let $\GE(\vs) = \lnot \NE(\vs)$, and let
$\HG$ be a ground Craig-Lyndon interpolant for $\FE(\us)\sth, \GE(\vs)\sth$.
Then $\la F, G, \ffs, \ggs, \HG\ra$ is an interpolant lifting base.

Neither the ``copy expansions'' $\FE(\us), \GE(\vs)$, nor the ground
substitution~$\sth$ have to be actually constructed for obtaining an
interpolant lifting base. Just their \emph{existence} is required.
\begin{example}
  \label{examp-ibase}
  Each of the following examples shows a lifting base $\la F, G, \ffs, \ggs,
  \HG\ra$ together with suitable $\FE(u), \GE(v), \sth$. Introductory comments
  show specific properties.

  \subexample{ex-ib-nonlocal} $\ffs$ contains a non-constant; members of
  $\ffs$ and of $\ggs$ occur in $\HG$: $ \la F, G, \ffs, \ggs,
  \HG\ra\; =\; \la \forall x\, \fp(x, \ff(x)),\; \exists x\, \fp(\fg,
  x),\ \ff,\; \fg,\; \fp(\fg, \ff(\fg))\ra$, with $\FE(u) = \fp(u,
  \ff(u))$, $\GE(v) = \fp(\fg, v)$, $\sth = \{u \mapsto \fg, v \mapsto
  \ff(\fg)\}.$ It holds that $\FE(u)\sth = \GE(v)\sth = \HG = \fp(\fg,
  \ff(\fg))$.

  \subexample{ex-ib-skolem} A member of $\ffs$ (i.e., $\ff_2$) is a Skolem
  function: $\la F, G, \ffs, \ggs, \HG\ra\; =\; \la \forall x \exists y\,
  \fp(x, y, \ff_1),\linebreak\exists x \exists y \, \fp(\fg, x, y),\;
  \ff_1\ff_2,\; \fg,\; \fp(\fg,\ff_2(\fg),\ff_1)\ra$, with $\FE(u) =
  \fp(u,\ff_2(u),\ff_1),\; \GE(v_1,v_2) =\linebreak \fp(\fg, v_1, v_2),\; \sth
  = \{u \mapsto \fg, v_1 \mapsto \ff_2(\fg), v_2 \mapsto \ff_1\}.$ It holds
  that $\FE(u)\sth = \GE(v_1,v_2)\sth = \HG = \fp(\fg,\ff_2(\fg),\ff_1)$.

  \subexample{ex-ib-expansion} $\FE(u_1,u_2)$ is a conjunction of different
  variants of the quantifier-free inner formula $\fp(x,\ff)$ of $F$:
  $\la F, G, \ffs, \ggs, \HG\ra = \la \forall x\, \fp(x, \ff),\; \exists x\,
  \fp(\fg_1, x) \land \exists x\, \fp(\fg_2, x),\; \ff,\; \fg_1\fg_2,\;
  \fp(\fg_1, \ff) \land \fp(\fg_2, \ff)\ra$, with $\FE(u_1,u_2) = \fp(u_1,
  \ff) \land \fp(u_2, \ff),\; \GE(v) = \fp(\fg_1, v) \land \fp(\fg_2, v),\;
  \sth = \{u_1 \mapsto \fg_1, u_2 \mapsto \fg_2, v \mapsto \ff\}$. It holds
  that $\FE(u_1,u_2)\sth = \GE(v)\sth = \HG = \fp(\fg_1, \ff) \land \fp(\fg_2,
  \ff)$.

  \subexample{ex-ib-qr} Formulas~$F,G$ extend those of
  Example~(\ref{ex-ib-nonlocal}) by literals with predicates~$\fq, \fr$ that
  occur in only one of $F,G$, and a second function symbol in $G$. Differently
  from the previous three cases, the ground interpolant is not an instance of
  $\FE$ and $\GE$. $\la F = \forall x\, \fp(x, \ff(x)) \land \forall x \forall
  y\, \fq(\ff(x),y),\; G = \exists x\, (\fp(\fg_1, x) \lor \fr(\fg_2(x))),\;
  \ffs = \ff,\; \ggs = \fg_1, \fg_2,\; \HG = \fp(\fg_1, \ff(\fg_1))\ra$, with
  $\FE(u_1, u_2, u_3) = \fp(u_1, \ff(u_1)) \land \fq(\ff(u_2), u_3),\; \GE(v)
  = \fp(\fg_1, v) \lor \fr(\fg_2(v)),\; \sth = \{u_1 \mapsto \fg_1, v \mapsto
  \ff(\fg_1), u_2 \mapsto \fg_2(\ff(\fg_1)), u_3 \mapsto \fg_1\}$. It holds
  that $\FE(u_1, u_2, u_3)\sth = \fp(\fg_1, \ff(\fg_1)) \land
  \fq(\ff(\fg_2(\ff(\fg_1))), \fg_1)$ and $\GE(v)\sth = \fp(\fg_1, \ff(\fg_1))
  \lor \fr(\fg_2(\ff(\fg_1)))$. Other values of~$\FE$, $\GE$ and~$\sth$ are
  also possible. For example, $u_2, u_3$ could be merged with~$u_1$, or $\sth$
  could assign $u_2, u_3$ to other ground terms. \lipicsEnd
\end{example}

Given a lifting base $\la F, G, \ffs, \ggs, \HG \ra$ for first-order sentences
$F, G$ such that $F \entails G$, we can construct a Craig-Lyndon interpolant
for $F, G$ as specified with Theorem~\ref{thm-lifting} below. Its statement
needs additional notation. We write $\strictsubterm{s}{t}$ to express that $s$
is a strict subterm of~$t$.
If $\symseq{f}$ is a set or sequence of functions, then an $\symseq{f}\sterm$
is a term whose outermost symbol is in $\symseq{f}$.
An occurrence of an $\symseq{f}\sterm$ in a formula $F$ is
\defname{$\symseq{f}$-maximal} if it is not within an occurrence of another
$\symseq{f}\sterm$.
If $\sigma$ is an injective variable substitution whose range is a set of
ground $\symseq{f}\sterms$, then $\invsubst{F}{\sigma}$ denotes $F$ with all
$\symseq{f}$-maximal occurrences of terms $t$ in the range of $\sigma$
replaced by the variable that is mapped by $\sigma$ to $t$.
As an example, let $\symseq{f} = \ff\fg$, let $F =
\fp(\fh(\ff(\fa),\fg(\ff(\fa))))$ and let $\sigma = \{x \mapsto \ff(\fa),\, y
\mapsto \fg(\ff(\fa))\}$. Then $\invsubst{F}{\sigma} = \fp(\fh(x,y))$. We can
now state the theorem that specifies the variable introduction into ground
interpolants.

\begin{theorem}[Interpolant Lifting]
  \label{thm-lifting}
  Let $\la F, G, \ffs, \ggs, \HG \ra$ be an interpolant lifting base. Let
  $\{t_1, \ldots, t_n\}$ be the set of the $\fgs\sterms$ with an
  $\fgs$-maximal occurrence in $\HG$, ordered such that if
  $\strictsubterm{t_i}{t_j}$, then $i < j$. Let $\{v_1,\ldots,v_n\}$ be a set
  of fresh variables and let $\stt$ be the injective substitution $\stt \eqdef
  \{v_i \mapsto t_i \mid i \in \{1,\ldots,n\}\}$. For $i \in \{1,\ldots,n\}$
  let $Q_i \eqdef \exists$ if $v_i\stt$ is an $\ffs\sterm$ and $Q_i \eqdef
  \forall$ otherwise, that is, if $v_i\stt$ is a $\ggs\sterm$. Then
  \[H = Q_1 v_1 \ldots Q_n v_n\, \HG\invsubstpost{\stt}\] is a Craig-Lyndon
  interpolant for $F, G$.
\end{theorem}

Theorem~\ref{thm-lifting} shows the construction of a first-order sentence $H$
from a given ground formula~$\HG$ and sets, $\ffs, \ggs$, of function symbols.
Sentence $H$ is obtained from $\HG$ by replacing $\ffs\ggs$-maximal
occurrences of $\ffs\sterms$ and $\ggs\sterms$ with variables, and prepending
a quantifier prefix over these. Variables replacing an $\ffs\sterm$
($\ggs\sterm$) are existential (universal), and whenever variables $x,y$
replace terms~$s,t$, respectively, such that $\strictsubterm{s}{t}$, then the
quantifier over $x$ precedes that over $y$. The obtained $H$ is a Craig-Lyndon
interpolant for the first-order sentences~$F,G$, provided $\la F, G, \ffs,
\ggs, \HG \ra$ is an interpolant lifting base.
\begin{example}
  \label{examp-lifting}
  Consider the interpolant lifting bases from Example~\ref{examp-ibase}.
  Respective Craig-Lyndon interpolants according
  to Theorem~\ref{thm-lifting} are as follows.
  For (\ref{ex-ib-nonlocal}) and for (\ref{ex-ib-qr}):
  $\forall v_1 \exists v_2\, \fp(v_1, v_2)$.
  For (\ref{ex-ib-skolem}): $\exists v_1 \forall v_2 \exists v_3\, \fp(v_2,
  v_3, v_1)$. Also other orderings of the quantifiers are possible according
  to Theorem~\ref{thm-lifting}. The only required condition is (expressed with
  the variable names of the shown interpolant) that $\forall v_2$ must precede
  $\exists v_3$.
  For (\ref{ex-ib-expansion}): $\exists v_1 \forall v_2 \forall v_3\,
  (\fp(v_2, v_1) \land \fp(v_3, v_1))$, which is equivalent to $\exists v_1
  \forall v_2 \fp(v_2,v_1)$. Also arbitrary other quantifier orderings
  are possible.
  \lipicsEnd
\end{example}

\section{Ground Interpolation with Clausal Tableaux}
\label{sec-tableaux}

We discuss clausal tableaux as a technique for first-order theorem proving and
ground interpolation in the two-stage approach.

\subsection{Clausal Tableaux -- Proof Objects for Automated Reasoning}

The framework of \name{clausal tableaux}
\cite{letz:diss,handbook:tableaux:letz,letz:stenz:handbook,letz:habil,handbook:ar:haehnle:2001}
was developed in the 1990s by Letz as a bridge between analytic tableaux and a
family of methods for fully automated first-order proving with highly
optimized systems. These methods, model elimination \cite{loveland:1968} and
the connection method \cite{bibel:atp:1987}, share with resolution the
operation on sets of clauses, but instead of generating consequences they
enumerate proof structures.
First-order provers that can be described as constructing clausal tableaux
include the \name{Prolog Technology Theorem Prover (PTTP)}
\cite{pttp:1984}, \SETHEO
\cite{setheo:1992,letz:stenz:handbook}
and \METEOR \cite{meteor:1994}. Until around 2001, \SETHEO was a competitive
first-order prover. \CMProver \cite{cw:pie:2016,cw:pie:2020} is a
Prolog-based system that is still maintained. In 2003 \leanCoP
\cite{otten:2003:leancopinabstract} was designed as a minimalistic
Prolog-based system that has since been used for numerous studies and
adaptations to non-classical logics.
For Craig interpolation, the crucial relevance of the clausal tableau
framework is motivated by its \name{proof objects} -- the \name{clausal
  tableaux}. Based on clauses instead of complex formulas, they are compatible
with highly optimized fully automated systems. Their tree structure allows
inductive calculation of ground interpolants as known from analytic tableaux
and sequent systems. The simplicity of clausal tableaux facilitates abstract
investigations, proof transformations, and developing strengthened variations
of Craig-Lyndon interpolation.

\pagebreak
\begin{definition}[Clausal Tableau and Related Notions]
  \label{def-ct}
A \defname{clausal tableau} (briefly \name{tableau}) \defname{for} a set $F$
of clauses is a finite ordered tree whose nodes~$N$ with exception of the root
are labeled with a literal $\nlit{N}$, such that for each inner node~$M$ the
disjunction of the literals of all its children in their left-to-right order,
$\nclause{M}$, is an instance of a clause in~$F$. The clauses $\nclause{M}$
are called the \defname{clauses of} the tableau. A tableau whose clauses are
ground is called \defname{ground}.
A branch of a tableau is \defname{closed} iff it contains nodes with
complementary literals. A node is \defname{closed} iff all branches through it
are closed. A tableau is \defname{closed} iff its root is closed.
\end{definition}
All occurrences of variables in the clauses of a clausal tableau are free and
their scope spans the whole tableau. That is, we consider \name{free-variable
  tableaux} \cite[p.~158ff]{handbook:tableaux:letz}
\cite[Sect.~2.2]{letz:habil}, or \name{rigid} variables
\cite[p.~114]{handbook:ar:haehnle:2001}.
That a clausal tableau indeed represents a proof is stated in the following
proposition, which follows from Herbrand's theorem (Theorem~\ref{thm-h}).
\begin{proposition}
  \label{prop-unsat-ct}
  A set $F$ of clauses is unsatisfiable iff there exists a closed clausal
  tableau for $F$.
\end{proposition}
Clauses in a closed tableau according to Prop.~\ref{prop-unsat-ct} may
have variables. A closed \emph{ground} tableau, in direct correspondence to
Herbrand's theorem, can then be obtained by instantiating each variable with
an arbitrary ground term built from functions in the given set of clauses and,
if there is no constant among these, an extra constant.

As proof systems, clausal tableaux and cut free analytic tableaux, as well as
clausal tableaux with atomic cut (Sect.~\ref{sec-atomic-cut}) and analytic
tableaux with atomic cut, polynomially simulate each other if
structure-preserving normal forms are permitted in both types of tableaux
\cite[p.~119]{letz:diss}.

\subsection{The Connection Tableau Calculus}

For Craig interpolation, our main interest in clausal tableaux is as proof
objects that were delivered by an automated system, without caring about
\emph{how} they were found in proof search with some calculus. Nevertheless,
we briefly present a clausal tableau calculus, the \name{connection tableau
  calculus}, referring to \cite{letz:habil} for a comprehensive discussion.
Differently from analytic tableaux, the initially given formula is not placed
on the tableau, but kept separately, as a set of \emph{input clauses}. The
calculus builds the tableau by attaching copies, i.e., variants with fresh
variables, of input clauses to the tableau and by instantiating variables in
these tableau clauses. Instantiating is done with most general unifiers that
equate a leaf literal with the complement of an ancestor literal, such that a
branch gets closed. The rules of the calculus involve several nondeterministic
selections. For completeness, a backtracking regime has to ensure that each
possible selection is eventually made.

\begin{definition}[Connection Tableau Calculus]
  \label{def-calc-ser} The \defname{connection tableau calculus}
  consists of the following three rules.
  \begin{description}

  \item{\rname{Start}:} If the tableau consists only of the root node, select
    a clause from the input clauses, make a fresh copy and attach children
    with its literals to the root.

  \item{\rname{Extension}:} Select an open branch with leaf $N$ and select an
    input clause $L' \lor C$ such that $\nlit{N}$ and $\du{\nlit{L'}}$ have a
    most general unifier. Make a fresh copy $L'' \lor C'$ of the clause and
    attach children with its literals to $N$. Let $\sigma$ be the most general
    unifier of $L$ and $\du{L''}$ and apply $\sigma$ to all literals of the
    tableau. The branch ending in the child corresponding to $L''$ is then
    closed.

\item{\rname{Reduction}:} Select an open branch with leaf $N$ and select an
  ancestor $N'$ of $N$ such that $\nlit{L}$ and $\du{\nlit{L'}}$ have a most
  general unifier $\sigma$. Apply $\sigma$ to the tableau. The branch ending
  in $N$ is then closed.
\end{description}
\end{definition}
\begin{example}
  \label{examp-ser}
  Let $F = \{\lnot \fp(x,\ff(x)),\; (\fp(x,y) \lor \lnot \fq(y)),\; (\fq(x)
  \lor \fp(\fg(y),x))\}$. The connection tableau calculus can build a closed
  clausal tableau for $F$ as follows.

  
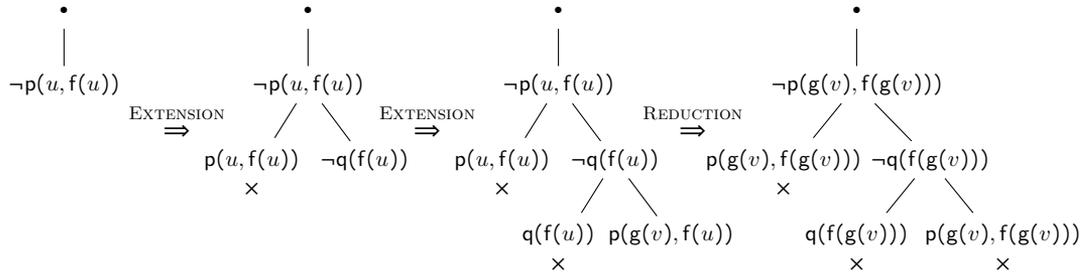
\vspace{-3mm}
  \noindent\hspace{-1.4mm}
  \begin{tikzpicture}[scale=0.84,
        baseline=(a.north),
      sibling distance=5em,level distance=8ex,
      every node/.style = {transform shape,anchor=mid}]
      \node (a) {\vbar\textbullet}
      child { node {$\lnot \fp(u,\ff(u))$} };
  \end{tikzpicture}
  \raisebox{-13ex}{$\!\stackrel{\text{\rname{Extension}}}{\Rightarrow}\hspace{-1.7em}$}
  \begin{tikzpicture}[scale=0.84,
      baseline=(a.north),
      sibling distance=5em,level distance=8ex,
      every node/.style = {transform shape,anchor=mid}]
      \node (a) {\vbar\textbullet}
      child { node {$\lnot \fp(u,\ff(u))$}
        child { node {$\begin{array}[t]{c}\fp(u,\ff(u))\\\times\end{array}$}}
        child { node {$\lnot \fq(\ff(u))$ }}};
  \end{tikzpicture}
  \raisebox{-13ex}{$\hspace{-1.6em}\stackrel{\text{\rname{Extension}}}{\Rightarrow}\hspace{-1.7em}$}
  \begin{tikzpicture}[scale=0.84,
      baseline=(a.north),
      sibling distance=5em,level distance=8ex,
      every node/.style = {transform shape,anchor=mid}]
      \node (a) {\vbar\textbullet}
      child { node {$\lnot \fp(u,\ff(u))$}
        child { node {$\begin{array}[t]{c}\fp(u,\ff(u))\\\times\end{array}$}}
        child { node {$\lnot \fq(\ff(u))$ }
          child { node {$\begin{array}[t]{c}\fq(\ff(u))\\\times\end{array}$} }
          child { node {$\fp(\fg(v),\ff(u))$}}
        }};
  \end{tikzpicture}
  \raisebox{-13ex}{$\hspace{-3.9em}\stackrel{\text{\rname{Reduction}}}{\Rightarrow}\hspace{-2.3em}$}
  \begin{tikzpicture}[scale=0.84,
      baseline=(a.north),
      sibling distance=6.5em,level distance=8ex,
      every node/.style = {transform shape,anchor=mid}]
      \node (a) {\vbar\textbullet}
      child { node {$\lnot \fp(\fg(v),\ff(\fg(v)))$}
        child { node {$\begin{array}[t]{c}\fp(\fg(v),\ff(\fg(v)))\\\times\end{array}$} }
        child { node {$\lnot \fq(\ff(\fg(v)))$ }
          child { node {$\begin{array}[t]{c}\fq(\ff(\fg(v)))\\\times\end{array}$} }
          child { node {$\begin{array}[t]{c}\fp(\fg(v),\ff(\fg(v)))\\\times\end{array}$}}
        }};
  \end{tikzpicture}

  \medskip

  \noindent
  The first step, \rname{Start}, adds a copy $\lnot \fp(u,\ff(u))$ of the
  first input clause. Then an \rname{Extension} step adds a copy $\fp(x',y')
  \lor \lnot \fq(y')$ of the second input clause and applies the unifier $\{x'
  \mapsto u, y' \mapsto \ff(u)\}$. A second \rname{Extension} step adds a copy
  $\fq(x'') \lor \fp(\fg(v),x'')$ of the third input clause and applies the
  unifier $\{x'' \mapsto \ff(u)\}$. Finally, a \rname{Reduction} step with the
  node labeled by $\fp(a,\ff(u))$ as $N$ and its ancestor labeled by $\lnot
  \fp(u,\ff(u))$ as $N'$ applies the unifier $\{u \mapsto \fg(v)\}$.
  \lipicsEnd
\end{example}

If we add a further input clause $C = \fp(x,y) \lor \lnot \fr(y)$ to $F$ from
Example~\ref{examp-ser} and select it as start clause, then no further rule is
applicable. Also, if, as in the example, $\lnot \fp(x,\ff(x))$ is selected as
start clause but the first extension step is with $C$, then no further rule is
applicable. In such cases alternate selections have to be explored. In
implementations this is typically done with chronological backtracking
embedded in an iterative deepening upon the depth of the tableau tree or some
other measure (e.g., \cite[Sect.~6.2]{setheo:1992}). With iterative deepening,
if the minimal depth of a closed tableau for the given formula is $n$, proof
search exhaustively explores for all $i \in \{0, \ldots, n-1\}$ the trees of
depth up to $i$ that are generated by the calculus rules. Finally, in
iteration $n$ it terminates with the first closed tableau it finds. In
variations of this setup, the prover enumerates alternate closed tableaux,
ordered by increasing depth, which is of interest for interpolation since
different tableaux yield different interpolants.

A core idea of the connection method is to guide proof search by the
\emph{connections} in the given formula, pairs of literal occurrences with the
same predicate but complementary polarity. This is reflected in the connection
tableau calculus in that at each step except of \rname{Start} a pair of
literal occurrences is made complementary through unification, closing an open
branch. The generated tableaux are thus \name{strongly connected}, which is
defined as follows.
\begin{definition}[Strong Connection Condition]
  A clausal tableau is \defname{strongly connected} iff every inner
  node with exception of the root has a node with complementary literal as a
  child.
\end{definition}
The strong connection condition does not affect completeness, i.e., whenever
there is a closed clausal tableau for a set $F$ of clauses, then there is a
strongly connected closed clausal tableau for $F$
\cite[Sect.~5.2]{letz:habil}, \cite[Sect.~4]{handbook:ar:haehnle:2001}.

\subsection{The Regularity Restriction}
\label{sec-regularity}

\name{Regularity} is an important restriction of clausal tableaux, defined as
follows.

\begin{definition}[Regular]
A clausal tableau is \defname{regular} iff no node has an ancestor
with the same literal.
\end{definition}

The number of nodes of a non-regular closed tableau can be strictly reduced
with the following operation that again yields a closed tableau for the same
set of clauses \cite[Sect.~2]{letz:habil}: \textit{Select a node $N$ with an
  ancestor $N'$ such that both nodes are labeled with the same literal. Remove
  the edges originating\footnote{Tableau edges are considered as directed
  downward.} in the parent $N''$ of $N$ and replace them with the edges
  originating in $N$.}
Repeating this until the result is regular provides a polynomial proof
transformation procedure, or tableau simplification, to achieve regularity.
Any closed tableau for a given set~$F$ of clauses with a \emph{minimal} number
of nodes must be regular. Hence, enforcing regularity can be useful at proof
search. For interpolation, regularity simplification reduces the size of the
tableau if it is obtained from a prover that does not ensure regularity.

Regularity combined with the strong connection condition is complete
\cite{letz:diss,handbook:ar:haehnle:2001}, i.e., if there is a closed tableau for a
set of clauses, then there is one that is both regular and strongly connected.
However, with respect to size the interplay of both restrictions is not
smooth: clausal tableaux with the strong connection condition cannot
polynomially simulate clausal tableaux without that condition, and regular
clausal tableaux cannot simulate clausal tableaux with the strong connection
condition \cite[Sect.~3.4.1]{letz:diss} \cite[Chapter~7]{letz:habil}.

\subsection{The Hyper Property}
\label{sec-hyper}

The \name{hyper} property \cite{cw:range:2023} is a restriction of clausal
tableaux with applications in strengthened variations of Craig-Lyndon
interpolation and in the conversion of resolution proofs to clausal tableaux.
\begin{definition}[Hyper]
  A clausal tableau is \defname{hyper} iff the nodes labeled with a negative
  literal are exactly the leaf nodes.
\end{definition}
The name \name{hyper} alludes to hyperresolution and hypertableaux, which aim
at narrowing proof search through a related restriction.
Any closed clausal tableau for a set~$F$ of clauses can be converted to one
with the hyper property, as shown with the algorithm presented below
\cite{cw:range:2023}. Its specification involves a further tableau property,
\name{leaf-closed}, which, like regularity, can be achieved with a
straightforward simplification.

\begin{definition}[Target Node, Leaf-Closed]
A clausal tableau node is \defname{closing} iff it has an ancestor with
complementary literal. With a closing node~$N$, a particular such ancestor is
associated as \defname{target of}~$N$, written~$\ntgt{N}$. A tableau is
\defname{leaf-closing} iff all closing nodes are leaves. A closed tableau that
is leaf-closing is called \defname{leaf-closed}.
\end{definition}

\begin{algo}[Hyper Conversion]
\label{algo-proc-hyper}
\

\algoinput A leaf-closed and regular clausal tableau.

\algomethod Repeat the following operations until the resulting tableau is
hyper.
\begin{enumerate}
\item \label{step-proc-pick} Let $N^\prime$ be the first node visited in
  pre-order\footnote{\name{Pre-order tree traversal} is the method of
  depth-first traversal where the current node is visited \emph{before} its subtrees
  are recursively traversed left-to-right.} with a child that is an inner node
  with a negative literal label. Let $N$ be the leftmost such child.

\item Create a fresh copy~$U$ of the subtree rooted at $N^\prime$. In~$U$
  remove the edges that originate in the node corresponding to $N$.

\item \label{step-proc-attach} Replace the edges originating in $N^\prime$
  with the edges originating in $N$.

\item \label{step-proc-fix} For each leaf descendant~$M$ of $N^\prime$ with
  $\nlit{M} = \du{\nlit{N}}$: Create a fresh copy~$U^\prime$ of $U$. Change
  the origin of the edges originating in the root of $U^\prime$ to~$M$.

\item \label{step-proc-simp} Simplify the tableau to leaf-closing and regular
  form.
\end{enumerate}
\algooutput A leaf-closed, regular and hyper clausal tableau whose clauses are
clauses of the input tableau.
\end{algo}

\begin{example}
  The following tableaux show a conversion with
  Algorithm~\ref{algo-proc-hyper} in two steps.

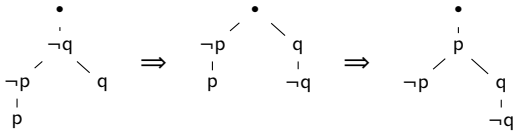
\vspace{-3mm}
  \begin{tikzpicture}[scale=0.8, %
    baseline=(a.north), sibling distance=4em,level distance=4ex, every
    node/.style = {transform shape,anchor=mid}]
  \node (a) {\vbar\textbullet}
  child { node {$\lnot \fq$}
    child { node {$\lnot \fp$}
      child { node {$\fp$} } }
    child { node {$\fq$} }
  };
\end{tikzpicture}
\raisebox{-7ex}{$\;\;\rewrite\;\;$}
\begin{tikzpicture}[scale=0.8, %
    baseline=(a.north), sibling distance=4em,level distance=4ex, every
    node/.style = {transform shape,anchor=mid}]
  \node (a) {\vbar\textbullet}
  child { node {$\lnot \fp$}
    child { node {$\fp$} }
  }
  child { node {$\fq$}
    child { node {$\lnot \fq$} }
  };
\end{tikzpicture}
\raisebox{-7ex}{$\;\;\rewrite\;\;$}
\begin{tikzpicture}[scale=0.8, %
    baseline=(a.north), sibling distance=4em,level distance=4ex, every
    node/.style = {transform shape,anchor=mid}]
  \node (a) {\vbar\textbullet}
  child { node {$\fp$}
    child { node {$\lnot \fp$}
    }
    child { node {$\fq$}
      child { node {$\lnot \fq$} }
    }
  };
\end{tikzpicture}
\lipicsEnd
\end{example}

The algorithm is specified by destructive tableau manipulations. A
\defname{fresh copy} of an ordered tree~$T$ is an ordered tree~$T^\prime$ with
fresh nodes and edges, related to $T$ through a bijection~$c$ such that any
node $N$ of $T$ has the same literal label as node $c(N)$ of $T^\prime$ and
such that the $i$-th edge originating in node $N$ of $T$ ends in node~$M$ iff
the $i$-th edge originating in node $c(N)$ of $T^\prime$ ends in node $c(M)$.
In each iteration the procedure chooses an inner node with negative literal
label and modifies the tableau. At termination the tableau is then hyper.
Since the procedure copies parts of subtrees it is not a polynomial operation
but practical usefulness has been demonstrated
\cite{cw:range:2023,cw:2024:synthesis}.

\subsection{Interpolant Calculation from a Clausal Tableau}
\label{sec-ground-ipol-tableaux}

The calculation of a ground Craig interpolant from a clausal tableau adapts
the propositional core of interpolation methods for sequent systems and
analytic tableaux
\cite{takeuti:book:1987,smullyan:book:1968,fitting:book:foar:2nd}
to the setting of clausal tableaux.
To calculate a Craig-Lyndon separator for sets $F, G$ of clauses we generalize
clausal tableaux by an additional node label, \name{side}, which is shared by
siblings and indicates whether a tableau clause is an instance of an input
clause in $F$ or in $G$.
\begin{definition}[\Sided Clausal Tableau and Related Notions]
  \

\subdefinition{def-coltab} Let $\FL, \GR$ be sets of clauses. A
\defname{\sided clausal tableau for} $\FL, \GR$ (briefly \defname{tableau for
  $\FL, \GR$}) is a clausal tableau for $\FL \cup \GR$ whose nodes~$N$ with
exception of the root are labeled additionally with a \name{side} $\nside{N}
\in \{\aaa, \bbb\}$, such that the following conditions are met by all nodes
$N, N'$: (1) If $N$ and $N^\prime$ are siblings, then $\nside{N} =
\nside{N^\prime}$; (2) If $N$ has a child $N'$ with $\nside{N^\prime} = \aaa$
($\nside{N^\prime} = \bbb$), then $\nclause{N}$ is an instance of a clause in
$\FL$ ($\GR$). We also refer to the side of the children of a node~$N$ as
\defname{side of} $\nclause{N}$.

\subdefinition{def-path} For $\xside \in \{\aaa,\bbb\}$ and all nodes $N$ of a
\sided clausal tableau define
\[\npath{\xside}{N}\; \eqdef
\{\nlit{N^\prime} \mid N^\prime \in \mathit{Path} \text{ and }
\nside{N^\prime} = \xside\},\] where
$\mathit{Path}$ is the union of $\{N\}$ and the set of the ancestors of $N$.
\end{definition}

\noindent
Example~\ref{examp-ct-ipol} in Sect.~\ref{sec-two-stage-overall} shows
examples of the defined notions. To integrate truth value simplifications
\eref{tv:at}--\eref{tv:of} into interpolant calculation from the very
beginning, we define the following variations of conjunction and disjunction:
For formulas $F_1, \ldots, F_n$ define $\bigwedgedot_{i=1}^n F_i$
($\bigveedot_{i=1}^n F_i$) as $\false$ ($\true$) if at least one of the
formulas $F_i$ is identical to $\false$ ($\true$), else as the conjunction
(disjunction) of those formulas $F_i$ that are not identical to $\true$
($\false$). We can now specify ground interpolant calculation inductively as a
function of tableau nodes.
\begin{definition}[Ground Interpolant Calculation from a Clausal Tableau]
Let $N$ be a node of a leaf-closed \sided clausal ground tableau. The value of
\defname{$\nipol{N}$} is a ground NNF formula, defined inductively as
specified with the tables below, the left for the base case where $N$ is a
leaf, the right for the case where $N$ is an inner node with children $N_1,
\ldots, N_n$.

\medskip

$\begin{array}[t]{c@{\hspace{1em}}c@{\hspace{1em}}c}
\nside{N} &  \nside{\ntgt{N}} & \nipol{N}\\\midrule
\aaa & \aaa & \false\\[0.5ex]
\aaa & \bbb & \nlit{N}\\[0.5ex]
\bbb & \aaa & \du{\nlit{N}}\\[0.5ex]
\bbb & \bbb & \true
\end{array}$
\hspace{1cm}
$\begin{array}[t]{c@{\hspace{1em}}c}
  \nside{N_1} & \nipol{N}\\\midrule
  \aaa & \bigveedot_{i=1}^{n} \nipol{N_i}\\[1ex]
  \bbb & \bigwedgedot_{i=1}^{n} \nipol{N_i}
\end{array}$
\end{definition}
The key property of $\nipol{N}$ is stated in the following lemma.

\pagebreak
\begin{lemma}[Invariant of Ground Interpolant Calculation from a Clausal
    Tableau]
  \label{lem-ground-ipol-invariant}
  Let $\FL, \GR$ be sets of ground clauses and let $N$ be a node of a
  leaf-closed \sided clausal ground tableau for $\FL, \GR$. It then holds
  that

\sublemma{lem-gi-sem}
$\FL \cup \npathL{N} \entails \nipol{N}\;$ and
$\;G \cup \npathR{N} \entails \lnot \nipol{N}$
 \sublemma{lem-gi-syn}
$\lit{\nipol{N}} \subseteq
\lit{\FL \cup \npathL{N}} \cap
\complit{G \cup \npathR{N}}$.
\end{lemma}
If $N_0$ is the tableau root, then $\npathL{N_0} = \npathR{N_0} = \{\}$.
Hence:
\begin{corollary}[Ground Interpolation with Clausal Tableaux]
\label{cor-ground-ipol-correct}
Let $\FL, \GR$ be sets of ground clauses and let $N_0$ be the root of a
leaf-closed \sided clausal ground tableau for $\FL,\GR$. Then $\nipol{N_0}$ is
a Craig-Lyndon interpolant for $\FL, \lnotGR$.
\end{corollary}

A ground interpolant according to Corollary~\ref{cor-ground-ipol-correct},
i.e., the value of $\nipolfun$ for the tableau root, is a ground NNF formula.
Due to the integrated truth value simplification it does not have truth value
constants as strict subformulas. The number of its literal occurrences is at
most the number of tableau leaves.

\subsection{The Two-Stage Approach: Overall Workflow}
\label{sec-two-stage-overall}

We illustrate the overall workflow of the two-stage approach to Craig
interpolation with Algorithm~\ref{algo-two-stage} below, assuming first-order
proving and ground interpolation is performed with a clausal tableau prover as
described above. In Sect.~\ref{sec-resolution} we will see that also a
resolution-based system can be used, either directly or supplemented with a
proof translation to clausal tableaux. Some subtasks can be performed in
alternate ways, possibly with substantial effect on success of the prover as
well as on size and shape of the resulting interpolant. The keyword
\OPTIONS\ introduces discussions of such alternatives.

\begin{algo}[Craig-Lyndon Interpolation for First-Order
    Logic with Clausal Tableaux]
  \label{algo-two-stage}
 \

 \algoinput First-order formulas $F, G$ such that $F \entails G$.

 \algooutput A Craig-Lyndon interpolant $H$ for $F,G$.

 \algomethod The algorithm proceeds in the following phases.
 \begin{enumerate}[I.]
 \item \label{step-fv-const} \textbf{Eliminating Free Variables.} Replace free
   variables in $F, G$ with fresh constants to obtain sentences $F^S, G^S$.
 \item \label{step-clausification} \textbf{Preprocessing: Skolemization, Clausification, Simplification.}
   Preprocess each of $F^S, \lnot G^S$ separately, as outlined in
   Sect.~\ref{sec-preprocessing} to obtain clause sets $F', G'$ with fresh
   Skolem functions $\ffs', \ggs'$. In case $F'$ ($G'$) contains the empty
   clause, exit with $H \eqdef \false$ ($H \eqdef \true$). \OPTIONS: Various
   forms of Skolemization; various ways of CNF conversion, including
   structure-preserving forms, and various simplifications, constrained by
   Prop.~\ref{prop-interpol-constraints}.
 \item \label{step-proving} \textbf{First-Order Proving.} Use a first-order
   prover to obtain a closed clausal tableau for $F' \cup G'$. \OPTIONS:
   Choice and configuration of the first-order prover.
 \item \label{step-grounding} \textbf{Proof Grounding.} Instantiate all
   variables in the tableau with ground terms.
   \OPTIONS: The choice of the terms
   for instantiating has effect on the interpolant calculated in
   phase~\ref{step-extract} and on the lifting in phase~\ref{step-lifting}. It
   is possible to use the same term for all variables, or different terms for
   different variables. Terms with \FONLY, \GONLY or \FGSHARED functions can
   be preferred.
 \item \label{step-side-assignment} \textbf{Side Assignment.} Attach side
   labels to the tableau clauses, according to whether they are an instance of
   a clause in $F'$ or in $G'$. \OPTIONS: It is possible that a clause in $F'$
   and a clause in $G'$ are identical or have common instances. Thus, there
   are cases where a tableau clause can be assigned either one of the two side
   labels. The choice can have effect on the ground interpolant calculated in
   phase~\ref{step-extract}.
 \item \label{step-extract} \textbf{Ground Interpolant Calculation.} Calculate
   the ground interpolant $\HG$ with the $\nipolfun$ function applied to the
   tableau. $\la F^S,G^S,\ffs,\ggs,\HG\ra$ is then a lifting base, where $\ffs,
   \ggs$ are determined from the Skolem functions $\ffs', \ggs'$ and
   $\fun{F^S}, \fun{G^S}$, as described after Def.~\ref{def-ib}. \OPTIONS:
   Incorporation of equivalence-preserving simplifications, including
   dedicated simplifications for equality, e.g., with $t = t \equiv \true$.
 \item \label{step-lifting} \textbf{Interpolant Lifting.} Let $H^S$ be the
   result of replacing terms in $\HG$ and adding a quantifier prefix
   according to Theorem~\ref{thm-lifting}
   with respect to the obtained lifting base. \OPTIONS:
   Theorem~\ref{thm-lifting} constrains the quantifier prefix by a partial
   order that may have different linear extensions.

 \item \label{step-constants-to-vars} \textbf{Reintroducing Free Variables.}
   Obtain the final result $H$ by replacing in $H^S$ any constants introduced
   in step~\ref{step-fv-const} with the corresponding free variables.
 \end{enumerate}
\end{algo}

\enlargethispage{5pt}

The following simple example illustrates the steps of the two-stage approach
to Craig interpolation with Algorithm~\ref{algo-two-stage}.

\begin{example}
  \label{examp-ct-ipol}
  Let $F \eqdef \forall x\, \fp(x) \land \forall x\, (\lnot \fp(x) \lor
  \fq(x))$ and let $G \eqdef \forall x\, (\lnot \fq(x) \lor \fr(x)) \imp
  \forall y\, \fr(y)$. Since these formulas do not have free variables, we
  skip steps~\ref{step-fv-const} and~\ref{step-constants-to-vars}.
  Step~\ref{step-clausification}, Skolemization and clausification applied
  separately to $F$ and to $\lnot G$, yields the first-order clause sets $F' =
  \{\fp(x),\; (\lnot \fp(x) \lor \fq(x))\}$ and $G' = \{(\lnot \fq(x) \lor
  \fr(x)),\; \lnot \fr(\fg)\}$, where $\fg$ is a Skolem constant. In
  step~\ref{step-proving}, first-order proving, the connection tableau
  calculus (Def.~\ref{def-calc-ser}) for start clause $\lnot \fr(\fg)$ yields
  the following leaf-closed tableau for $F' \cup G'$ (additional node
  annotations will be explained in a moment).

  
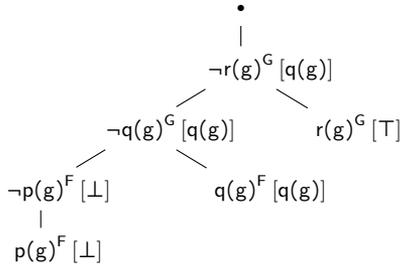
\begin{tikzpicture}[scale=0.83,
      sibling distance=9em,level distance=6.8ex,
      every node/.style = {transform shape,anchor=mid}]]
      \node (a) {\vbar\textbullet}
      child { node {$\tbbb{\lnot \fr(\fg)}$\nannot{$\fq(\fg)$}}
        child { node {$\tbbb{\lnot \fq(\fg)}$\nannot{$\fq(\fg)$}}
          child { node {$\taaa{\lnot \fp(\fg)}$\nannot{$\false$}}
            child { node {$\taaa{\fp(\fg)}$\nannot{$\false$}} }
          }
          child { node {$\taaa{\fq(\fg)}$\nannot{$\fq(\fg)$}} }
        }
        child { node {$\tbbb{\fr(\fg)}$\nannot{$\true$}}
        }
      };
  \end{tikzpicture}

  \noindent
  The calculus propagates the constant $\fg$ through unification into
  variables so that in case of the example the tableau is already ground,
  leaving nothing to do for step~\ref{step-grounding}, proof grounding.
  Step~\ref{step-side-assignment}, side assignment, leaves no options as each
  tableau clause is either an instance of a clause in $F'$ or of a clause in
  $G'$, but never of a clause in both. Thus, the clauses of the tableau with
  side $\aaa$ are $\fp(\fg)$ and $\lnot \fp(\fg) \lor \fq(\fg)$, and the
  clauses with side $\bbb$ are $\lnot \fq(\fg) \lor \fr(\fg)$ and $\lnot
  \fr(\fg)$. The respective side labels of the nodes are indicated as
  superscripts.
  To give examples for $\npath{\xside}{N}$, let $N$ be the bottom left node,
  which has literal $\fp(\fg)$. It then holds that $\npath{\aaa}{N} = \{\lnot
  \fp(\fg), \fp(\fg)\}$ and $\npath{\bbb}{N} = \{\lnot \fr(\fg), \lnot
  \fq(\fg)\}$.

  We now have a two-sided leaf-closed ground tableau for $F', G'$ and can, in
  step~\ref{step-extract}, calculate the ground interpolant $\HG$ with the
  $\nipolfun$ function. The values of $\nipolfun$ for the individual nodes are
  annotated in brackets. Its value for the root is $\HG = \fq(\fg)$.
  The tuple $\la F, G, \{\}, \{\fg\}, \fq(\fg) \ra$ forms an interpolant
  lifting base $\la F, G, \ffs, \ggs, \HG\ra$. This holds in general for
  $\la F, G, \ffs, \ggs, \HG\ra$ obtained with the described steps of
  Algorithm~\ref{algo-two-stage}. In our example, we can verify this with $\FE
  = \fp(u) \land (\lnot \fp(u) \lor \fq(u))$, $\GE = \lnot ((\lnot \fq(v) \lor
  \fr(v)) \land \lnot \fr(\fg))$ and $\sth = \{u \mapsto \fg,\; v \mapsto
  \fg\}$.
  Finally, in step~\ref{step-lifting}, interpolant lifting, we apply
  Theorem~\ref{thm-lifting} and obtain $H = \forall v_1\, \fq(v_1)$ as
  a Craig-Lyndon interpolant for $F,G$.
  \lipicsEnd
\end{example}

\section{Ground Interpolation with Resolution}
\label{sec-resolution}

We discuss resolution as a technique for first-order theorem proving and
ground interpolation in the two-stage approach, and relate it to clausal
tableaux in these roles.

\subsection{From a Deduction via a Deduction Tree to a Ground Deduction}

We consider a simple sound and complete first-order resolution calculus, which
we call $\RCALC$. It has the two following two inference rules.

\smallskip
    \begin{tabular}[t]{ll}
    \textbf{Binary Resolution} &
    $\begin{array}{cc}
      C \lor L & D \lor K\\\midrule
      \multicolumn{2}{c}{(C \lor D)\sigma}
    \end{array}$\\[2.5ex]
    \multicolumn{2}{p{6.5cm}}{where $\sigma$ is the most general unifier of
      $L$ and $\du{K}$}
    \end{tabular}
    \hspace{\fill}
    \begin{tabular}[t]{ll}
    \textbf{Factoring} &
    $\begin{array}{c}
      C \lor L \lor K\\\midrule
      (C \lor L)\sigma
    \end{array}$\\[2.5ex]
    \multicolumn{2}{p{6.5cm}}{where $\sigma$ is the most general unifier of
      $L$ and $K$}
  \end{tabular}

\medskip

\noindent
Clauses are considered here as \emph{multisets} of literals. It is assumed
that premises of binary resolution do not share variables, which is achieved
by renaming variables if necessary.
The conclusion of binary resolution is called \defname{resolvent}
\defname{upon} $L, K$, the conclusion of factoring \defname{factor with
  respect to} $L, K$.
The notion of \name{proof} is captured in the following definition.
\begin{definition}[Deduction]
Let $\ICALC$ be a calculus characterized by a set of inference rules for
clauses. An \defname{$\ICALC$-deduction} of a clause $C$ from a set $F$ of
clauses is a
sequence of clauses $C_1, \ldots, C_k = C$ such that each $C_i$ is either in
$F$ ($C_i$ is then called an \defname{input clause}), or the conclusion of an
inference rule of $\ICALC$ for premises preceding $C_i$. An $\ICALC$-deduction
of the empty clause $\false$ from $F$ is called
\defname{$\ICALC$-refutation}, or \defname{$\ICALC$-proof} of~$F$.
\end{definition}
The $\RCALC$-calculus is \name{sound}: if there is an $\RCALC$-refutation of
$F$, then $F$ is unsatisfiable, which follows since for both inference rules
the conclusion is entailed by the premises. The $\RCALC$-calculus is also
complete: if $F$ is unsatisfiable, then there is an $\RCALC$-refutation of $F$
\cite{chang:lee,resolution:handbook:2001}.

\begin{example}
  \label{examp-deduction}
  Let $F$ be the clause set $\{\lnot \fp(x) \lor \fp(\ff(x)),\; \fp(\fg(x)),\;
  \lnot \fp(\ff(\ff(\ff(\ff(\fg(x))))))\}$, which is unsatisfiable. The
  following table shows an $\RCALC$-deduction of $\false$ from $F$. We use
  $\ff^2(\_)$ as shorthand for $\ff(\ff(\_))$, and analogously
  $\ff^4(\_)$ for $\ff(\ff(\ff(\ff(\_))))$.

  \smallskip

  {\small
  \begin{tabular}[b]{lll}
    Clause-Id & Clause & Justification\\\midrule
    $C_1$ & $\lnot \fp(x) \lor \fp(\ff(x))$ & Input clause\\
    $C_2$ & $\fp(\fg(x))$ & Input clause\\
    $C_3$ & $\lnot \fp(\ff^4(\fg(x)))$ & Input clause\\
    $C_4$ & $\lnot \fp(x) \lor \fp(\ff^2(x))$ & Resolvent of $C_1$ and $C_1$\\
    $C_5$ & $\lnot \fp(x) \lor \fp(\ff^4(x))$ & Resolvent of $C_4$ and $C_4$\\
    $C_6$ & $\fp(\ff^4(\fg(x)))$ & Resolvent of $C_2$ and $C_5$\\
    $C_7$ & $\false$ & Resolvent of $C_6$ and $C_3$
  \end{tabular}}
\lipicsEnd
\end{example}

Binary resolution and factoring can be combined into a single inference rule,
as in the original presentation of resolution \cite{robinson:1965}. Numerous
refinements of resolution aim at improving proof search by reducing the vast
number of deductions that can be generated.
Interpolant calculation, however, starts from a \emph{given proof}. Many
resolution refinements can be translated to our two basic rules. \ProverN
actually comes with a tool to convert its proofs to basic rules. Thus,
interpolant calculation for just a basic form of resolution does not exclude
employing advanced resolution refinements at \emph{proof search}.

Our interpolant calculation operates on a resolution deduction that has only
ground clauses. It is obtained via expanding the given first-order
$\RCALC$-deduction into a deduction \emph{tree}.
\begin{definition}[Deduction Tree]
  \label{def-r-deduction-tree}
  An \defname{$\ICALC$-deduction tree} of a clause $C$ is an (upward) tree with
  nodes labeled by clauses, such that the clause of the root is $C$ and the
  clause of any inner node is the conclusion of an inference rule of $\ICALC$
  from the clauses of the parents as premises.
\end{definition}
Variables in a deduction tree have global scope, as in a clausal tableau, and
differently from a deduction. Since the deduction tree expands the DAG
structure of the premise/conclusion relationship represented by the deduction
into a tree, the number of nodes of the deduction tree may be exponentially
larger than the number of clauses of the deduction.
The construction of a deduction tree from a given deduction can be sketched as
follows. We start with creating the leaf nodes of the deduction tree, one leaf
node for each \emph{instance} of an input clause in the deduction, each leaf
node with a \emph{fresh copy} of the respective input clause such that no
variables are shared between leaves. Then we proceed downwards to the root by
attaching children according to the inferences in the deduction, however, now
without renaming variables before applying inference rules. Instead, we
compute the most general unifier and apply it to all variables in the tree
under construction, such that also occurrences in ancestor nodes are
substituted.\footnote{We assume here w.l.o.g. that for the most general
unifier its domain and the set of variables occurring in its range are
disjoint subsets of the set of variables in the unified terms
\cite[Rem.~4.2]{eder:subst:1985}.}

\begin{example}
  \label{examp-deduction-tree}
  The $\RCALC$-deduction of Example~\ref{examp-deduction} expands into the
  following $\RCALC$-deduction tree. We may assume that variable $y$ stems
  from the sole involved copy of $C_2$.

  \medskip

  \noindent
  \scalebox{0.78}
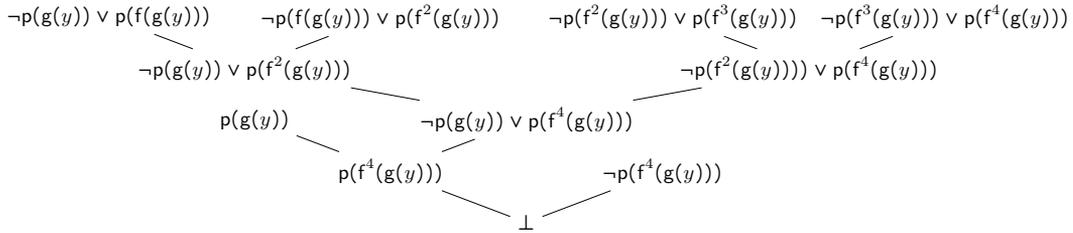
{
  \begin{tikzpicture}[scale=1.0,
      baseline=(a.north),sibling distance=13em,level distance=2.5em]
    \node (a) {$\false$} [grow'=up,sibling distance=13em]
    child {node {$\fp(\ff^4(\fg(y)))$}
      child {node {$\fp(\fg(y))$}}
      child {node {$\lnot \fp(\fg(y)) \lor \fp(\ff^4(\fg(y)))$}
        [sibling distance=27em]
        child {node {$\lnot \fp(\fg(y)) \lor \fp(\ff^2(\fg(y)))$}
          [sibling distance=13em]
          child {node {$\lnot \fp(\fg(y)) \lor \fp(\ff(\fg(y)))$}}
          child {node {$\lnot \fp(\ff(\fg(y))) \lor \fp(\ff^2(\fg(y)))$}}}
        child {node {$\lnot \fp(\ff^2(\fg(y)))) \lor
            \fp(\ff^4(\fg(y)))$}
          [sibling distance=13em]
          child {node {$\lnot \fp(\ff^2(\fg(y))) \lor \fp(\ff^3(\fg(y)))$}}
          child {node {$\lnot \fp(\ff^3(\fg(y))) \lor \fp(\ff^4(\fg(y)))$}}
          }}
    }
    child {node {$\lnot \fp(\ff^4(\fg(y)))$}
    };
  \end{tikzpicture}}

\vspace{-8pt}  
\lipicsEnd
\end{example}

We now move to \emph{ground} resolution, with the ground resolution calculus
$\GRCALC$ that operates on sets of ground clauses. It has the following two
inference rules.

\smallskip

    \begin{tabular}[t]{ll}
    \textbf{Ground Resolution} &
    $\begin{array}{cc}
      C \lor L & D \lor \du{L}\\\midrule
      \multicolumn{2}{c}{C \lor D}
    \end{array}$
    \end{tabular}
    \hspace{2em}
    \begin{tabular}[t]{ll}
    \textbf{Merging} &
    $\begin{array}{c}
      C \lor L \lor L\\\midrule
      C \lor L
    \end{array}$
  \end{tabular}

\medskip

\noindent The conclusion of merging is called \defname{merge with respect to}
$L$. $\GRCALC$- and $\RCALC$-deductions can be related via their
deduction tree as follows.
\begin{definition}[$\GRCALC$-Deduction for an $\RCALC$-Deduction]
  \label{def-rgr-conversion}
  Let $F$ be a set of clauses. An $\GRCALC$-deduction~$\DED$ from a set
  $F_{\textsc{grd}}$ of ground clauses is said to be \defname{for} an
  $\RCALC$-deduction~$\DEDR$ from a set $F$ of clauses if $\DED$ can be
  obtained from $\DEDR$ by expansion to an $\RCALC$-Deduction tree,
  instantiating all variables in this $\RCALC$-deduction tree by ground terms,
  which yields a $\GRCALC$-deduction tree,
  followed by converting this $\GRCALC$-deduction tree
  into a $\GRCALC$-deduction.
\end{definition}

Definition~\ref{def-rgr-conversion} is expressed in terms of transformations,
which for practical interpolation have to be implemented. For instantiating
with ground terms there are options as discussed for
phase~\ref{step-grounding} of Algorithm~\ref{algo-two-stage}. The final
conversion of the $\GRCALC$-deduction tree to a $\GRCALC$-deduction may be
just a linearizing of the tree structure. In case different nodes are labeled
with identical ground clauses, it might be shortened to a representation as
DAG.

\subsection{Interpolant Calculation from a Ground Deduction}

To compute a Craig-Lyndon interpolant for first-order formulas $F,G$ such that
$F \entails G$ on the basis of a resolution proof we have to build an
interpolant lifting base $\la F, G, \ffs, \ggs, \HG \ra$.
\begin{proposition}[Deduction and Lifting Base]
  \label{prop-lifting-ded}
Let $F, G$ be formulas such that $F \entails G$. Let $F',G'$ are obtained from
$F, \lnot G$ by prenexing, Skolemization, and CNF transformation. Let $\DED$
be an $\GRCALC$-proof for an $\RCALC$-proof of $F' \cup G'$. Let $\ffs$
($\ggs$) be the union of the Skolem functions $\ffs'$ ($\ggs'$) introduced for
$F$ ($G$) and the \FONLY (\GONLY) functions. If, so far, neither in $\ffs$,
nor in $\ggs$, nor in the $\FGSHARED$ functions is a constant, then add a
fresh constant $c_0$ to either $\ffs$ or $\ggs$. Let $\FGRD$ ($\GGRD$) be the
input clauses of $\DED$ that are instances of $F'$ ($G'$). If $\HG$ is a
Craig-Lyndon interpolant of $\FGRD, \lnot \GGRD$, then $\la F, G, \ffs, \ggs,
\HG\ra$ is an interpolant lifting base.
\end{proposition}
The construction of $\ffs$ and $\ggs$ is immediate from
Prop.~\ref{prop-lifting-ded}. To complete the lifting base we calculate a
ground interpolant $\HG$ from the ground deduction $\DED$. To this end we
enhance literal occurrences in proofs with a \name{provenance} label,
analogous to the \name{side} in clausal tableaux.

\begin{definition}[Provenance-Labeling and Related Notions]
  \

  \subdefinition{def-provclause} A clause (considered as multiset of literals)
  is \defname{provenance-labeled} if each of its literals is associated with a
  \defname{provenance label}, a nonempty subset of $\{\FF, \GG\}$. A literal
  $L$ with provenance label $\prov$ is written $L^{\prov}$.

  \subdefinition{def-provresol} A \defname{provenance-labeled ground
    resolution proof} for sets $F, G$ of ground clauses is a ground resolution
  proof for $F \cup G$, where each literal in an input clause from $F$ has
  provenance label $\{\FF\}$, each literal in an input clause from $G$ has
  provenance label $\{\GG\}$, and provenance labels are propagated as follows:
  For a resolvent $C \lor D$ the provenance labels are taken from $C, D$ in
  the premises; for a merge $C \lor L$, the provenance label of $L$ is the
  union of the provenance labels of the two merged occurrences of $L$ in the
  premise, and the provenance labels of $C$ are taken from $C$ in the premise.

  \subdefinition{def-subclause} For $\mathit{\prov} \in \{\FF, \GG\}$ and
  provenance-labeled clause~$C = \bigvee_i^n L_i^{\prov_i}$ define
  \[\rsubclause{\mathit{\prov}}{C} \eqdef \bigvee_{\mathit{\prov} \in \prov_i} L_i.\]

\end{definition}

\begin{example}
  Let $C = p^{\{\FF\}} \lor q^{\{\GG\}} \lor r^{\{\FF,\GG\}} \lor
  s^{\{\FF\}} \lor s^{\{\GG\}}$, then $\rsubclause{\FF}{C} = p \lor r \lor s$
  and $\rsubclause{\GG}{C} = q \lor r \lor s$.
\lipicsEnd
\end{example}

\begin{definition}[Ground Interpolant Calculation from a Resolution Proof]
  \label{def-ripol}
  Let $C$ be a clause in a provenance-labeled ground resolution proof for
  clausal ground formulas $F, G$. The value of \defname{$\ripol(C)$} is a
  ground NNF formula, defined inductively as follows. For the base cases where
  $C$ is an input clause and the case where $C$ is obtained by merging, the
  value of $\ripol(C)$ is specified in the following table.

  \smallskip
  $\begin{array}{ll}
    \text{Case} & \ripol(C)\\\midrule
    C \text{ is an input clause } \bigvee_{i = 1}^n L_i^{\{\FF\}} \text{ from } F & \false\\
    C \text{ is an input clause } \bigvee_{i = 1}^n L_i^{\{\GG\}} \text{ from } G & \true\\
    C \text{ is obtained as merge from premise } D & \ripol(D)\\
  \end{array}$
  \smallskip

  \noindent
  For the case where $C \lor D$ is obtained as resolvent from premises $C \lor
  L^{\prova}$ and $D \lor \duh{L}^{\provb}$, the value of $\ripol(C \lor D)$ is
  specified in the following table. It depends on the provenance labels
  $\prova,\provb$ of the literals resolved upon. We use the shorthands $H_1 =
  \ripol(C \lor L^{\prova})$ and $H_2 = \ripol(D \lor \duh{L}^{\provb})$. For
  two of the subcases, alternate possibilities are given.

  \smallskip
  $\begin{array}
    {ccll}
    \prova & \provb & \ripol(C \lor D) & \text{Remark}\\\midrule
    \{\FF\} & \{\FF\} & H_1 \lor H_2\\
    \{\FF\} & \{\GG\} & H_1 \lor (L \land H_2)\\
    \{\FF\} & \{\GG\} & (L \lor H_1) \land H_2 & \text{Alternate possibility}\\
    \{\FF\} & \{\FF, \GG\} & H_1 \lor (L \land H_2)\\
    \{\GG\} & \{\GG\} & H_1 \land H_2\\
    \{\GG\} & \{\FF, \GG\} & H_1 \land (\du{L} \lor H_2)\\
    \{\FF, \GG\} & \{\FF, \GG\} & (\du{L} \land H_1) \lor (L \land H_2)\\
    \{\FF, \GG\} & \{\FF, \GG\} & (L \lor H_1) \land (\du{L} \lor H_2) &
    \text{Alternate equivalent possibility}\\
  \end{array}$
  \medskip
\end{definition}

\noindent
The key property of $\ripol(C)$ that holds for all clauses $C$ of the
provenance-labeled resolution proof is stated in the following lemma.

\begin{lemma}[Invariant of Ground Interpolant Calculation from a Resolution Proof]
  \label{lem-ground-ripol-invariant}
  Let $F, G$ be sets of ground clauses and let $C$ be a provenance-labeled
  clause of a resolution proof for $F, G$. It then holds that
  \begin{enumerate}[(i)]
  \item \label{item-gri-f} $F \entails \ripol(C) \lor \rsubclause{\FF}{C}\;$ and
    $\;G \entails \lnot \ripol(C) \lor \rsubclause{\GG}{C}$.
  \item \label{item-gri-lit} $\lit{\ripol(C)} \subseteq \lit{F} \cap \complit{G}$.
  \end{enumerate}
\end{lemma}

If $C_0$ is the empty clause $\false$ at the root of the resolution proof,
then $\rsubclause{\FF}{C_0} = \rsubclause{\GG}{C_0} = \false$. Hence:
\begin{corollary}[Ground Interpolation with Resolution]
\label{cor-ground-ripol-correct}
Let $F, G$ be clausal ground formulas and let $C_0$ be the empty clause
obtained with a provenance-labeled ground resolution proof from $F, G$. Then
$\ripol(C_0)$ is a Craig-Lyndon interpolant for $F, \lnot G$.
\end{corollary}

\label{sec-hkpym}

Definition~\ref{def-ripol} refines a well known interpolation system for
propositional resolution called \name{HKPYM} by Bonacina and Johansson
\cite{bonacina:2015:ground}, after the initials of several authors who
discovered or investigated it independently. Huang \cite{huang:95} uses it for
first-order ground resolution proofs, like we do, but extended to
paramodulation. He assumes merging implicitly with ground resolution. The
essential difference is that Huang uses for \emph{all} cases of ground
resolution with exception of $\{\FF\} \{\FF\}$ and $\{\GG\} \{\GG\}$ the
schema of our case $\{\FF, \GG\} \{\FF, \GG\}$, that is, $(\du{L} \land H_1)
\lor (L \land H_2)$. Similarly, \HKPYM as defined in
\cite{bonacina:2015:ground} uses for \emph{all} these cases our
\name{alternate possibility} for case $\{\FF, \GG\} \{\FF, \GG\}$. Hence,
these methods do not construct Craig-\emph{Lyndon} interpolants. Although
Huang uses provenance to label literal \emph{occurrences in the proof}, like
we do, versions of \HKPYM often use a labeling that just distinguishes on the
basis of the two given interpolated formulas $F,G$ whether an atom is
\FGSHARED (is ``transparent'' or ``grey''), \FONLY (has one of two
``colors''), or \GONLY (has the other ``color''). With this coarse labeling
more literals may enter the interpolant, instead of truth value constants that
could be eliminated by simplifying with \eref{tv:at}--\eref{tv:nt}.

\subsection{Deduction Trees and Clausal Tableaux}
\label{sec-atomic-cut}

It is well known \cite{letz:habil} that a resolution deduction tree of the
empty clause $\false$ represents a closed clausal tableaux in a specific form
and vice versa. This correspondence is of interest for interpolation since it
can be utilized in practice to convert resolution proofs to clausal tableaux
and, moreover, indicates a systematization of resolution-based interpolation
methods.
\begin{definition}[Atomic Cut, Clausal Tableau in Cut Normal Form]
  \

  \subdefinition{def-atomic-cut}
  An \defname{atomic cut} is a clause of the form
  $\lnot p(\tts) \lor p(\tts)$.

  \subdefinition{ded-cut-normal-form} A closed clausal tableau is in
  \defname{cut normal form} for a set of clauses $F$ if for all inner
  nodes $N$ whose children are not leaves $\nclause{N}$ is an
  atomic cut, and for all inner nodes $N$ whose children are leaves
  $\nclause{N}$ is an instance of a clause in $F$.
\end{definition}

\begin{example}
  The following closed clausal tableau is in cut normal form. It represents
  the $\RCALC$-deduction tree from Example~\ref{examp-deduction-tree}.

  \vspace{-5pt}
  \noindent\scalebox{0.78}{
  \begin{tikzpicture}[scale=1.0,
      baseline=(a.north),sibling distance=20em]
    \coordinate (a) [sibling distance=21em,level distance=2.5em]
    child {node {$\lnot\fp(\ff^4(\fg(y)))$}
      [sibling distance=21em]
      child {node {$\lnot\fp(\fg(y))$}
        child {node {\cutleaf{$\lnot\fp(\fg(y))$}}}}
      child {node {$\fp(\fg(y))$}
        [sibling distance=27em]
        child {node {$\lnot\fp(\ff^2(\fg(y)))$}
          [sibling distance=13em]
          child {node {$\lnot\fp(\ff(\fg(y)))$}
            [sibling distance=6.5em]
            child {node {\cutleaf{$\lnot \fp(\fg(y))$}}}
            child {node {\cutleaf{$\fp(\ff(\fg(y)))$}}}}
          child {node {$\fp(\ff(\fg(y)))$}
            [sibling distance=6.5em]
            child {node {\cutleaf{$\lnot \fp(\ff(\fg(y)))$}}}
            child {node {\cutleaf{$\fp(\ff^2(\fg(y)))$}}}
          }}
        child {node {$\fp(\ff^2(\fg(y))))$}
          [sibling distance=12.5em]
          child {node {$\lnot\fp(\ff^3(\fg(y)))$}
            [sibling distance=6em]
            child {node {\cutleaf{$\lnot \fp(\ff^2(\fg(y)))$}}}
            child {node {\cutleaf{$\fp(\ff^3(\fg(y)))$}}}
          }
          child {node {$\fp(\ff^3(\fg(y)))))$}
            [sibling distance=6em]
            child {node {\cutleaf{$\lnot \fp(\ff^3(\fg(y)))$}}}
            child {node {\cutleaf{$\fp(\ff^4(\fg(y)))$}}}
          }
          }}
    }
    child {node {$\fp(\ff^4(\fg(y)))$}
        child {node {\cutleaf{$\lnot \fp(\ff^4(\fg(y)))$}}}
    };
  \end{tikzpicture}}

\vspace{-5pt}  
\lipicsEnd
\end{example}

A clausal tableau in cut normal form for $F$ is a special case of a clausal
tableau \emph{for} the formula $F \cup \bigcup_{p \in \predplain{F}}{\{\lnot
  p(\xs) \lor p(\xs)\}}$, which is equivalent to $F$. Cut normal form can be
seen as a notational variant of semantic trees \cite{chang:lee,letz:habil}.
A given deduction tree of the empty clause~$\false$ from a set of clauses $F$
can be converted in linear time to a closed clausal tableau in cut normal form
for $F$ as follows: (1) Delete the factoring or merging steps (considering
factoring/merging as integrated into the resolution rule).
(2)~Remove the root label $\false$. (3)~Replace the labels representing the
premises $C \lor L$ and $C \lor \du{L}$ of a resolution step with the
complements of the literals resolved upon, $\du{L}$ and $L$, respectively.
(4)~Turn the tree upside down, such that the root is now at the top. (5)~At
each leaf that was labeled by an instance $L_1 \lor \ldots L_n \lor L$ of an
input clause, and is now labeled by $L$, attach that clause, i.e., attach
children labeled by $L_1, \ldots, L_n, L$.

In the resulting tableau, each involved instance of an input clause is
falsified by the branch leading to it. That is, each literal in the clause has
a complement in the branch. The converse translation of a closed clausal
tableaux in cut normal form to a deduction tree of the empty clause $\false$
is straightforward, with a potential quadratic increase in size because
intermediate resolvents have to be attached \cite{letz:habil}. Expressed in
terms of semantic trees, this converse translation underlies a classic
completeness proof of resolution \cite{chang:lee}.

The \name{hyper} property (Sect.~\ref{sec-hyper}) is, under assumption of
regularity and leaf-closedness, incompatible with presence of an atomic cut:
Consider an atomic cut $\lnot A \lor A$ as tableau clause. If the tableau is
hyper, then the node labeled with $\lnot A$ is a leaf that is closed by an
ancestor with label $A$. But this ancestor is also an ancestor of the other
node of the atomic cut, labeled with $A$, which violates regularity.
Thus, Algorithm~\ref{algo-proc-hyper}, which converts a clausal tableau to
hyper form, when applied to a closed clausal tableau in cut normal form yields
a closed clausal tableau without atomic cuts. All tableau clauses of the
converted tableau are instances of the input clauses at the leaves of the
tableau in cut normal form. The hyper conversion thus ``eliminates'' the
atomic cuts, which can be practically applied to convert resolution proofs to
clausal tableaux without atomic cuts \cite{cw:range:2023}.

Since a closed clausal tableau in cut normal form is a special case of a
closed clausal tableau we can also use it directly for interpolant calculation
with the $\nipolfun$ operator. How should the side labeling be chosen? The
clause instances at the leaves evidently receive the side label of the
original clause of which they are an instance. The atomic cuts where the
predicate is \FONLY get side~$\aaa$, the atomic cuts where the predicate is
\GONLY get side~$\bbb$. For atomic cuts where the predicate is \FGSHARED we
can take either side or we can stack instances with both side labels upon each
other, leading the following eight combinations.

\smallskip
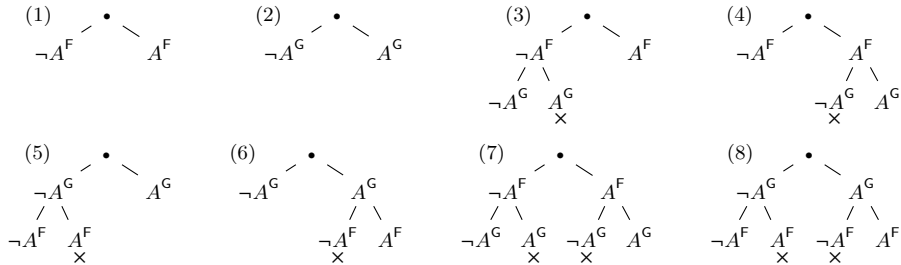
\bgroup
\setlength{\tabcolsep}{8pt}
\begin{tabular}{rrrr}
\footnotesize\begin{tikzpicture}[scale=0.9,baseline=(a.south),
    every node/.style = {transform shape,anchor=north},
    level distance=0.8em, sibling distance=5em]
  \node[] at (-1,0) {(1)};
  \node (a) {\vbar\textbullet}
  child { node {$\lnot A^\aaa$}}
  child { node {$A^\aaa$} };
\end{tikzpicture}
&
\footnotesize\begin{tikzpicture}[scale=0.9,baseline=(a.south),
    every node/.style = {transform shape,anchor=north},
    level distance=0.8em, sibling distance=5em]
  \node[] at (-1,0) {(2)};
  \node (a) {\vbar\textbullet}
  child { node {$\lnot A^\bbb$}}
  child { node {$A^\bbb$} };
\end{tikzpicture}
&
\footnotesize\begin{tikzpicture}[scale=0.9,baseline=(a.south),
    every node/.style = {transform shape,anchor=north},
    level distance=0.8em, sibling distance=5em]
    \node[] at (-1,0) {(3)};
    \node (a) {\vbar\textbullet}
    child { node {$\lnot A^\aaa$}
      [sibling distance=2.5em,level distance=1.5em]
      child { node {$\lnot A^\bbb$} }
      child { node [align=center]{$A^\bbb$\\[-3pt]$\times$}} }
    child { node {$A^\aaa$} };
\end{tikzpicture}
&
\footnotesize\begin{tikzpicture}[scale=0.9,baseline=(a.south),
      every node/.style = {transform shape,anchor=north},
      level distance=0.8em, sibling distance=5em]
    \node[] at (-1,0) {(4)};
    \node (a) {\vbar\textbullet}
    child { node {$\lnot A^\aaa$}}
    child { node {$A^\aaa$}
      [sibling distance=2.5em,level distance=1.5em]
      child { node [align=center]{$\lnot A^\bbb$\\[-3pt]$\times$}}
      child { node {$A^\bbb$} }
    };
  \end{tikzpicture}\\[1.25cm]
\footnotesize\begin{tikzpicture}[scale=0.9,baseline=(a.south),
      every node/.style = {transform shape,anchor=north},
      level distance=0.8em, sibling distance=5em]
    \node[] at (-1,0) {(5)};
    \node (a) {\vbar\textbullet}
    child { node {$\lnot A^\bbb$}
      [sibling distance=2.5em,level distance=1.5em]
      child { node {$\lnot A^\aaa$} }
      child { node [align=center]{$A^\aaa$\\[-3pt]$\times$}} }
    child { node {$A^\bbb$} };
  \end{tikzpicture}
&
\footnotesize\begin{tikzpicture}[scale=0.9,baseline=(a.south),
    every node/.style = {transform shape,anchor=north},
    level distance=0.8em, sibling distance=5em]
      \node[] at (-1,0) {(6)};
      \node (a) {\vbar\textbullet}
      child { node {$\lnot A^\bbb$}}
      child { node {$A^\bbb$}
        [sibling distance=2.5em,level distance=1.5em]
        child { node [align=center]{$\lnot A^\aaa$\\[-3pt]$\times$}}
        child { node {$A^\aaa$} }
      };
\end{tikzpicture}
&
\footnotesize\begin{tikzpicture}[scale=0.9,baseline=(a.south),
    every node/.style = {transform shape,anchor=north},
    level distance=0.8em, sibling distance=5em]
      \node[] at (-1,0) {(7)};
      \node (a) {\vbar\textbullet}
      child { node {$\lnot A^\aaa$}
        [sibling distance=2.5em,level distance=1.5em]
        child { node {$\lnot A^\bbb$} }
        child { node [align=center]{$A^\bbb$\\[-3pt]$\times$}} }
      child { node {$A^\aaa$}
        [sibling distance=2.5em,level distance=1.5em]
        child { node [align=center]{$\lnot A^\bbb$\\[-3pt]$\times$}}
        child { node {$A^\bbb$} }
      };
\end{tikzpicture}
&
\footnotesize\begin{tikzpicture}[scale=0.9,baseline=(a.south),
    every node/.style = {transform shape,anchor=north},
    level distance=0.8em, sibling distance=5em]
      \node[] at (-1,0) {(8)};
      \node (a) {\vbar\textbullet}
      child { node {$\lnot A^\bbb$}
        [sibling distance=2.5em,level distance=1.5em]
        child { node {$\lnot A^\aaa$} }
        child { node [align=center]{$A^\aaa$\\[-3pt]$\times$}} }
      child { node {$A^\bbb$}
        [sibling distance=2.5em,level distance=1.5em]
        child { node [align=center]{$\lnot A^\aaa$\\[-3pt]$\times$}}
        child { node {$A^\aaa$} }
      };
\end{tikzpicture}
\end{tabular}
\egroup

\smallskip

Some branches occurring at stacking are immediately closed, corresponding to a
literal in the interpolant calculated by $\nipolfun$. Each combination has
exactly two open branches, such that the stacking effects no substantial
increase of the tree size.
Depending on the employed stacking schemas, interpolant calculation with
$\nipolfun$ on the clausal tableau in cut normal form simulates different
resolution-based calculi for ground interpolation, modulo commutativity of
$\land$ and $\lor$ and truth value simplification. With stacking according to
schema~(7) we obtain Huang's method, with schema~(8) \HKPYM. We assume
here that as target for closing a branch with a leaf from an instance of an
input clause the node with the same side label as the leaf is selected.
Schema~(2) gives McMillan's method
\cite{mcmillan:2003,mcmillan:2005,bonacina:2015:ground}. For a more detailed
exposition and examples see \cite{cw:ipol}. Our calculation by $\ripol$ is
obtained if the stacking schema is chosen according to the provenance label of
the literal occurrences upon which the resolution step is performed. In case
both are labeled by $\{\FF, \GG\}$, schema~(7) is chosen. Schema~(8)
corresponds to the alternate possibility for this case in
Def.~\ref{def-ripol}. For the other cases, $H_1 \lor (L \land H_2)$ is
simulated for positive (negative)~$L$ by schema~(4) (schema~(3)), and $(L \lor
H_1) \land H_2$ for positive (negative)~$L$ by schema~(5) (schema~(6)).

\section{Craig-Lyndon Interpolation and Equality}
\label{sec-equality}

Adding equality axioms, e.g., those from Table~\ref{tab-rst}, is a simple way
to incorporate equality into first-order logic. For provers with no dedicated
equality support, adding such axioms is common practice. Equality-specific
inferences often can be translated into inferences of an equality-free
calculus, if the input is enriched by equality axioms.
Paramodulation \cite{paramodulation:1969}
provides an example. The same holds for superposition rules
\cite{bachm:ganz:92} since they are restrictions of paramodulation.
Thus, a practical workflow for incorporating equality into interpolation with
first-order provers is performing the proof search with dedicated equality
support, followed by translating the proof to an equality-free calculus with
axiomatized equality. Interpolant calculation is then applied to the
translated proof, with equality handled as a predicate.

\begin{table}
  \caption{A first-order axiomatization of equality, shown as clauses.
    Substitutivity axioms \axname{SubstPred}$_{p,i}$ and
    \axname{SubstFun}$_{f,i}$ are for each predicate $p$ (function $f$) with
    arity $n > 0$ in the vocabulary, and for each argument position $i \in
    \{1,\ldots,n\}$.}
  \label{tab-rst}
\vspace{-5pt}
  \begin{tabular}{ll}
    \axname{Reflexivity} & $x = x$\\
    \axname{Symmetry} & $x \neq y \lor y = x$\\
    \axname{Transitivity} & $x \neq y \lor y \neq z \lor x = z$\\
    \axname{SubstPred}$_{p,i}$ &
    $\lnot p(x_1,\ldots,x_{i-1},x,x_{i+1},\ldots,x_n) \lor
    x \neq y \lor p(x_1,\ldots,x_{i-1},y,x_{i+1},\ldots,x_n)$\\
    \axname{SubstFun}$_{f,i}$ &
    $x \neq y \lor f(x_1,\ldots,x_{i-1},x,x_{i+1},\ldots,x_n) =
    f(x_1,\ldots,x_{i-1},y,x_{i+1},\ldots,x_n)$\\
  \end{tabular}
\vspace{-6pt}
\end{table}

A Craig-Lyndon interpolant for formulas $F, G$ of first-order logic with
equality is then a Craig-Lyndon interpolant for formulas $E_F \land F,\;
E_G \imp G$ of first-order logic without equality, where $E_F$ and $E_G$ are
conjunctions of equality axioms, say from Table~\ref{tab-rst}. Axioms
\axname{SubstPred}$_{p,i}$ are placed in $E_F$ ($E_G$) if $p$ is an \FONLY
(\GONLY) predicate. In these axioms $=$ occurs only in negative polarity. The
other equality axioms, which involve $=$ positively, can be placed arbitrarily
in $E_F$ or $E_G$. By controlling their placement, Craig-Lyndon interpolation
yields interpolants according to the following theorem, due to Oberschelp
\cite{oberschelp:1968} (see also \cite{motohashi:1984}).
\begin{theorem}[Oberschelp Interpolation]
  \label{thm-oberschelp}
  Let $F, G$ be formulas of first-order logic with equality such that $F
  \entails G$. Then there exists a formula $H$ of first-order logic with
  equality such that
  \begin{enumerate}[(1)]
  \item $F \entails H$ and $H \entails G$.
  \item $\pred{H} \subseteq \pred{F} \cap \pred{G}$, $\var{H} \subseteq
    \var{F} \cap \var{G}$,
    and $\const{H} \subseteq \const{F} \cap \const{G}$.
  \item $\fun{H} \subseteq \fun{F} \cup \fun{G}$.
  \item \label{item-oberschelp-fg} If $=$ occurs positively (negatively) in $H$,
    then $=$ occurs positively (negatively) in $F$ ($G$).
  \end{enumerate}
\end{theorem}

\begin{proof}
  We obtain $H$ as Craig-Lyndon interpolant for $E_F \land F,\; E_G \imp G$,
  where $E_F$ and $E_G$ are defined as follows, considering
  conditions~(\ref{item-oberschelp-fg}) contrapositively. If $=$ occurs only
  negatively or not at all in $F$ (only positively or not at all in $G$), then
  $E_F$ ($E_G$) is the conjunction of the axioms \axname{SubstPred}$_{p,i}$
  for the \FONLY (\GONLY) predicates~$p$, and $E_G$ ($E_F$) is the conjunction
  of the remaining equality axioms for the vocabularies of~$F, G$. Else $=$
  occurs positively in $F$ or negatively in $G$. In this case let $E_F$
  ($E_G$) include the substitutivity axioms for \FONLY (\GONLY) predicates and
  functions, and place the remaining equality axioms for the vocabularies of
  $F, G$ arbitrarily in $E_F$ or $E_G$.
\end{proof}

Theorem~\ref{thm-oberschelp} strengthens Craig-Lyndon interpolation only for
formulas without functions except of constants. As shown by Brand in the
context of his \name{modification method} \cite{brand:75} (see
\cite{equality:sequent:handbook:ar} for a summary and further references), the
substitutivity axioms are dispensable for a set of clauses that are
\defname{flat}, i.e., all occurrences of non-variable terms are arguments to
the equality predicate~$=$. Any clause can be converted to an equivalent flat
clause through ``pulling out'' terms with \eref{t:pos}. For example,
$\fg(\fa,\fb) = \fb$ is equivalent to the flat clause $\fb \neq x \lor \fa
\neq z \lor \fg(z,x) = x$. Conversion to flat form introduces $=$ only with
\emph{negative} polarity. This can be utilized to show the following theorem,
due to Fujiwara \cite{fujiwara:1978} and proven also by Motohashi
\cite{motohashi:1984}, which strengthens both Craig-Lyndon interpolation and
Theorem~\ref{thm-oberschelp}.

\begin{theorem}[Oberschelp-Fujiwara Interpolation]
  \label{thm-oberschelp-fujiwara}
  Let $F, G$ be formulas of first-order logic with equality such that $F
  \entails G$. Then there exists a formula $H$ of first-order logic with
  equality such that
  \begin{enumerate}[(1)]
  \item $F \entails H$ and $H \entails G$.
  \item $\voc{H} \subseteq \voc{F} \cap \voc{G}$.
  \item \label{item-oberschelp-fuji-fg} If $=$ occurs positively (negatively) in $H$, then $=$
    occurs positively (negatively) in $F$ ($G$).
  \end{enumerate}
\end{theorem}

\begin{proof}
  We obtain $H$ as Craig-Lyndon interpolant for $E_F \land F',\; E_G \imp
  G'$, where $F'$ and~$\lnot G'$ are flattened formulas that are equivalent to
  $F$ and $\lnot G$, respectively, and $E_F$ and $E_G$ are defined as follows.
  If $=$ occurs only negatively (positively) or not at all in $F$ ($G$), then
  $E_F$ ($E_G$) is $\true$ and $E_G$ ($E_F$) is the conjunction of the
  equality axioms \axname{Reflexivity}, \axname{Symmetry}, and
  \axname{Transitivity}. Else $=$ occurs positively in $F$ or negatively in
  $G$. In this case, place \axname{Reflexivity}, \axname{Symmetry}, and
  \axname{Transitivity} arbitrarily in $E_F$ or $E_G$.
\end{proof}

\section{Contributions of Automated Reasoning to Craig Interpolation}
\label{sec-contrib}

Methods from automated reasoning led to new strengthened variations of Craig
interpolation, a new technique for Craig interpolation in a non-classical
logic, techniques for putting specific strengthened variations of Craig
interpolation into practice, and some general observations on Craig
interpolation. We outline some of these results and provide references.

\subparagraph{Craig-Lyndon Interpolation via Consequence Finding.}

First-order resolution is \name{complete for consequence finding}
\cite{lee:thesis:1967}, i.e., whenever a clause $C$ is entailed by a set of
clauses $F$, then there is a deduction from $F$ of a clause $D$ that
subsumes~$C$. On this basis, Slagle \cite{slagle:interpolation:1970} presented
in 1970 variations of Craig-Lyndon interpolation. For propositional logic he
strengthens Craig-Lyndon interpolation in that the interpolant $H$ for sets
$F,G$ of clauses is a set of clauses \emph{deduced by resolution} from $F$ and
that each clause of $G$ is \emph{subsumed} by some clause of~$H$.

\subparagraph{Craig Interpolation with Local Proofs.}

Aside of the two-stage approach, another approach to Craig interpolation with
automated reasoning systems is pursued: \name{Local proofs}, where the basic
idea is that all inference steps are \name{local}, that is, they do not
involve both \FONLY and \GONLY symbols, which allows a particularly easy
interpolant calculation. Jhala and McMillan \cite{jhala:mcmillan:2006}
introduced this approach in 2006 for propositional logic. Subsequently it was
generalized to first-order logic \cite{mcmillan:quantified:2008}.
Kov{\'{a}}cs, Voronkov and their collaborators investigated it for first-order
logic and implemented it in the \Vampire prover
\cite{kovacs:voronkov:2009,hoder:kovacs:voronkov:grey:2012,kovacs:voronkov:2017}.
Further aspects were shown by Bonacina and Johansson \cite{bonacina:2015:on}.
Interpolation with local proofs is incomplete for first-order logic
\cite{kovacs:voronkov:2017}: Let $F = \forall x\, \fp(\fa,x)$ and let $G =
\forall y\, \lnot \fp(y,\fb)$. Then $F \land G$ is unsatisfiable and $\exists
y \forall x\, \fp(y,x)$ is a separator for $F,G$. But, since $\fa$ is \FONLY
and $\fb$ is \GONLY, any inference with~$F$ and~$G$ as premises is non-local.
Results obtained with this approach include proof transformations to minimize
calculated interpolants \cite{hoder:kovacs:voronkov:grey:2012},
transformations from non-local to local proofs, e.g.,
\cite{hoder:kovacs:voronkov:grey:2012}, and a proof that there is no lower
bound on the number of quantifier alternations in separators for two universal
first-order sentences \cite{kovacs:voronkov:2017}.

\vspace{-4pt}
\subparagraph*{Range-Restricted and Horn Interpolation.}

\name{Range-restriction} \cite{vgt} is a syntactic property of first-order
formulas that ensures domain independence \cite{foundations} and constrains
both the CNF and the DNF of the formula. A set of clauses is \name{Horn} if each
clause has at most one positive literal. Range-restriction as well as the Horn
property transfer from interpolated formulas to interpolants, if interpolation
is performed with clausal tableaux that are \emph{hyper} \cite{cw:range:2023}.

\vspace{-4pt}
\subparagraph*{Practical Access Interpolation.}
Alternatively to range-restriction, constraining quantification to respect
\name{binding patterns} \cite{benedikt:book} (see also
\refchapter{chapter:firstorder} and \refchapter{chapter:queryrewriting}) ensures
domain independence. Interpolation with clausal tableaux that are hyper also
transfers this constrained quantification from interpolated formulas to
interpolants, in a workflow with structure-preserving normal forms and a
variation of interpolant lifting that is interleaved with ground
interpolation, because the constrained quantification has no prenex form
\cite{cw:craig:report:2018}.

\enlargethispage{6pt}
\vspace{-4pt}
\subparagraph*{Craig-Lyndon Interpolation for the Intermediate Logic of Here and There with
  Application to Synthesis of Answer Set Programs.} \name{Strong equivalence}
\cite{strongly:equivalent:2001} is a useful notion of equivalence for answer
set programs under stable model semantics \cite{gelfond:lifschitz:1988}. It
can be expressed as equivalence in the three-valued logic of here and there
(HT), also known as Gödel's $G_3$. Although interpolation for this logic was
known \cite{maksimova:superintutionistic:1977}, practical construction of
Craig interpolants from proofs has been shown only recently
\cite{cw:2024:synthesis}, based on an encoding of HT in classical logic that
is conventionally used to prove strong equivalence with classical provers
\cite{lin:equivalence:2002}. To construct a Craig-Lyndon interpolant for HT
formulas $F,G$, first a classical Craig-Lyndon interpolant $H''$ for their
encodings $F',G'$ is constructed. Then, from $H''$ a stronger formula $H'$ is
constructed that can be decoded into a HT formula~$H$, the desired
Craig-Lyndon interpolant for $F,G$. The underlying argument uses that any
proof of $F' \entails H''$ can be modified to a proof of $F' \entails H'$.
A corresponding Beth theorem can be applied to synthesize answer set programs
modulo strong equivalence with respect to a background program
\cite{cw:2024:synthesis}. The interpolation technique has been adapted
\cite{cw:htinterpolation:2026} to Mints' sequent system for HT \cite{mints}.

\section{Second-Order Quantifier Elimination}
\label{sec-soqe}

\name{Second-order quantifier elimination} is an approach to \emph{uniform}
interpolation based on \emph{equivalence}, computing for a given second-order
formula an \emph{equivalent} formula of first-order logic with equality.
We assume that the given second-order formula has the form
\begin{equation}
  \label{eq-soqe-main}
  \exists p\, F,
\end{equation}
where $F$ is a first-order formula and $p$ is a predicate. This is without
loss of generality: first-order logic allows to represent $n$-ary functions by
$n+1$-ary predicates, $\forall p\, F \equiv \lnot \exists p\, \lnot F$, and
multiple occurrences of second-order quantifiers may be eliminated
innermost-first. Also for $F$ without equality, first-order formulas
equivalent to $\exists p\, F$ may be with equality.
Since not every second-order formula is equivalent to a first-order formula
with equality, second-order quantifier elimination cannot succeed for all
inputs. But there are formula classes for which it succeeds and elimination
algorithms often succeed on problems from applications, sometimes subsuming
special algorithms.

In this section we will summarize the two core approaches to second-order
quantifier elimination. As a comprehensive source we recommend the monograph
by Gabbay, Schmidt and Sza{\l}as \cite{soqe:book:2008}. As entry points for
recent developments, we refer to the SOQE workshop series
\cite{soqe:workshop:2017,soqe:workshop:2021} and to \refchapter{chapter:kr}.
The following examples illustrate various aspects of second-order quantifier
elimination and indicate some potential applications.

\begin{example}
  Each of the following examples shows an equivalence of a second-order and a
  first-order formula. In all cases, the first-order formula can be obtained
  with known algorithms for second-order quantifier elimination.

\subexample{examp-forgetting} The following equivalence illustrates
\name{forgetting}.
\[
  \exists q\; (\forall x\, (\fp(x) \imp q(x)) \land
  \forall x\, (q(x) \imp \fr(x))) \; \equiv \;
  \forall x\, (\fp(x) \imp \fr(x)).
\]
Forgetting about predicate $\fq$ in a first-order formula $F$ appears as
elimination problem $\exists q\, F$. Predicate~$q$ does not occur in the
result, while relationships between the other predicates are retained.

\subexample{examp-leibniz} By the principle of \name{Leibniz' equality} two
objects are \emph{equal} if they are not distinguishable by properties. This
can be directly phrased as follows.
\[
  \forall p\; (p(\fa) \imp p(\fb))\; \equiv\; \fa = \fb.
\]

\subexample{examp-circ} \name{Predicate circumscription}
\cite{circumjmc,dls:1997} is a technique from knowledge representation to
enforce that predicate extensions are as small as possible. The
\name{circumscription of predicate~$p$ in formula $F$} states that $F$ holds
and that there is no predicate~$p'$ such that $F\{p \mapsto p'\}$ holds and
the extension of $p'$ is strictly contained in that of $p$. This can be
expressed as a second-order formula, shown here for $F = \fp(\fa) \land
\fp(\fb)$.
  \[
    \begin{array}{ll}
      & \fp(\fa) \land \fp(\fb) \land \lnot \exists\, p'\, [(p'(\fa) \land p'(\fb)) \land
        \forall x\, (p'(x) \imp \fp(x)) \land  \lnot \forall x\, (\fp(x) \imp
        p'(x))]\\
      \equiv & \fp(\fa) \land \fp(\fb) \land \forall x\, (\fp(x) \imp (x=\fa \lor x=\fb)).
    \end{array}
  \]

  \subexample{examp-modal} Second-order quantifier elimination can be applied
  to automate correspondence theory in modal logic. Modal formulas can be
  expressed in \name{standard translation} as formulas of classical
  first-order logic. For a modal \emph{axiom}, the unary predicates that
  correspond to parameters are then universally quantified. For example, axiom
  \textbf{T}, that is, $\Box p \imp p$, is represented by the left side of the
  equivalence below. Its right side expresses the corresponding reflexivity of
  the accessibility relation as a first-order formula.
  \[
    \forall p \forall w\, [\forall v\, (\fr(w,v) \imp p(v)) \imp p(w)]\;
    \equiv\; \forall w\, \fr(w, w)
  \]

\subexample{examp-abduction} Let $F, G$ be first-order formulas and let $P =
\{p_1\ldots p_n\}$ be a set of predicates. Then the \name{weakest sufficient
  condition of $G$ with respect to $F$ in terms of $P$}
\cite{lin-snc,dls:snc,soqe:book:2008} is the weakest
(w.r.t.\ entailment) formula~$H$ involving only predicates from $P$ such that
$F \land H \entails G$. Formula $H$ can be characterized semantically as
$\forall q_1 \ldots \forall q_m\, (F \imp G)$, where $\{q_1, \ldots, q_m\} =
\predplain{F \imp G} \setminus P$. Second-order quantifier elimination then
``computes'' the weakest sufficient condition as a first-order formula. It can
be viewed as an \name{abductive explanation}, the weakest explanation of
observation $G$ with background knowledge base $F$, where $P$ is the set of
\name{abducible predicates}. The following equivalence provides an example. We
use as shorthands $\fs,\fr$ for the 0-ary predicates $\f{sprinklerWasOn},
\f{rainedLastNight}$; $\fg,\fb$ for the constants $\f{grass},\f{boots}$; and
$w$ for the predicate $\mathit{wet}$.
  \[
  \forall w\, ([(\fs \imp w(\fg)) \land (\fr \imp w(\fg)) \land (w(\fg) \imp w(\fb))] \imp
  w(\fb))\; \equiv\; \fs \lor \fr.
  \]
  The background knowledge base~$F$ is a conjunction of three implications.
  The observation~$G$ is $w(\fb)$. The abducibles~$P$ are the 0-ary predicates
  $\fs$ and $\fr$. We obtain $H = \fs \lor \fr$ as abductive explanation, or
  weakest sufficient condition.

  \subexample{examp-decision} Second-order quantifier elimination succeeds for
  the class of relational monadic formulas, i.e., formulas with no functions
  and only unary predicates. An elimination algorithm, such as Behmann's
  method \cite{beh:22,cw-relmon}, provides a \emph{decision
  procedure}: eliminating all predicates yields a first-order formula with no
  predicates, but possibly with equality, whose status is easy to check:
  unsatisfiable, valid, or expressing a constraint on the cardinality of the
  domain. In the following example the domain must have at least two elements.
  \[\exists p\, (\exists x\, p(x) \land \exists x\, \lnot p(x))\; \equiv\; \exists
  x \exists y\, x \neq y.\lipicsEnd\]
\end{example}

To see why second-order quantifier elimination, based on equivalence, produces
uniform interpolants, based on notions of consequence and on vocabulary
restrictions, we make the following observations.
\begin{proposition}
  \label{prop-uip}
  Let $F, G$ be first order formulas and let $p_1, \ldots, p_n$ be predicates
  such that $G \equiv \exists p_1 \ldots \exists p_n\, G$. Then
  \[F \entails G\; \text{ iff }\; \exists p_1 \ldots \exists p_n\, F \entails G.\]
\end{proposition}
For uniform interpolation \emph{syntactic} properties concerning the
vocabulary of the involved formulas are considered. They relate to the
\emph{semantic} characterization of Prop.~\ref{prop-uip} as follows. The
precondition $G \equiv \exists p_1 \ldots \exists p_n\, G$, which can be
equivalently expressed as $\exists p_1 \ldots \exists p_n\, G \entails G$, is
implied by the syntactic property $\predplain{G} \cap \{p_1, \ldots, p_n\} =
\emptyset$. The second-order formula $\exists p_1 \ldots \exists p_n\, F$
satisfies the syntactic properties $\predplain{\exists p_1 \ldots \exists
  p_n\, F} \cap \{p_1,\ldots,p_n\} = \emptyset$ and $\predplain{\exists p_1
  \ldots \exists p_n\, F} \subseteq \predplain F$. If $H$ is a first-order
formula equivalent to $\exists p_1 \ldots \exists p_n\, F$ obtained with some
second-order quantifier elimination method that does not introduce additional
predicate symbols and eliminates all occurrences of $p_1, \ldots, p_n$, then
the syntactic properties $\predplain{H} \cap \{p_1,\ldots,p_n\} = \emptyset$
and $\predplain{H} \subseteq \predplain F$ hold. If $H$ is, more generally,
characterized just as a first-order formula equivalent to $\exists p_1 \ldots
\exists p_n\, F$, then any Craig interpolant $H'$ for $F', H$, where $F'$ is
$F$ with $p_1,\ldots,p_n$ replaced by dedicated fresh predicate symbols,
provides a first-order equivalent to $\exists p_1 \ldots \exists p_n\, F$ with
the syntactic properties $\predplain{H'} \cap \{p_1,\ldots,p_n\} = \emptyset$
and $\predplain{H'} \subseteq \predplain F$.

\subsection{Direct Methods -- Ackermann's Lemma and the DLS Algorithm}

The family of \name{direct} methods for second-order quantifier elimination is
characterized by equivalence-preserving rewriting to a form where all
occurrences of second-order quantifiers are applied to formulas of certain
shapes for which known equivalences permit schematic elimination.
\eref{p:ackpos} and~\eref{p:ackneg} are such schemas, known as the two
versions of \name{Ackermann's lemma}, since they go back to Ackermann's 1935
paper \cite{ackermann:35} on elimination.
Ackermann's lemma is the basis of the \name{DLS} algorithm
\cite{dls:early,dls:1997}, named after its creators Doherty, {\L}ukaszewicz,
and Sza{\l}as. It initiated in the mid-1990s the \name{direct methods}, or
\name{Ackermann approach} \cite{schmidt:2012:ackermann} to elimination. Before
presenting DLS, we show some applications of Ackermann's lemma.

\begin{example}
  \

  \subexample{examp-acklemma-forgetting} Example~\ref{examp-forgetting}
  directly matches Ackermann's lemma in form~\eref{p:ackneg}. If we
  flip the conjuncts on the left side of the example, it matches Ackermann's
  lemma in form~\eref{p:ackpos}.

  \subexample{examp-acklemma-leibniz} The left side of
  Example~\ref{examp-leibniz} is equivalent to $\lnot \exists p\, [\forall x\,
    (x=\fa \imp p(x)) \land \lnot p(\fb)],$ where the second-order subformula
  matches the left side of \eref{p:ackneg}.

  \subexample{examp-acklemma-modal} The left side of Example~\ref{examp-modal}
  is equivalent to $\lnot \exists w \exists p\, [\forall v\, (\fr(w,v) \imp
    p(v)) \land \lnot p(w)],$ where the second-order subformula matches the
  left side of \eref{p:ackneg}. \lipicsEnd

\end{example}

\begin{algo}[DLS]
  \

\algoinput A second-order formula $\exists p\, F$, where $p$ is a predicate
and $F$ is a formula of first-order logic with equality.

\smallskip

\algooutput A formula of first-order logic with equality that is equivalent to
the input formula or $\FAIL$, indicating failure of the algorithm.

\algomethod The algorithm proceeds in four phases.

\begin{enumerate}[I.]
\item
  \label{item-dls-preprocessing}
  \textbf{Preprocessing.} Convert the input to the form
  \begin{equation}
    \label{eq-dls-preprocessing-1}
    \exists \xs \exists p \bigvee_{i=1}^n (A_i \land B_i),
  \end{equation}
  where $A_i$ and $B_i$ are first-order formulas such that ${-}p \notin
  \pred{A_i}$ and ${+}p \notin \pred{B_i}$, for $i \in \{1,\ldots,n\}$.
  Specifically: (1.) Eliminate $\imp$ and $\equi$ (\eref{imp},
  \eref{equi:cnf}, \eref{equi:dnf}). (2.) Remove void quantifiers
  (\eref{q:vv}). (3.) Move $\lnot$ inwards until all occurrences precede atoms
  (\eref{nnf:not}--\eref{nnf:and}, \eref{q:na}, \eref{q:ne}). (4.) Move
  $\forall$ to the right and $\exists$ to the left, as long as possible
  (\eref{q:aa}--\eref{q:qa}, renaming bound variables if necessary). (5.)
  Distribute all top-level conjunctions over the disjunctions occurring in
  conjuncts (\eref{dist:dnf}). If the result, after rearranging with respect
  to associativity and commutativity of $\land, \lor$, is not of the form
  (\ref{eq-dls-preprocessing-1}), then exit with $\FAIL$. Otherwise replace
  (\ref{eq-dls-preprocessing-1}) with the equivalent formula
  \begin{equation}
    \label{eq-dls-preprocessing-2}
    \exists \xs \bigvee_{i=1}^n \exists p\, (A_i \land B_i)
  \end{equation}
  and apply the next phases of the algorithm separately to each disjunct of
  (\ref{eq-dls-preprocessing-2}). If this succeeds with formula $R_i$ as
  first-order equivalent of the disjunct $\exists p\, (A_i \land B_i)$ for all
  $i \in \{1,\ldots,n\}$, then return $\exists \xs\, \bigvee_{i=1}^n R_i$ as
  overall output of the algorithm.
\item \textbf{Preparation for Ackermann's Lemma.} The goal of this phase is to
  transform a formula $\exists p\, (A \land B)$, where ${-}p \notin \pred{A}$
  and ${+}p \notin \pred{B}$ to a form that matches the left side of
  Ackermann's lemma in form~\eref{p:ackpos} or \eref{p:ackneg}. Both forms can
  always be obtained with distributing conjunctions over disjunctions
  (\eref{dist:dnf}), pulling out terms (\eref{t:pos}--\eref{t:dis}),
  Skolemization (\eref{sk}), and restoring implication (\eref{imp}). The
  algorithm computes both forms since the unskolemization in the next phase
  may succeed only for one, and also one form may be substantially smaller
  than the other.
\item
  \label{item-dls-ackermann-application}
  \textbf{Application of Ackermann's Lemma.} Eliminate the second-order
  quantifier by rewriting with Ackermann's lemma and then try to unskolemize
  the Skolem functions introduced in the previous step on the basis of
  \eref{sk}, applied from right to left, with an unskolemization procedure.
  Unskolemization either succeeds or terminates with failure. If it fails for
  both forms of Ackermann's lemma, exit with $\FAIL$. It if fails for one
  form, proceed with the other one. If it succeeds on both, pick one to
  proceed.
\item \textbf{Simplification.} Simplify the result of the previous phase by
  equivalence-preserving transformations. Since \eref{t:pos} and \eref{t:neg}
  have been applied during the preparation for Ackermann's lemma for
  pulling out terms, their converse application to push in terms now often
  shortens the formula substantially. \lipicsEnd
\end{enumerate}
\end{algo}

If DLS exits with $\FAIL$, this is either due to failure in
phase~\ref{item-dls-preprocessing} (\name{Preprocessing}) or due to failure of
unskolemization in phase~\ref{item-dls-ackermann-application}
(\name{Application of Ackermann's Lemma}) for both forms of Ackermann's lemma.
Skolemization \cite{soqe:book:2008,nonnengart:weidenbach:small:nf:2001} and
unskolemization \cite{mccune:unskolemizing:1988,scan:engel,dls:conradie:2006}
are advanced topics on their own. Failure of DLS due to failure of
unskolemization can be avoided by introducing \name{branching quantifiers},
also known as \name{Henkin quantifiers}. The result of elimination is then,
however, not necessarily a formula of classical first-order logic.

The output formula of DLS may be with equality, also in cases where the input
formula is without equality. The following example illustrates introduction of
equality and Skolemization in the preparation phase of DLS and unskolemization
in the lemma application phase.
\begin{example}
We apply DLS to the second-order formula $\exists p\, (A \land B)$, where $A =
\forall x\, (\fq(x) \imp (p(x,\fa) \lor p(x,\fb)))$ and $B = \forall x\, \lnot
p(x,\fc)$. This second-order formula already has the
shape~(\ref{eq-dls-preprocessing-2}) such that the preprocessing phase has no
effect. Preparation for Ackermann's lemma has then to be considered for the
two forms of Ackermann's lemma, \eref{p:ackpos} and \eref{p:ackneg}. Form
\eref{p:ackpos} requires to bring $B$ into the shape $\forall x y\,
(p(x,y) \imp G)$. With \eref{t:neg} we can rewrite $B$ accordingly to $\forall
x y\, (p(x,y) \imp y\neq c)$. We can then apply \eref{p:ackpos} with $A$ in
the role of $F$ and obtain $A\{p \mapsto \lambda x y. y\neq c\}$, that is,
$\forall x\, (\fq(x) \imp (\fa\neq\fc \lor \fb\neq\fc))$ as result.

Preparation for Ackermann's lemma in the form \eref{p:ackneg} is more
intricate, involving Skolemization. We convert $A$ in the following steps to
an equivalent formula, a conjunction where one conjunct has the shape $\forall
x y\, (G \imp p(x,y))$ as required by \eref{p:ackneg} and the other conjunct
has no occurrence of $p$.

\smallskip
\setlength{\tabcolsep}{2.8pt}
\begin{tabular}{rcll}
  1. & & $\forall x\, [\fq(x) \imp (p(x,\fa) \lor p(x,\fb))]$ & $A$\\
  2. & $\equiv$ & $\forall x\, [\fq(x) \imp \exists u\, ((u=\fa \lor u=\fb)
    \land p(x,u))]$
  & by \eref{t:dis}\\
  3. & $\equiv$ & $\forall x\, [\fq(x) \imp \exists u\, ((u=\fa \lor u=\fb)
    \land \forall y\, (y=u \imp p(x,y)))]$ & by \eref{t:pos}\\
  4. & $\equiv$ & $\forall x \exists u \forall y\, [\fq(x) \imp ((u=\fa \lor
    u=\fb) \land (y=u \imp p(x,y)))]$ & by prenexing\\
  5. & $\equiv$ & $\forall x \forall y\, [\fq(x) \imp ((s(x)=\fa \lor
    s(x)=\fb) \land (y=s(x) \imp p(x,y)))]$ & by \eref{sk} (Skolemization)\\
  6. & $\equiv$ & $\exists s \forall x \forall y\, ([\fq(x) \imp ((s(x)=\fa \lor
    s(x)=\fb))]\; \land$\\
  && \hspace{3.95em} $[(\fq(x) \land y=s(x)) \imp p(x,y)])$ & by \eref{dist:cnf}\\
  7. & $\equiv$ & $\exists s\, (\forall x\, [\fq(x) \imp ((s(x)=\fa \lor
    s(x)=\fb))]\;\land$\\
  && \hspace{1.4em} $\forall x y\, [(\fq(x) \land y=s(x)) \imp p(x,y)])$.
  & by \eref{q:aa}, \eref{q:vv}\\
\end{tabular}
\smallskip

Let $A'$ and $B'$ be the two top conjuncts of the last of these formulas,
i.e., $A' = \forall x y\, [(\fq(x) \land y=s(x)) \imp p(x,y)])$ and $B' =
\forall x\, [\fq(x) \imp ((s(x)=\fa \lor s(x)=\fb))]$. Then $A \equiv \exists
s\, (B' \land A')$ and $\exists p\, (A \land B) \equiv \exists s \exists p\, (A'
\land B' \land B)$. Applying \eref{p:ackneg} with $A'$ in the role of $\forall
x y\, (G \imp p(x,y))$ and $B' \land B$ in the role of $F$ yields
$\exists s\, (B' \land B)\{p \mapsto \lambda xy.(\fq(x) \land y=s(x))\}$,
that is,
\[\exists s\, (\forall x\, [\fq(x) \imp ((s(x)=\fa \lor s(x)=\fb))] \land
\forall x\, \lnot (\fq(x) \land \fc=s(x))).\]
With \eref{q:aa} and \eref{sk}, now applied right-to-left as unskolemization,
we obtain
\[\forall x \exists y\, ([\fq(x) \imp ((y=\fa \lor y=\fb))] \land
\lnot (\fq(x) \land \fc=y)),\] which is equivalent to $\forall x\, (\fq(x)
\imp (\fa\neq\fc \lor \fb\neq\fc))$, the result obtained before with
Ackermann's lemma in the form \eref{p:ackpos}. DLS can then pick and return
this shorter representation.
\lipicsEnd
\end{example}

Comprehensive examples of DLS are given in \cite{dls:early,dls:1997}. Details
of DLS can be varied and refined, as discussed by Gustafsson
\cite{dls:gustafsson} and Conradie \cite{dls:conradie:2006}, who also
identifies a class of formulas on which DLS succeeds. To let DLS succeed on
all inputs $\exists p\, F$ for propositional $F$, distribution of disjunction
over conjunctions (\eref{dist:cnf}) has to be incorporated \cite{cw-relmon}.

Behmann's method \cite{beh:22,cw-relmon} is similar to DLS but restricted to
relational monadic formulas, for which it guarantees success. It is based on a
simple special case of Ackermann's lemma: $\exists p\, [\forall x\, (p(x) \lor
  F) \land \forall x\, (G \lor \lnot p(x))] \equiv \forall x\, (F \lor G)$.
Ackermann's works on elimination include a variation of unskolemization for
predicates that is shown here as \eref{aqs} \cite{ackermann:35:arity}. Applied
from right to left it reduces the arity of predicates, ideally leading to
relational monadic form, which then makes Behmann's method applicable. Some
modern applications of this approach are discussed in \cite{cw-relmon}.

If the elimination result is generalized to a formula of classical fixpoint
logic, then Ackermann's lemma can be generalized by replacing condition $p
\notin \predplain{G}$ with ${-}p \notin \pred{G}$, i.e., allowing in $G$
occurrences of $p$ provided these are positive. This generalization of
Ackermann's lemma is due to Nonnengart and Sza{\l}as
\cite{nonnengart:szalas:fixpoint:1998}. In classical fixpoint logic atomic
formulas can be fixpoint formulas $\gfp{p}{\xs}{G}(\tts)$ and
$\lfp{p}{\xs}{G}(\tts)$, for \name{greatest} and \name{least fixpoint}. Our
notation $F\{p \mapsto \lambda \xs . G\}$ for substitution of a predicate $p$
by a formula $\lambda \xs . G$ is extended to substitution of a predicate $p$
by a fixpoint formula $\gfp{p}{\xs}{G}$ or $\lfp{p}{\xs}{G}$, where operators
$\GFP, \LFP$ bind the variables $\xs$ in the same way as~$\lambda$. The
expressions $\gfp{p}{\xs}{G}(\tts)$ and $\lfp{p}{\xs}{G}(\tts)$ denote
$\gfp{p}{\xs}{G}$ and $\lfp{p}{\xs}{G}$, respectively, applied to the
tuple~$\tts$ of terms. For example, $\fp(\fa)\{\fp \mapsto \gfp{\fp}{x}{G}\} =
\gfp{\fp}{x}{G}(\fa)$. The semantics of fixpoint formulas can be specified as
follows: $\gfp{p}{\xs}{G}(\ts)$ is true iff $\gfp{p}{\xs}{G}\{p \mapsto
r\}(\ts)$ is true, where $r$ is the \emph{greatest} (w.r.t. $\subseteq$)
relation satisfying $\forall \xs\, (r(\xs) \equiv G\{p \mapsto r\})$.
Analogously, $\lfp{p}{\xs}{G}(\ts)$ is true iff $\lfp{p}{\xs}{G}\{p \mapsto
r\}(\ts)$ is true, where $r$ is the \emph{smallest} (w.r.t. $\subseteq$)
relation satisfying $\forall \xs\, (r(\xs) \equiv G\{p \mapsto r\})$.

\begin{theorem}[Fixpoint Generalization of Ackermann's Lemma \cite{nonnengart:szalas:fixpoint:1998}]
  \label{thm-ack-fix-lemma}

  Let $F$ be a formula of first-order logic with equality, let $p$ be a
  predicate and let $G$ be a formula such that ${-}p \notin \pred{G}$ and such
  that no free variables of $G$ are bound by a quantifier in $F$. Then

\subtheorem{thm-ack-fix-pos}
If ${-}p \notin \pred{F}$, then
$\exists p\,
[\forall \xs\, (p(\xs) \imp G) \land F]\; \equiv\; F\{p \mapsto \gfp{p}{\xs}{G}\}$.

\subtheorem{thm-ack-fix-neg}
If ${+}p \notin \pred{F}$, then
$\exists p\,
[\forall \xs\, (G \imp p(\xs)) \land F]\; \equiv\; F\{p \mapsto \lfp{p}{\xs}{G}\}$.
\end{theorem}

DLS* \cite{dls:star} is an extension of DLS that takes this fixpoint
generalization of Ackermann's lemma into account.
Ackermann's lemma (\eref{p:ackpos}, \eref{p:ackneg}) as well as its fixpoint
generalization can be generalized by considering instead of the polarity of
$p$ in $F$ that $p$ is \emph{monotone} (\emph{down-monotone}) in $F$ and, for
the fixpoint version, that $p$ is monotone in $G$
\cite{gabbay:szalas:monotone,soqe:book:2008}.

\subsection{Predicate Elimination with Resolution -- The SCAN Algorithm}

The \name{SCAN} algorithm for second-order quantifier elimination was
introduced in 1992 by Gabbay and Ohlbach \cite{scan}. Its name is an acronym
of \name{``Synthesizing Correspondence Axioms for Normal logics''}. Actually,
the method was a re-discovery of a result by Ackermann from 1935
\cite{ackermann:35}. Our exposition is oriented at the monograph by Gabbay,
Schmidt and Sza{\l}as \cite{soqe:book:2008}.
The general idea of SCAN is to generate sufficiently many logical consequences
of the given second-order formula such that all further consequences that can
be generated from consequences with predicates to be eliminated are
\emph{redundant}. The set of consequences without predicates to be eliminated
is then equivalent to the given second-order formula.

\begin{table}
  \caption{The constraint resolution calculus \crc.}
  \label{tab-crc}

  \vspace{-4pt}
  \hspace*{1.5em}
  \begin{tabular}[t]{l@{\hspace{2em}}l}
    \textbf{Deduction} &
    $\begin{array}{c}
      N\\\midrule
      N \cup \{C\}
    \end{array}$\\[2.5ex]
    \multicolumn{2}{p{4.5cm}}{where $C$ is a \cresolvent or
        a \cfactor of premises in $N$}
  \end{tabular}
  \hspace{3em}
  \begin{tabular}[t]{l@{\hspace{2em}}l}
    \textbf{Purification} &
    $\begin{array}{c}
      N \cup \{C \lor (\lnot)p(\sss)\}\\\midrule
      N
    \end{array}$\\[2.5ex]
    \multicolumn{2}{p{6.7cm}}{if $p$ is a non-base predicate and no
        non-redundant inferences with respect to the particular literal
        $(\lnot)p(\sss)$ in the premise $\{C \lor (\lnot)p(\sss)\}$ and the
        rest of the clauses in $N$ can be performed}
  \end{tabular}
  \vspace{-4pt}
\end{table}

\begin{table}
  \caption{The inference rules of \crc.}
  \label{tab-inf-crc}

  \vspace{-4pt}
  \hspace*{1.5em}
    \begin{tabular}[t]{ll}
    \textbf{\crc-Resolution} &
    $\begin{array}{cc}
      C \lor p(\sss) & D \lor \lnot p(\ts)\\\midrule
      \multicolumn{2}{c}{C \lor D \lor \sss \neq \ts}
    \end{array}$\\[2.5ex]
    \multicolumn{2}{p{6cm}}{provided $p$ is a non-base predicate, the
        two premises have no variables in common and are distinct clauses}
    \end{tabular}
    \hspace{\fill}
    \begin{tabular}[t]{ll}
    \textbf{(Positive) \crc-Factoring} &
    $\begin{array}{c}
      C \lor p(\sss) \lor p(\ts)\\\midrule
      C \lor p(\sss) \lor \sss \neq \ts
    \end{array}$\\[2.5ex]
    \multicolumn{2}{p{6cm}}{provided $p$ is a non-base predicate}
    \end{tabular}
  \vspace{-8pt}    
\end{table}

\begin{algo}[SCAN]
  \

  \algoinput A second-order formula $\exists p_1 \ldots \exists p_n\, F$,
  where $p_1, \ldots, p_n$ are predicates and $F$ is a formula of first-order
  logic with equality.

  \algooutput The algorithm does not terminate for all inputs. If it
  terminates, the output is a formula of first-order logic with equality that
  is equivalent to the input formula or $\FAIL$, indicating failure of the
  algorithm.

  \algomethod The algorithm proceeds in three stages.
  \begin{enumerate}[I.]
    \item \textbf{Clausification.} The usual CNF conversion for first-order
      formulas, including Skolemization, is applied to the first-order
      component $F$ of the input. Its result is a set $N$ of clauses such
      that $\exists \sffs \forall \xs\, N \equiv F$, where $\sffs$ are the
      Skolem functions introduced at the conversion and $\xs$ are the free
      variables of $N$.

    \item \label{item-scan-resol} \textbf{Constraint Resolution.} This stage
      operates on the clause set $N$ obtained in the previous stage.
      Predicates $p_1,\ldots,p_n$ are distinguished as \name{non-base}
      predicates. The calculus \defname{$\crc$} (Table~\ref{tab-crc}) is
      applied to $N$ to generate a set $N_{\infty}$ of clauses such that none
      of the non-base predicates occurs in $N_{\infty}$ and $N_{\infty}$ is
      equivalent to $\exists p_1 \ldots \exists p_n\, \exists \sffs \forall
      \xs\, N$. If this stage terminates, the obtained $N_{\infty}$ is finite.

    \item \label{item-scan-unskolem} \textbf{Unskolemization.} Apply
      unskolemization, for example with McCune's algorithm
      \cite{mccune:unskolemizing:1988}, to $N_{\infty}$ to eliminate the
      Skolem functions $\sffs$ that were introduced at clausification. If this
      succeeds, return the result of unskolemization, else, exit with $\FAIL$.
      \lipicsEnd
  \end{enumerate}
\end{algo}

SCAN may fail either in stage~\ref{item-scan-resol} due to non-termination of
\cresolution or in stage~\ref{item-scan-unskolem} due to failure of
unskolemization.
The \name{deduction} rule of calculus~$\crc$ computes new clauses using the
inference rules \name{\cresolution} and \name{\cfactoring}
(Table~\ref{tab-inf-crc}). As usual for resolution, premises are assumed to be
normalized by variable renaming such that they have no shared variables. The
role of unification in resolution is in constraint resolution taken by adding
negated equality literals, ''constraints'', to the conclusion.

Theorem~\ref{thm-crc-correctness} below characterizes correctness for
constraint resolution with the \crc calculus, which includes the properties
required in the \name{Constraint Resolution} stage of SCAN. The theorem
statement uses the terminology from the framework for saturation-based proving
by Bachmair and Ganzinger \cite{resolution:handbook:2001}.
A central notion is the property \name{redundant}, which can hold for
inferences and for clauses. We may assume the so-called \name{trivial
  redundancy criterion}, where an inference is \defname{redundant} in a clause
set $N$ if its conclusion is in $N$, and a clause is never considered as
\name{redundant}.
Optionally, to take account of equivalence-preserving deletion and reduction
rules, e.g., deletion of tautological or subsumed clauses, which can be freely
added to \crc without compromising correctness, other redundancy criteria can be
employed, where
also \emph{clauses} may be classified as redundant
\cite{soqe:book:2008,resolution:handbook:2001}.
A clause set $N$ is \defname{\csaturated up to redundancy} if all inferences
with non-redundant premises from $N$ are redundant in $N$. A clause set~$N$ is
\defname{\cclosed} if $N$ is \csaturated up to redundancy and the
\name{purification} rule is not applicable. A \defname{\cderivation} is a
(possibly infinite) sequence $N_0, N_1, \ldots$ of clause sets such that for
every~$i \geq 0$, $N_{i+1}$ is obtained from $N$ by the application of a rule
in \crc.
The \defname{limit} of a \cderivation is the set $N_{\infty} \eqdef \bigcup_{j
  \geq 0} \bigcap_{k \geq j} N_k$ of persisting clauses. A \cderivation $N (=
N_0), N_1, \ldots$ from $N$ is \defname{fair} iff the conclusion of every
non-redundant inference from non-redundant premises in $N_{\infty}$ is in some
$N_j$. Intuitively, fairness means that no non-redundant inferences are
delayed indefinitely. We are now ready to state the correctness theorem \crc,
which underlies SCAN.

\begin{theorem}[Correctness of Constraint Resolution with \crc
    \cite{scan,soqe:book:2008}]
  \label{thm-crc-correctness}
  Let $N$ be a set of clauses and suppose $p_1,\ldots,p_k$ are distinguished
  as non-base predicates in $N$. Let $N (= N_0), N_1, \ldots$ be a fair
  \cderivation from $N$ with limit $N_{\infty}$. Then (1)~$N_{\infty}$ is
  \cclosed; (2)~None of the non-base predicates occurs in $N_{\infty}$;
  (3)~$N_{\infty}$ is equivalent to $\exists p_1 \ldots \exists p_n\, N$.
\end{theorem}

The set $N_{\infty}$ may be infinite. If \cresolution terminates, then
$N_{\infty}$ is finite and can be passed to the unskolemization stage of SCAN.
Ackermann \cite{ackermann:35} actually considered infinite sets $N_{\infty}$
as elimination results.
Like the output of DLS, the output of SCAN may be with equality, also in cases
where the input formula is without equality.
The following simple example illustrates the three stages of SCAN.

\begin{example} Consider the second-order formula
  $\exists q\, [(\fp(\fa) \imp q(\fa)) \land (q(\fb) \imp \exists x\, \fr(x))]$.
  Clausification of its first order component yields the following clauses,
  where $s$ is a Skolem constant.

  \smallskip
  
  \begin{tabular}{rL{10em}l}
    $C_1$ & $\lnot \fp(\fa) \lor q(\fa)$ & Input clause\\
    $C_2$ & $\lnot q(\fb) \lor \fr(s)$ & Input clause
  \end{tabular}
  \smallskip
  
  \noindent
  We now perform constraint resolution with the non-base predicate $q$. A
  \crc-resolution deduction step adds the following clause.

  \smallskip
  \begin{tabular}{rL{10em}l}
    $C_3$ & $\lnot \fp(\fa) \lor \fr(s) \lor \fa \neq \fb$ & \crc-resolvent of
    $C_1$ and $C_2$
  \end{tabular}
  \smallskip
  
  \noindent
  A purification step then deletes $C_1$, and a second purification step
  deletes $C_2$. We thus leave the constraint resolution stage with the
  singleton set containing $C_3$ as result $N_{\infty}$. Unskolemization then
  gives us the first-order formula $\exists x\, (\lnot \fp(\fa) \lor \fr(x)
  \lor \fa \neq \fb)$ as final result of the second-order quantifier
  elimination by SCAN. This formula may be rearranged as $(\fp(\fa) \land
  \fa=\fb) \imp \exists x\, \fr(x)$. \lipicsEnd
\end{example}

Refinements and variations of SCAN are discussed in
\cite{scan:engel,scan-system-paper,cw-skp,scan:complete:2004,soqe:book:2008}.
If the \crc calculus is equipped with deletion of subsumed clauses, then SCAN
is complete for the case where the quantified predicates are nullary
\cite{soqe:book:2008}. An adaptation of SCAN for modal logics is complete for
Sahlqvist formulas \cite{scan:complete:2004}.

\section*{Acknowledgments}
\addcontentsline{toc}{section}{Acknowledgments}

The author thanks Wolfgang Bibel, Patrick Koopmann and Philipp Rümmer for
their valuable comments and suggestions. Funded by the Deutsche
Forschungsgemeinschaft (DFG, German Research Foundation) --
Project-ID~457292495.

\bibliography{biblioautomated,taci}

\providecommand{\noopsort}[1]{}\providecommand{\noopsort}[1]{}
\begin{thebibliography}{100}

\bibitem{foundations}
Serge Abiteboul, Richard Hull, and Victor Vianu.
\newblock {\em Foundations of Databases}.
\newblock Addison Wesley, 1995.

\bibitem{ackermann:35}
Wilhelm Ackermann.
\newblock Untersuchungen über das {E}liminationsproblem der mathemati\-schen
  {L}ogik.
\newblock {\em Math. Ann.}, 110:390--413, 1935.
\newblock \href {https://doi.org/10.1007/BF01448035}
  {\path{doi:10.1007/BF01448035}}.

\bibitem{ackermann:35:arity}
Wilhelm Ackermann.
\newblock Zum {E}liminationsproblem der mathematischen {L}ogik.
\newblock {\em Math. Ann.}, 111:61--63, 1935.
\newblock \href {https://doi.org/10.1007/BF01472201}
  {\path{doi:10.1007/BF01472201}}.

\bibitem{meteor:1994}
Owen~L. Astrachan.
\newblock {METEOR}: {E}xploring model elimination theorem proving.
\newblock {\em J. Autom. Reasoning}, 13(3):283--296, 1994.
\newblock \href {https://doi.org/10.1007/BF00881946}
  {\path{doi:10.1007/BF00881946}}.

\bibitem{handbook:baaz:egly:leitsch:nf}
Matthias Baaz, Uwe Egly, and Alexander Leitsch.
\newblock Normal form transformations.
\newblock In John~Alan Robinson and Andrei Voronkov, editors, {\em Handb. of
  Autom. Reasoning}, volume~1, chapter~5, pages 273--333. Elsevier, 2001.
\newblock \href {https://doi.org/10.1016/B978-044450813-3/50007-2}
  {\path{doi:10.1016/B978-044450813-3/50007-2}}.

\bibitem{baaz:leitsch:2011}
Matthias Baaz and Alexander Leitsch.
\newblock {\em Methods of Cut-Elimination}.
\newblock Springer, 2011.
\newblock \href {https://doi.org/10.1007/978-94-007-0320-9}
  {\path{doi:10.1007/978-94-007-0320-9}}.

\bibitem{resolution:handbook:2001}
Leo Bachmair and Harald Ganzinger.
\newblock Resolution theorem proving.
\newblock In Alan Robinson and Andrei Voronkov, editors, {\em Handb. of Autom.
  Reasoning}, volume~1, chapter~2, pages 19--99. Elsevier, 2001.
\newblock \href {https://doi.org/10.1016/B978-044450813-3/50004-7}
  {\path{doi:10.1016/B978-044450813-3/50004-7}}.

\bibitem{bachm:ganz:92}
Leo Bachmair, Harald Ganzinger, Christopher Lynch, and Wayne Snyder.
\newblock Basic paramodulation and superposition.
\newblock In Deepak Kapur, editor, {\em CADE-11}, volume 607 of {\em LNCS
  (LNAI)}, pages 462--476, 1992.
\newblock \href {https://doi.org/10.1007/3-540-55602-8_185}
  {\path{doi:10.1007/3-540-55602-8_185}}.

\bibitem{beh:22}
Heinrich Behmann.
\newblock {Beiträge zur Algebra der Logik, insbesondere zum
  Entscheidungsproblem}.
\newblock {\em Math. Ann.}, 86(3--4):163--229, 1922.
\newblock \href {https://doi.org/10.1007/BF01457985}
  {\path{doi:10.1007/BF01457985}}.

\bibitem{chapter:queryrewriting}
Michael Benedikt.
\newblock Interpolation and query rewriting.
\newblock In {\noopsort{Cate}{ten Cate}} et~al. \cite{taci}, chapter~14.
\newblock Preprint: \url{https://arxiv.org/abs/2606.15737}.
\newblock \href {https://doi.org/10.5334/bdg.n} {\path{doi:10.5334/bdg.n}}.

\bibitem{benedikt:book}
Michael Benedikt, Julien Leblay, Balder ten Cate, and Efthymia Tsamoura.
\newblock {\em Generating Plans from Proofs: The Interpolation-based Approach
  to Query Reformulation}.
\newblock Morgan \& Claypool, 2016.
\newblock \href {https://doi.org/10.1007/978-3-031-01856-5}
  {\path{doi:10.1007/978-3-031-01856-5}}.

\bibitem{bibel:atp:1987}
Wolfgang Bibel.
\newblock {\em Automated Theorem Proving}.
\newblock Vieweg, Braunschweig, 1987.
\newblock First edition 1982.
\newblock \href {https://doi.org/10.1007/978-3-322-90102-6}
  {\path{doi:10.1007/978-3-322-90102-6}}.

\bibitem{bibel:otten:2020}
Wolfgang Bibel and Jens Otten.
\newblock From {S}ch\"utte's formal systems to modern automated deduction.
\newblock In Reinhard Kahle and Michael Rathjen, editors, {\em The Legacy of
  Kurt Sch\"utte}, chapter~13, pages 215--249. Springer, 2020.
\newblock \href {https://doi.org/10.1007/978-3-030-49424-7_13}
  {\path{doi:10.1007/978-3-030-49424-7_13}}.

\bibitem{hammer}
Jasmin~Christian Blanchette, Cezary Kaliszyk, Lawrence~C. Paulson, and Josef
  Urban.
\newblock Hammering towards {QED}.
\newblock {\em J. Formaliz. Reason.}, 9(1):101--148, 2016.
\newblock \href {https://doi.org/10.6092/ISSN.1972-5787/4593}
  {\path{doi:10.6092/ISSN.1972-5787/4593}}.

\bibitem{bonacina:2015:ground}
Maria~Paola Bonacina and Moa Johansson.
\newblock Interpolation systems for ground proofs in automated deduction: a
  survey.
\newblock {\em J. Autom. Reasoning}, 54(4):353--390, 2015.
\newblock \href {https://doi.org/10.1007/s10817-015-9325-5}
  {\path{doi:10.1007/s10817-015-9325-5}}.

\bibitem{bonacina:2015:on}
Maria~Paola Bonacina and Moa Johansson.
\newblock On interpolation in automated theorem proving.
\newblock {\em J. Autom. Reasoning}, 54(1):69--97, 2015.
\newblock \href {https://doi.org/10.1007/s10817-014-9314-0}
  {\path{doi:10.1007/s10817-014-9314-0}}.

\bibitem{brand:75}
D.~Brand.
\newblock Proving theorems with the modification method.
\newblock {\em SIAM J. of Computing}, 4(4):412--430, 1975.
\newblock \href {https://doi.org/10.1137/0204036} {\path{doi:10.1137/0204036}}.

\bibitem{chapter:firstorder}
Balder {\noopsort{Cate}{ten Cate}} and Jesse Comer.
\newblock Interpolation in first-order logic.
\newblock In {\noopsort{Cate}{ten Cate}} et~al. \cite{taci}, chapter~2.
\newblock Preprint: \url{https://arxiv.org/abs/2510.03822}.
\newblock \href {https://doi.org/10.5334/bdg.b} {\path{doi:10.5334/bdg.b}}.

\bibitem{taci}
Balder {\noopsort{Cate}{ten Cate}}, Jean~Christoph Jung, Patrick Koopmann,
  Christoph Wernhard, and Frank Wolter, editors.
\newblock {\em Theory and Applications of {C}raig Interpolation}.
\newblock Ubiquity Press, 2026.
\newblock To appear; preprints accessible from
  \url{https://cibd.bitbucket.io/taci/}.
\newblock \href {https://doi.org/10.5334/bdg} {\path{doi:10.5334/bdg}}.

\bibitem{chang:lee}
Chin-Liang Chang and Richard Char-Tung Lee.
\newblock {\em Symbolic Logic and Automated Theorem Proving}.
\newblock Academic Press, 1973.

\bibitem{dls:conradie:2006}
Willem Conradie.
\newblock On the strength and scope of {DLS}.
\newblock {\em J. Applied Non-Classical Logic}, 16(3--4):279--296, 2006.
\newblock \href {https://doi.org/10.3166/jancl.16.279-296}
  {\path{doi:10.3166/jancl.16.279-296}}.

\bibitem{craig:linear}
William Craig.
\newblock Linear reasoning. {A} new form of the {H}erbrand-{G}entzen theorem.
\newblock {\em J. Symb. Log.}, 22(3):250--268, 1957.
\newblock \href {https://doi.org/10.2307/2963593} {\path{doi:10.2307/2963593}}.

\bibitem{craig:uses}
William Craig.
\newblock Three uses of the {H}erbrand-{G}entzen theorem in relating model
  theory and proof theory.
\newblock {\em J. Symb. Log.}, 22(3):269--285, 1957.
\newblock \href {https://doi.org/10.2307/2963594} {\path{doi:10.2307/2963594}}.

\bibitem{craig:2008:history}
William Craig.
\newblock Elimination problems in logic: A brief history.
\newblock {\em Synthese}, 164(3):321--332, 2008.
\newblock \href {https://doi.org/10.1007/s11229-008-9352-4}
  {\path{doi:10.1007/s11229-008-9352-4}}.

\bibitem{craig:2008:road}
William Craig.
\newblock The road to two theorems of logic.
\newblock {\em Synthese}, 164(3):333--339, 2008.
\newblock \href {https://doi.org/10.1007/s11229-008-9353-3}
  {\path{doi:10.1007/s11229-008-9353-3}}.

\bibitem{equality:sequent:handbook:ar}
Anatoli Degtyarev and Andrei Voronkov.
\newblock Equality reasoning in sequent-based calculi.
\newblock In Alan Robinson and Andrei Voronkov, editors, {\em Handb. of Autom.
  Reasoning}, volume~1, chapter~10, pages 611--706. Elsevier, 2001.
\newblock \href {https://doi.org/10.1016/B978-044450813-3/50012-6}
  {\path{doi:10.1016/B978-044450813-3/50012-6}}.

\bibitem{dls:1997}
Patrick Doherty, Witold {\L}ukaszewicz, and Andrzej Sza{\l}as.
\newblock Computing circumscription revisited: A reduction algorithm.
\newblock {\em J. Autom. Reasoning}, 18(3):297--338, 1997.
\newblock \href {https://doi.org/10.1023/A:1005722130532}
  {\path{doi:10.1023/A:1005722130532}}.

\bibitem{dls:star}
Patrick Doherty, Witold {\L}ukaszewicz, and Andrzej Sza{\l}as.
\newblock General domain circumscription and its effective reductions.
\newblock {\em Fundam. Informaticae}, 36(1):23--55, 1998.
\newblock \href {https://doi.org/10.3233/FI-1998-3612}
  {\path{doi:10.3233/FI-1998-3612}}.

\bibitem{dls:snc}
Patrick Doherty, Witold {\L}ukaszewicz, and Andrzej Sza{\l}as.
\newblock Computing strongest necessary and weakest sufficient conditions of
  first-order formulas.
\newblock In Bernhard Nebel, editor, {\em IJCAI-01}, pages 145--151. Morgan
  Kaufmann, 2001.

\bibitem{eder:subst:1985}
Elmar Eder.
\newblock Properties of substitutions and unification.
\newblock {\em J. Symb. Comput.}, 1(1):31--46, 1985.
\newblock \href {https://doi.org/10.1016/S0747-7171(85)80027-4}
  {\path{doi:10.1016/S0747-7171(85)80027-4}}.

\bibitem{een:biere:elim}
Niklas E{\'e}n and Armin Biere.
\newblock Effective preprocessing in {SAT} through variable and clause
  elimination.
\newblock In Fahiem Bacchus and Toby Walsh, editors, {\em SAT~'05}, volume 3569
  of {\em LNCS}, pages 61--75, 2005.
\newblock \href {https://doi.org/10.1007/11499107_5}
  {\path{doi:10.1007/11499107_5}}.

\bibitem{scan:engel}
Thorsten Engel.
\newblock Quantifier elimination in second-order predicate logic.
\newblock Master's thesis, Max-Planck-Institut für Informatik, Saarbrücken,
  1996.

\bibitem{fitting:book:foar:2nd}
Melvin Fitting.
\newblock {\em First-Order Logic and Automated Theorem Proving}.
\newblock Springer, 2nd edition, 1996.
\newblock \href {https://doi.org/10.1007/978-1-4612-2360-3}
  {\path{doi:10.1007/978-1-4612-2360-3}}.

\bibitem{fujiwara:1978}
Tsuyoshi Fujiwara.
\newblock A variation of {L}yndon-{K}eisler's homomorphism theorem and its
  applications to interpolation theorems.
\newblock {\em J. Math. Soc. Japan}, 30(2):287--302, 1978.
\newblock \href {https://doi.org/10.2969/jmsj/03020287}
  {\path{doi:10.2969/jmsj/03020287}}.

\bibitem{scan}
Dov Gabbay and Hans~J{\"u}rgen Ohlbach.
\newblock Quantifier elimination in second-order predicate logic.
\newblock In Bernhard Nebel, Charles Rich, and William~R. Swartout, editors,
  {\em {KR}'92}, pages 425--435. Morgan Kaufmann, 1992.

\bibitem{soqe:book:2008}
Dov~M. Gabbay, Renate~A. Schmidt, and Andrzej Sza{\l}as.
\newblock {\em Second-Order Quantifier Elimination: Foundations, Computational
  Aspects and Applications}.
\newblock College Publications, 2008.

\bibitem{gabbay:szalas:monotone}
Dov~M. Gabbay and Andrzej Sza{\l}as.
\newblock Second-order quantifier elimination in higher-order contexts with
  applications to the semantical analysis of conditionals.
\newblock {\em Studia Logica}, 87(1):37--50, 2007.
\newblock \href {https://doi.org/10.1007/S11225-007-9075-4}
  {\path{doi:10.1007/S11225-007-9075-4}}.

\bibitem{gelfond:lifschitz:1988}
Michael Gelfond and Vladimir Lifschitz.
\newblock The stable model semantics for logic programming.
\newblock In Robert~A. Kowalski and Kenneth~A. Bowen, editors, {\em ICLP/SLP
  1988}, pages 1070--1080, Cambridge, MA, 1988. MIT Press.

\bibitem{scan:complete:2004}
Valentin Goranko, Ullrich Hustadt, Renate~A. Schmidt, and Dimiter Vakarelov.
\newblock {SCAN} is complete for all {S}ahlqvist formulae.
\newblock In Rudolf Berghammer, Bernhard M{\"{o}}ller, and Georg Struth,
  editors, {\em RelMiCS 7}, volume 3051 of {\em LNCS}, pages 149--162, 2004.
\newblock \href {https://doi.org/10.1007/978-3-540-24771-5_13}
  {\path{doi:10.1007/978-3-540-24771-5_13}}.

\bibitem{dls:gustafsson}
Joakim Gustafsson.
\newblock An implementation and optimization of an algorithm for reducing
  formulae in second-order logic.
\newblock Technical Report LiTH-MAT-R-96-04, Univ. Linköping, 1996.

\bibitem{handbook:ar:haehnle:2001}
Reiner {\noopsort{Haehnle}{Hähnle}}.
\newblock Tableaux and related methods.
\newblock In Alan Robinson and Andrei Voronkov, editors, {\em Handb. of Autom.
  Reasoning}, volume~1, chapter~3, pages 101--178. Elsevier, 2001.
\newblock \href {https://doi.org/10.1016/b978-044450813-3/50005-9}
  {\path{doi:10.1016/b978-044450813-3/50005-9}}.

\bibitem{harrison:2009}
John Harrison.
\newblock {\em Handbook of Practical Logic and Automated Reasoning}.
\newblock Cambridge University Press, 2009.
\newblock \href {https://doi.org/10.1017/CBO9780511576430}
  {\path{doi:10.1017/CBO9780511576430}}.

\bibitem{herbrand}
Jacques Herbrand.
\newblock {\em Recherches sur la théorie de la démonstration}.
\newblock PhD thesis, University of Paris, 1930.

\bibitem{cw:2024:synthesis}
Jan Heuer and Christoph Wernhard.
\newblock Synthesizing strongly equivalent logic programs: {B}eth definability
  for answer set programs via {C}raig interpolation in first-order logic.
\newblock In Christoph Benzm{\"u}ller, Marjin J.~H. Heule, and Renate~A.
  Schmidt, editors, {\em IJCAR 2024}, volume 14739 of {\em LNCS (LNAI)}, pages
  172--193. Springer, 2024.
\newblock \href {https://doi.org/10.1007/978-3-031-63498-7_11}
  {\path{doi:10.1007/978-3-031-63498-7_11}}.

\bibitem{hoder:kovacs:voronkov:grey:2012}
Krystof Hoder, Laura Kov{\'{a}}cs, and Andrei Voronkov.
\newblock Playing in the grey area of proofs.
\newblock In John Field and Michael Hicks, editors, {\em POPL 2012}, pages
  259--272. ACM, 2012.
\newblock \href {https://doi.org/10.1145/2103656.2103689}
  {\path{doi:10.1145/2103656.2103689}}.

\bibitem{huang:95}
Guoxiang Huang.
\newblock Constructing {C}raig interpolation formulas.
\newblock In Ding-Zhu Du and Ming Li, editors, {\em \mbox{COCOON} '95}, volume
  959 of {\em LNCS}, pages 181--190. Springer, 1995.
\newblock \href {https://doi.org/10.1007/BFb0030832}
  {\path{doi:10.1007/BFb0030832}}.

\bibitem{jarv:blocked:2010}
Matti J{\"a}rvisalo, Armin Biere, and Marijn Heule.
\newblock Blocked clause elimination.
\newblock In Javier Esparza and Rupak Majumdar, editors, {\em TACAS 2010},
  volume 6015 of {\em LNCS}, pages 129--144, 2010.
\newblock \href {https://doi.org/10.1007/978-3-642-12002-2_10}
  {\path{doi:10.1007/978-3-642-12002-2_10}}.

\bibitem{jhala:mcmillan:2006}
Ranjit Jhala and Kenneth~L. McMillan.
\newblock A practical and complete approach to predicate refinement.
\newblock In Holger Hermanns and Jens Palsberg, editors, {\em TACAS 2006},
  volume 3920 of {\em LNCS}, pages 459--473. Springer, 2006.
\newblock \href {https://doi.org/10.1007/11691372_33}
  {\path{doi:10.1007/11691372_33}}.

\bibitem{chapter:kr}
Jean~Christoph Jung, Patrick Koopmann, and Matthias Knorr.
\newblock Interpolation in knowledge representation.
\newblock In {\noopsort{Cate}{ten Cate}} et~al. \cite{taci}, chapter~15.
\newblock Preprint: \url{https://arxiv.org/abs/2512.08833}.
\newblock \href {https://doi.org/10.5334/bdg.o} {\path{doi:10.5334/bdg.o}}.

\bibitem{kaliszyk:2015:femalecop}
Cezary Kaliszyk and Josef Urban.
\newblock {FEMaLeCoP}: Fairly efficient machine learning connection prover.
\newblock In Martin Davis, Ansgar Fehnker, Annabelle McIver, and Andrei
  Voronkov, editors, {\em LPAR-20}, volume 9450 of {\em LNCS (LNAI)}, pages
  88--96. Springer, 2015.
\newblock \href {https://doi.org/10.1007/978-3-662-48899-7_7}
  {\path{doi:10.1007/978-3-662-48899-7_7}}.

\bibitem{blocked:fol:2017}
Benjamin Kiesl, Martin Suda, Martina Seidl, Hans Tompits, and Armin Biere.
\newblock Blocked clauses in first-order logic.
\newblock In Thomas Eiter and David Sands, editors, {\em LPAR-21}, volume~46 of
  {\em EPiC}, pages 31--48, 2017.
\newblock \href {https://doi.org/10.29007/C3WQ} {\path{doi:10.29007/C3WQ}}.

\bibitem{soqe:workshop:2017}
Patrick Koopmann, Sebastian Rudolph, Renate~A. Schmidt, and Christoph Wernhard,
  editors.
\newblock {\em SOQE 2017}, volume 2013 of {\em CEUR Workshop Proc.}
  CEUR-WS.org, 2017.
\newblock URL: \url{https://ceur-ws.org/Vol-2013}.

\bibitem{kovacs:voronkov:2009}
Laura Kov{\'{a}}cs and Andrei Voronkov.
\newblock Interpolation and symbol elimination.
\newblock In Renate~A. Schmidt, editor, {\em CADE-22}, volume 5663 of {\em
  LNCS}, pages 199--213. Springer, 2009.
\newblock \href {https://doi.org/10.1007/978-3-642-02959-2_17}
  {\path{doi:10.1007/978-3-642-02959-2_17}}.

\bibitem{vampire}
Laura Kov{\'a}cs and Andrei Voronkov.
\newblock First-order theorem proving and {Vampire}.
\newblock In Natasha Sharygina and Helmut Veith, editors, {\em CAV 2013},
  volume 8044 of {\em LNCS}, pages 1--35. Springer, 2013.
\newblock \href {https://doi.org/10.1007/978-3-642-39799-8_1}
  {\path{doi:10.1007/978-3-642-39799-8_1}}.

\bibitem{kovacs:voronkov:2017}
Laura Kov{\'{a}}cs and Andrei Voronkov.
\newblock First-order interpolation and interpolating proof systems.
\newblock In Thomas Eiter and David Sands, editors, {\em LPAR-21}, volume~46 of
  {\em EPiC}, pages 49--64. EasyChair, 2017.
\newblock \href {https://doi.org/10.29007/1qb8} {\path{doi:10.29007/1qb8}}.

\bibitem{kreisel:krivine:modeltheory:1971}
Georg Kreisel and Jean-Louis Krivine.
\newblock {\em Elements of mathematical logic. (Model theory)}.
\newblock North-Holland, Amsterdam, 2. edition, 1971.
\newblock First edition 1967. Translation of the French \textit{Eléments de
  logique mathématique, théorie des modeles} published by Dunod, Paris in
  1964.

\bibitem{kullmann:1999}
Oliver Kullmann.
\newblock On a generalization of extended resolution.
\newblock {\em Discret. Appl. Math.}, 96-97:149--176, 1999.
\newblock \href {https://doi.org/10.1016/S0166-218X(99)00037-2}
  {\path{doi:10.1016/S0166-218X(99)00037-2}}.

\bibitem{lee:thesis:1967}
Richard Char-Tung Lee.
\newblock {\em A completeness theorem and computer program for finding theorems
  derivable from given axioms}.
\newblock PhD thesis, University of California, Berkeley, CA, 1967.

\bibitem{letz:diss}
Reinhold Letz.
\newblock {\em First-Order Calculi and Proof Procedures for Automated
  Deduction}.
\newblock Dissertation, TU München, 1993.
\newblock
  \url{https://web.archive.org/web/20230604101128/https://www2.tcs.ifi.lmu.de/~letz/diss.ps},
  accessed Jun 12, 2026.

\bibitem{handbook:tableaux:letz}
Reinhold Letz.
\newblock First-order tableau methods.
\newblock In Marcello D'Agostino, Dov~M. Gabbay, Reiner Hähnle, and Joachim
  Posegga, editors, {\em Handb. of Tableau Methods}, pages 125--196. Kluwer
  Academic Publishers, 1999.
\newblock \href {https://doi.org/10.1007/978-94-017-1754-0_3}
  {\path{doi:10.1007/978-94-017-1754-0_3}}.

\bibitem{letz:habil}
Reinhold Letz.
\newblock {\em Tableau and Connection Calculi. Structure, Complexity,
  Implementation}.
\newblock Habilitationsschrift, TU München, 1999.
\newblock
  \url{https://web.archive.org/web/20230604101128/https://www2.tcs.ifi.lmu.de/~letz/habil.ps},
  accessed Jun 12, 2026.

\bibitem{setheo:1992}
Reinhold Letz, Johann Schumann, Stefan Bayerl, and Wolfgang Bibel.
\newblock {SETHEO:} {A} high-performance theorem prover.
\newblock {\em J. Autom. Reasoning}, 8(2):183--212, 1992.
\newblock \href {https://doi.org/10.1007/BF00244282}
  {\path{doi:10.1007/BF00244282}}.

\bibitem{letz:stenz:handbook}
Reinhold Letz and Gernot Stenz.
\newblock Model elimination and connection tableau procedures.
\newblock In Alan Robinson and Andrei Voronkov, editors, {\em Handb. of Autom.
  Reasoning}, volume~2, chapter~28, pages 2015--2114. Elsevier, 2001.
\newblock \href {https://doi.org/10.1016/B978-044450813-3/50030-8}
  {\path{doi:10.1016/B978-044450813-3/50030-8}}.

\bibitem{strongly:equivalent:2001}
Vladimir Lifschitz, David Pearce, and Agust{\'{\i}}n Valverde.
\newblock Strongly equivalent logic programs.
\newblock {\em ACM Trans. Comp. Log.}, 2(4):526--541, 2001.
\newblock \href {https://doi.org/10.1145/383779.383783}
  {\path{doi:10.1145/383779.383783}}.

\bibitem{lin-snc}
Fangzhen Lin.
\newblock On strongest necessary and weakest sufficient conditions.
\newblock {\em Artificial Intelligence}, 128:143--159, 2001.

\bibitem{lin:equivalence:2002}
Fangzhen Lin.
\newblock Reducing strong equivalence of logic programs to entailment in
  classical propositional logic.
\newblock In {\em KR-02}, pages 170--176. Morgan Kaufmann, 2002.

\bibitem{loveland:1968}
Donald~W. Loveland.
\newblock Mechanical theorem proving by model elimination.
\newblock {\em JACM}, 15(2):236--251, 1968.
\newblock \href {https://doi.org/10.1145/321450.321456}
  {\path{doi:10.1145/321450.321456}}.

\bibitem{maksimova:superintutionistic:1977}
Larisa~L. Maksimova.
\newblock {C}raig's theorem in superintuitionistic logics and amalgamable
  varieties of pseudo-{B}oolean algebras.
\newblock {\em Algebra and Logic}, 16(6):427--455, 1977.
\newblock \href {https://doi.org/10.1007/BF01670006}
  {\path{doi:10.1007/BF01670006}}.

\bibitem{circumjmc}
John McCarthy.
\newblock Circumscription -- a form of non-monotonic reasoning.
\newblock {\em AI}, 13:27--39, 1980.
\newblock \href {https://doi.org/10.1016/0004-3702(80)90011-9}
  {\path{doi:10.1016/0004-3702(80)90011-9}}.

\bibitem{mccune:unskolemizing:1988}
William McCune.
\newblock Un-{S}kolemizing clause sets.
\newblock {\em Information Processing Letters}, 29(5):257--263, 1988.

\bibitem{otter}
William McCune.
\newblock {OTTER} 3.3 {Reference} {Manual}.
\newblock Technical Report ANL/MCS-TM-263, Argonne National Laboratory, 2003.
\newblock \url{https://www.cs.unm.edu/~mccune/otter/Otter33.pdf}, accessed Jun
  12, 2026.

\bibitem{prover9}
William McCune.
\newblock Prover9 and {Mace4}.
\newblock \url{http://www.cs.unm.edu/~mccune/prover9}, accessed Jun 12, 2026,
  2005--2010.

\bibitem{mcmillan:2003}
Kenneth~L. McMillan.
\newblock Interpolation and {SAT}-based model checking.
\newblock In Warren A.~Hunt Jr. and Fabio Somenzi, editors, {\em CAV 2023},
  volume 2725 of {\em LNCS}, pages 1--13. Springer, 2003.
\newblock \href {https://doi.org/10.1007/978-3-540-45069-6_1}
  {\path{doi:10.1007/978-3-540-45069-6_1}}.

\bibitem{mcmillan:2005}
Kenneth~L. McMillan.
\newblock An interpolating theorem prover.
\newblock {\em Theor. Comput. Sci.}, 345(1):101--121, 2005.

\bibitem{mcmillan:quantified:2008}
Kenneth~L. McMillan.
\newblock Quantified invariant generation using an interpolating saturation
  prover.
\newblock In C.~R. Ramakrishnan and Jakob Rehof, editors, {\em TACAS 2008},
  volume 4963 of {\em LNCS}, pages 413--427. Springer, 2008.
\newblock \href {https://doi.org/10.1007/978-3-540-78800-3_31}
  {\path{doi:10.1007/978-3-540-78800-3_31}}.

\bibitem{mints}
Grigori Mints.
\newblock Cut-free formulations for a quantified logic of here and there.
\newblock {\em Ann. Pure Appl. Log.}, 162(3):237--242, 2010.
\newblock \href {https://doi.org/10.1016/J.APAL.2010.09.009}
  {\path{doi:10.1016/J.APAL.2010.09.009}}.

\bibitem{motohashi:1984}
Nobuyoshi Motohashi.
\newblock Equality and {L}yndon's interpolation theorem.
\newblock {\em J. Symb. Log.}, 49(1):123--128, 1984.
\newblock \href {https://doi.org/10.2307/2274095} {\path{doi:10.2307/2274095}}.

\bibitem{nonnengart:szalas:fixpoint:1998}
Andreas Nonnengart and Andrzej Sza{\l}as.
\newblock A fixpoint approach to second-order quantifier elimination with
  applications to correspondence theory.
\newblock In E.~Orlowska, editor, {\em Logic at Work. Essays Ded. to the Mem.
  of Helena Rasiowa}, pages 89--108. Springer, 1998.

\bibitem{nonnengart:weidenbach:small:nf:2001}
Andreas Nonnengart and Christoph Weidenbach.
\newblock Computing small clause normal forms.
\newblock In Alan Robinson and Andrei Voronkov, editors, {\em Handb. of Autom.
  Reasoning}, volume~1, chapter~6, pages 335--367. Elsevier, 2001.
\newblock \href {https://doi.org/10.1016/B978-044450813-3/50008-4}
  {\path{doi:10.1016/B978-044450813-3/50008-4}}.

\bibitem{oberschelp:1968}
Arnold Oberschelp.
\newblock On the {C}raig-{L}yndon interpolation theorem.
\newblock {\em J. Symb. Log.}, 33(2):271--274, 1968.
\newblock \href {https://doi.org/10.2307/2269873} {\path{doi:10.2307/2269873}}.

\bibitem{scan-system-paper}
Hans~Jürgen Ohlbach.
\newblock {SCAN} -- {E}limination of predicate quantifiers: System description.
\newblock In Michael~A. McRobbie and John~K. Slaney, editors, {\em CADE-13},
  volume 1104 of {\em LNCS (LNAI)}, pages 161--165. Springer, 1996.
\newblock \href {https://doi.org/10.1007/3-540-61511-3_77}
  {\path{doi:10.1007/3-540-61511-3_77}}.

\bibitem{otten:2021:nanocop}
Jens Otten.
\newblock The {nanoCoP} 2.0 connection provers for classical, intuitionistic
  and modal logics.
\newblock In Anupam Das and Sara Negri, editors, {\em TABLEAUX 2021}, volume
  12842 of {\em LNCS (LNAI)}, pages 236--249. Springer, 2021.
\newblock \href {https://doi.org/10.1007/978-3-030-86059-2_14}
  {\path{doi:10.1007/978-3-030-86059-2_14}}.

\bibitem{otten:2003:leancopinabstract}
Jens Otten and Wolfgang Bibel.
\newblock {leanCoP}: lean connection-based theorem proving.
\newblock {\em J. Symb. Comput.}, 36(1-2):139--161, 2003.
\newblock \href {https://doi.org/10.1016/S0747-7171(03)00037-3}
  {\path{doi:10.1016/S0747-7171(03)00037-3}}.

\bibitem{prawitz:1960:improved}
Dag Prawitz.
\newblock An improved proof procedure.
\newblock {\em Theoria}, 26:102--139, 1960.

\bibitem{prawitz:1969:advances}
Dag Prawitz.
\newblock Advances and problems in mechanical proof procedures.
\newblock {\em Machine Intelligence}, 4:59--71, 1969.
\newblock Reprinted with author preface in J. Siekmann, G. Wright (eds.):
  \textit{Automation of Reasoning, vol 2: Classical Papers on Computational
  Logic 1967--1970, Springer, 1983, pp.~283--297}.

\bibitem{satcop:2021}
Michael Rawson and Giles Reger.
\newblock Eliminating models during model elimination.
\newblock In Anupam Das and Sara Negri, editors, {\em TABLEAUX 2021}, volume
  12842 of {\em LNCS (LNAI)}, pages 250--265. Springer, 2021.
\newblock \href {https://doi.org/10.1007/978-3-030-86059-2_15}
  {\path{doi:10.1007/978-3-030-86059-2_15}}.

\bibitem{rwzb:lemmas:2023}
Michael Rawson, Christoph Wernhard, Zsolt Zombori, and Wolfgang Bibel.
\newblock Lemmas: Generation, selection, application.
\newblock In Revantha Ramanayake and Josef Urban, editors, {\em TABLEAUX 2023},
  LNAI, pages 153--174, 2023.
\newblock \href {https://doi.org/10.1007/978-3-031-43513-3_9}
  {\path{doi:10.1007/978-3-031-43513-3_9}}.

\bibitem{paramodulation:1969}
George~A. Robinson and Larry Wos.
\newblock Paramodulation and theorem-proving in first-order theories with
  equality.
\newblock In Bernard Meltzer and Donald Michie, editors, {\em Machine
  Intelligence IV}, pages 135--150. Edinburgh University Press, 1969.

\bibitem{robinson:1965}
J.~Alan Robinson.
\newblock A machine-oriented logic based on the resolution principle.
\newblock {\em JACM}, 12(1):23--41, 1965.
\newblock \href {https://doi.org/10.1145/321250.321253}
  {\path{doi:10.1145/321250.321253}}.

\bibitem{handbook:ar:2001}
J.~Alan Robinson and Andrei Voronkov, editors.
\newblock {\em Handbook of Automated Reasoning (in 2 volumes)}.
\newblock Elsevier, 2001.

\bibitem{schmidt:2012:ackermann}
Renate~A. Schmidt.
\newblock The {Ackermann} approach for modal logic, correspondence theory and
  second-order reduction.
\newblock {\em J. Applied Logic}, 10(1):52--74, 2012.
\newblock \href {https://doi.org/http://dx.doi.org/10.1016/j.jal.2012.01.001}
  {\path{doi:http://dx.doi.org/10.1016/j.jal.2012.01.001}}.

\bibitem{soqe:workshop:2021}
Renate~A. Schmidt, Christoph Wernhard, and Yizheng Zhao, editors.
\newblock {\em SOQE 2021}, volume 3009 of {\em CEUR Workshop Proc.}
  CEUR-WS.org, 2021.
\newblock URL: \url{https://ceur-ws.org/Vol-3009}.

\bibitem{eprover}
Stephan Schulz, Simon Cruanes, and Petar Vukmirovi{\'c}.
\newblock Faster, higher, stronger: {E} 2.3.
\newblock In Pascal Fontaine, editor, {\em CADE~27}, number 11716 in LNAI,
  pages 495--507. Springer, 2019.
\newblock \href {https://doi.org/10.1007/978-3-030-29436-6_29}
  {\path{doi:10.1007/978-3-030-29436-6_29}}.

\bibitem{schuette:1956}
Kurt Sch\"utte.
\newblock Ein {S}ystem des verknüpfenden {S}chliessens.
\newblock {\em Arch. math. Logik}, 2:55--67, 1956.
\newblock \href {https://doi.org/10.1007/BF01969991}
  {\path{doi:10.1007/BF01969991}}.

\bibitem{slagle:interpolation:1970}
James~R. Slagle.
\newblock Interpolation theorems for resolution in lower predicate calculus.
\newblock {\em JACM}, 14(3):535--542, 1970.
\newblock \href {https://doi.org/10.1145/321592.321604}
  {\path{doi:10.1145/321592.321604}}.

\bibitem{smullyan:book:1968}
Raymond~M. Smullyan.
\newblock {\em First-Order Logic}.
\newblock Springer, 1968.
\newblock Also republished with corrections by Dover publications, 1995.

\bibitem{pttp:1984}
Mark~E. Stickel.
\newblock A {P}rolog {T}echnology {T}heorem {P}rover.
\newblock {\em New Gener. Comput.}, 2(4):371--383, 1984.
\newblock \href {https://doi.org/10.1007/BF03037328}
  {\path{doi:10.1007/BF03037328}}.

\bibitem{tptp}
Geoff Sutcliffe.
\newblock {Stepping Stones in the {TPTP} World}.
\newblock In Christoph Benzm{\"{u}}ller, Marijn J.~H. Heule, and Renate~A.
  Schmidt, editors, {\em IJCAR 2024}, number 14739 in LNCS (LNAI), pages
  30--50. Springer, 2024.
\newblock \href {https://doi.org/10.1007/978-3-031-63498-7_3}
  {\path{doi:10.1007/978-3-031-63498-7_3}}.

\bibitem{dls:early}
Andrzej Sza{\l}as.
\newblock On the correspondence between modal and classical logic: {A}n
  automated approach.
\newblock {\em J. Logic and Computation}, 3:605--620, 1993.
\newblock \href {https://doi.org/10.1093/LOGCOM/3.6.605}
  {\path{doi:10.1093/LOGCOM/3.6.605}}.

\bibitem{takeuti:book:1987}
Gaisi Takeuti.
\newblock {\em Proof Theory}.
\newblock North-Holland, second edition, 1987.

\bibitem{troelstra:schwichtenberg:2000}
Arne~S. Troelstra and Helmut Schwichtenberg.
\newblock {\em Basic Proof Theory}.
\newblock Cambridge University Press, second edition, 2000.

\bibitem{vgt}
Allen {Van Gelder} and Rodney~W. Topor.
\newblock Safety and translation of relational calculus queries.
\newblock {\em ACM Trans. Database Syst.}, 16(2):235--278, 1991.
\newblock \href {https://doi.org/10.1145/114325.103712}
  {\path{doi:10.1145/114325.103712}}.

\bibitem{cw-skp}
Christoph Wernhard.
\newblock Semantic knowledge partitioning.
\newblock In Jos{\'e}~J{\'u}lio Alferes and Jo{\~a}o~Leite Leite, editors, {\em
  JELIA 04}, volume 3229 of {\em LNCS (LNAI)}, pages 552--564. Springer, 2004.
\newblock \href {https://doi.org/10.1007/978-3-540-30227-8_46}
  {\path{doi:10.1007/978-3-540-30227-8_46}}.

\bibitem{cw-relmon}
Christoph Wernhard.
\newblock Second-order quantifier elimination on relational monadic formulas --
  {A} basic method and some less expected applications.
\newblock In Hans de~Nivelle, editor, {\em TABLEAUX 2015}, volume 9323 of {\em
  LNCS (LNAI)}. Springer, 2015.
\newblock \href {https://doi.org/10.1007/978-3-319-24312-2_18}
  {\path{doi:10.1007/978-3-319-24312-2_18}}.

\bibitem{cw:pie:2016}
Christoph Wernhard.
\newblock The {PIE} system for proving, interpolating and eliminating.
\newblock In Pascal Fontaine, Stephan Schulz, and Josef Urban, editors, {\em
  PAAR~2016}, volume 1635 of {\em CEUR Workshop Proc.}, pages 125--138.
  CEUR-WS.org, 2016.
\newblock URL: \url{http://ceur-ws.org/Vol-1635/paper-11.pdf}.

\bibitem{cw:craig:report:2018}
Christoph Wernhard.
\newblock Craig interpolation and access interpolation with clausal first-order
  tableaux.
\newblock Technical Report Knowledge Representation and Reasoning 18-01,
  Technische Universit{\"a}t Dresden, 2018.
\newblock \href {https://doi.org/10.48550/arXiv.1802.04982}
  {\path{doi:10.48550/arXiv.1802.04982}}.

\bibitem{cw:pie:2020}
Christoph Wernhard.
\newblock Facets of the {PIE} environment for proving, interpolating and
  eliminating on the basis of first-order logic.
\newblock In Petra Hofstedt et~al., editors, {\em DECLARE 2019, Revised
  Selected Papers}, volume 12057 of {\em LNCS (LNAI)}, pages 160--177.
  Springer, 2020.
\newblock \href {https://doi.org/10.1007/978-3-030-46714-2_11}
  {\path{doi:10.1007/978-3-030-46714-2_11}}.

\bibitem{cw:ipol}
Christoph Wernhard.
\newblock Craig interpolation with clausal first-order tableaux.
\newblock {\em J. Autom. Reasoning}, 65(5):647--690, 2021.
\newblock \href {https://doi.org/10.1007/s10817-021-09590-3}
  {\path{doi:10.1007/s10817-021-09590-3}}.

\bibitem{cw:range:2023}
Christoph Wernhard.
\newblock Range-restricted and {H}orn interpolation through clausal tableaux.
\newblock In Revantha Ramanayake and Josef Urban, editors, {\em TABLEAUX 2023},
  volume 14278 of {\em LNCS (LNAI)}, pages 3--23. Springer, 2023.
\newblock \href {https://doi.org/10.1007/978-3-031-43513-3_1}
  {\path{doi:10.1007/978-3-031-43513-3_1}}.

\bibitem{cw:htinterpolation:2026}
Christoph Wernhard.
\newblock {C}raig-{L}yndon interpolation for the logic of here and there with a
  variation of {M}ints' sequent system.
\newblock In Stefan Hetzl, Jean~Christoph Jung, Renate~A. Schmidt, and
  Christoph Wernhard, editors, {\em CI-BD-SOQE 2026}, CEUR Workshop
  Proceedings, 2026.
\newblock Preprint: \url{https://arxiv.org/abs/2601.04080}.

\end{thebibliography}

\end{document}